\documentclass[12pt]{amsart}
\usepackage{graphicx}
\usepackage{hyperref}
\usepackage{graphicx}
\usepackage{epsfig}
\usepackage{pstricks}

\newtheorem{thm}{Theorem}[section]

\newtheorem{lemma}[thm]{Lemma}
\newtheorem{conj}[thm]{Conjecture}
\newtheorem{cor}[thm]{Corollary}
\newtheorem{remark}[thm]{Remark}
\newtheorem{example}[thm]{Example}

\usepackage{amsmath}
\usepackage{amsxtra}
\usepackage{amscd}
\usepackage{amsthm}
\usepackage{amsfonts}
\usepackage{amssymb}
\usepackage{eucal}
\newcommand{\bmb}{\left( \begin{array}{rr}}
\newcommand{\enm}{\end{array}\right)}
\newcommand{\cQ}{\mathcal Q}
\newcommand{\cT}{\mathcal T}

\newcommand{\C}{{\mathbb C}}
\newcommand{\Z}{{\mathbb Z}}

\newcommand{\bt}{{\mathbf t}}
\newcommand{\bu}{{\mathbf u}}
\newcommand{\bv}{{\mathbf v}}
\newcommand{\bw}{{\mathbf w}}

\newcommand{\bz}{{\mathbf z}}

\newcommand{\al}{{\alpha}}

\numberwithin{equation}{section}

\begin{document}

\title[Arctic curves of the 20V model on a triangle]{Arctic curves of the 20V model on a triangle}
\author{Philippe Di Francesco} 
\address{$\!\!\!\!\!\!\!\!\!\!\!\!$ Department of Mathematics, University of Illinois, Urbana, IL 61821, U.S.A. 
and 
\break Institut de Physique Th\'eorique, Universit\'e Paris Saclay, 
CEA, CNRS, F-91191 Gif-sur-Yvette, FRANCE\hfill
\break  e-mail: philippe@illinois.edu
}

\begin{abstract}
We apply the Tangent Method of Colomo and Sportiello to predict the arctic curves of the Twenty Vertex model with specific domain wall boundary conditions on a triangle, in the Disordered phase, leading to a phase diagram with six types of frozen phases and one liquid one. The result relies on a relation to the Six Vertex model with domain wall boundary conditions and suitable weights, as a consequence of integrability. We also perform the exact refined enumeration of configurations.
\end{abstract}

\maketitle
\date{\today}
\tableofcontents

\section{Introduction}

%previous work, 20V with various DWBC
%
%combinatorics, tilings 
%
%arctic phenomenon, tangent method cite various works

Two-dimensional integrable lattice models such as the Six Vertex (6V) model have a long history first rooted in the physics of spin systems with local interaction, for which exact bulk thermodynamic properties were derived \cite{Lieb}, including continuum descriptions via Coulomb Gas \cite{Nienhuis84} or Conformal Field Theory \cite{DSZ87}. More recently, these models also entered the realm of combinatorics (by considering domains with ``domain-wall" boundary conditions (DWBC)
\cite{DWBC}  and best illustrated by the correspondence between 6V-DWBC and Alternating Sign Matrices (ASM) \cite{kuperberg1996another}), probability theory
(by interpreting the configurations in terms of particle trajectories, see e.g. \cite{BCG16}), and algebraic geometry (by interpreting partition functions as K-theoretic characters of certain varieties, see e.g. \cite{GZJ}).

It was observed that in the presence of DWBC, the models behave quite differently and may display interesting scaling behavior, such as the arctic phenomenon. The latter is the emergence of sharp phase separations between ordered regions (crystal-like, near the boundaries of the domain) and disordered (liquid) ones. 

Such a phenomenon had been already observed in ``free fermion" tiling or dimer models, where typically tiles/dimers choose a preferred crystalline orientation near boundaries while they tend to be disordered away from them. This was  first observed in the uniform domino tilings of the Aztec diamond \cite{JPS}, giving rise to an arctic circle, and a general theory was developed for dimers \cite{KO2,KOS}. The free fermion character of these models can be visualized in their formulation in terms of non-intersecting lattice paths, i.e. families of paths with fixed ends, and sharing no vertex (i.e. avoiding each-other), and can consequently be expressed in terms of free lattice fermions. A manifestation of the free fermion character of these models is that their arctic curves are always analytic.

The 6V model in its disordered phase, while including a free fermion case, is generically a model of {\it interacting} fermions: it admits an ``osculating path" description, in which paths are non-intersecting, but are allowed to interact by ``kissing" i.e. sharing a vertex at which they bounce against each-other. 

The 6V model on a square domain with DWBC exhibits an arctic phenomenon in its disordered phase, which was predicted via non-rigorous methods\cite{CP2010,CNP}, the latest of which being the Tangent Method introduced by Colomo and Sportiello \cite{COSPO}. The new feature arising from these studies is that the arctic curves are generically {\it no longer analytic}, but rather {\it piecewise analytic}. For instance, the arctic curve for large Alternating Sign Matrices (uniformly weighted 6V-DWBC) is made of four pieces of different ellipses as predicted in \cite{CP2010} and later proved in \cite{Aggar}. 

The Tangent Method was validated in a number of cases, mostly in free fermion situations \cite{CPS,DFLAP,DFGUI,DFG2,DFG3,CorKeat}. Beyond free fermions and the case of the 6V-DWBC model (see also \cite{DF21V} for the case of U-turn reflective boundaries), the method was applied to another model of osculating paths: the Twenty Vertex (20V) model which is the triangular lattice version of the 6V model \cite{Kel,DG19,BDFG,DF20V,DF21V}.
In \cite{DG19}, four possible variations around DWBC were considered for the 20V model, denoted DWBC1,2,3,4.
It turns out that for DWBC1,2 (on a square domain) and DWBC3 (on a quadrangular domain) the (refined) enumeration of configurations can be achieved exactly in terms of the (possibly U-turn) 6V-DWBC model. Moreover in these two cases the total number of configurations matches the number of domino tilings of Aztec-like domains (see \cite{DG19}
and \cite{DF20V}) and there is an intriguing correspondence between arctic curves of both tilings and Vertex models.

In the present paper, we define new boundary conditions on a triangular domain for the 20V model, and we show that these give rise to an arctic phenomenon.
After performing an exact (refined) enumeration of the configurations,
we derive exact Tangent Method predictions for the (outer) arctic curves for arbitrary integrable Boltzmann weights of the disordered phase, by relating the model to the 6V-DWBC model, in the spirit of \cite{BDFG,DF21V}. 

The paper is organized as follows.
In Sect. \ref{defsec} we recall the definition of the 20V model, which is an ice-type model on the triangular lattice, and its integrable weights, and define a triangular domain and specific domain-wall boundary conditions for the model,
which we coin 20V${}_3$.
We also show simple transformations that allow to immediately obtain some other similar boundary conditions which we call 20V${}_1$ and 20V${}_2$.
In Sect. \ref{enumsec}, using integrability of the weights, we compute the fully inhomogeneous partition function of the 20V${}_3$ model in terms of the inhomogeneous partition function of the Six Vertex (6V) model, leading in particular to the (refined) enumeration of the configurations of the 20V${}_3$ model.
Sect.\ref{arcticsec} is devoted to the computation of arctic curves.  After recalling the principle of the Tangent Method, we give asymptotic estimates of the one-point function and path partition function determining the arctic curve: the main results obtained by applying the Tangent Method are Theorems \ref{NEbranchthm}, \ref{SEbranchthm}, and \ref{NWbranchthm}, respectively for three distinct portions of the 
outer arctic curve of the model, for arbitrary values of the integrable weights in the Disordered phase. We also briefly discuss the possible hidden structure for the inner part of the arctic curve, not predicted by our method, except in the free fermion case, where the arctic curve is expected to be analytic, i.e. all the branches are part of the same analytic curve.
In Section \ref{exsec} we present examples of arctic curves, successively in the uniformly weighted case for which the NE portion of arctic curve is part of an algebraic curve of degree 10, in the free-fermion case (where the Boltzmann weights are expressed in  terms of free-fermion 6V weights), and finally in the generic case. 
We gather a few concluding remarks in Sect. \ref{discsec}.

\noindent{\bf Acknowledgments.} We thank E. Guitter, R. Kenyon, and P. Ruelle for useful discussions. This work is partially supported by the Morris and Gertrude Fine endowment, and the NSF RTG Grant DMS19-37241.

\section{Definition of the model}\label{defsec}

\subsection{20V model and integrable weights} 

\begin{figure}
\begin{center}
\includegraphics[width=13cm]{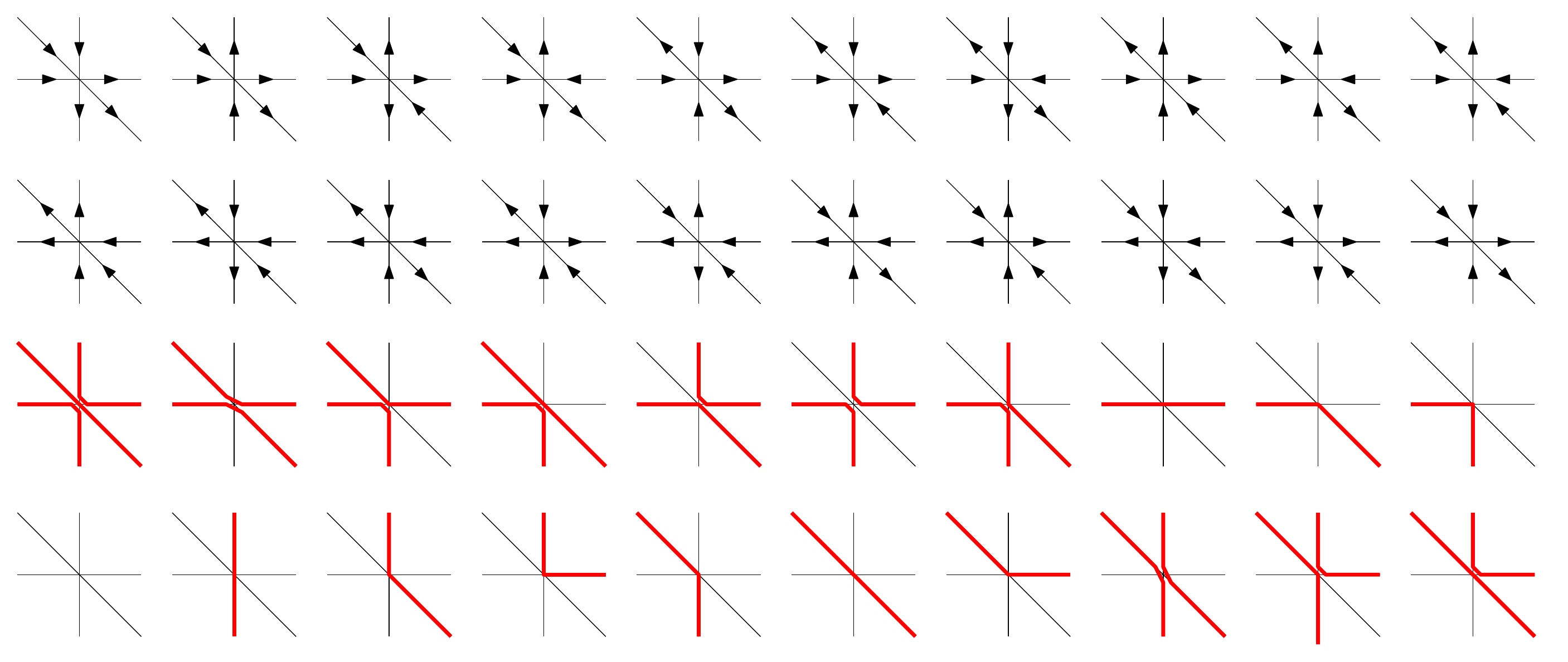}
\end{center}
\caption{\small Top two rows: the twenty vertex configurations subject to the ice rule. Bottom tweo rows: the bijection to osculating Schr\"oder paths.}
\label{fig:twentyV}
\end{figure}

\begin{figure}
\begin{center}
\includegraphics[width=9cm]{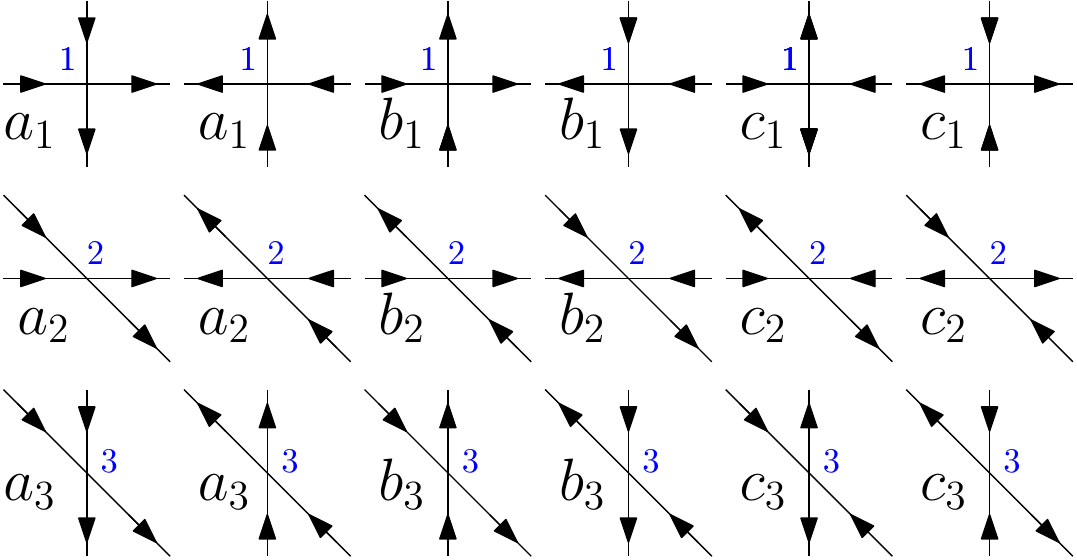}
\end{center}
\caption{\small The weights of the three 6V models on the three sub-lattices $1$, $2$ and $3$ of the Kagome lattice. }
\label{fig:KagomeWeights}
\end{figure}

The 20 Vertex (20V) model is the triangular lattice version of the square ``ice" model. Its configurations are choices of orientations of the edges of the triangular lattice, with the constraint that there are exactly three incoming and three outgoing edges adjacent to each vertex (``ice rule"). This gives rise to the $20={6 \choose 3}$ possible vertex environments depicted in Fig. \ref{fig:twentyV} (top two rows): here and in the following, we represent the triangular lattice with vertices in $\Z^2$ for convenience. A standard bijection allows to reformulate 20V configurations in terms of osculating Schr\"oder paths, namely paths on $\Z^2$ with horizontal, diagonal, vertical steps $(1,0)$, $(1,-1)$, $(0,-1)$, which are non-intersecting but are allowed to have contact (``kissing") points (see Fig. \ref{fig:twentyV}, two bottom rows). The vertices of the triangular lattice are at the intersection of three (horizontal, diagonal vertical)  lines. We may resolve these intersection by slightly moving up all diagonal lines. The resulting lattice is the Kagome lattice, with three
times as many vertices, but simple intersections (either horizontal-vertical, horizontal-diagonal or diagonal-vertical),
giving rise to three natural square sublattices which we label 1,2,3. 

\begin{figure}
\begin{center}
\includegraphics[width=13cm]{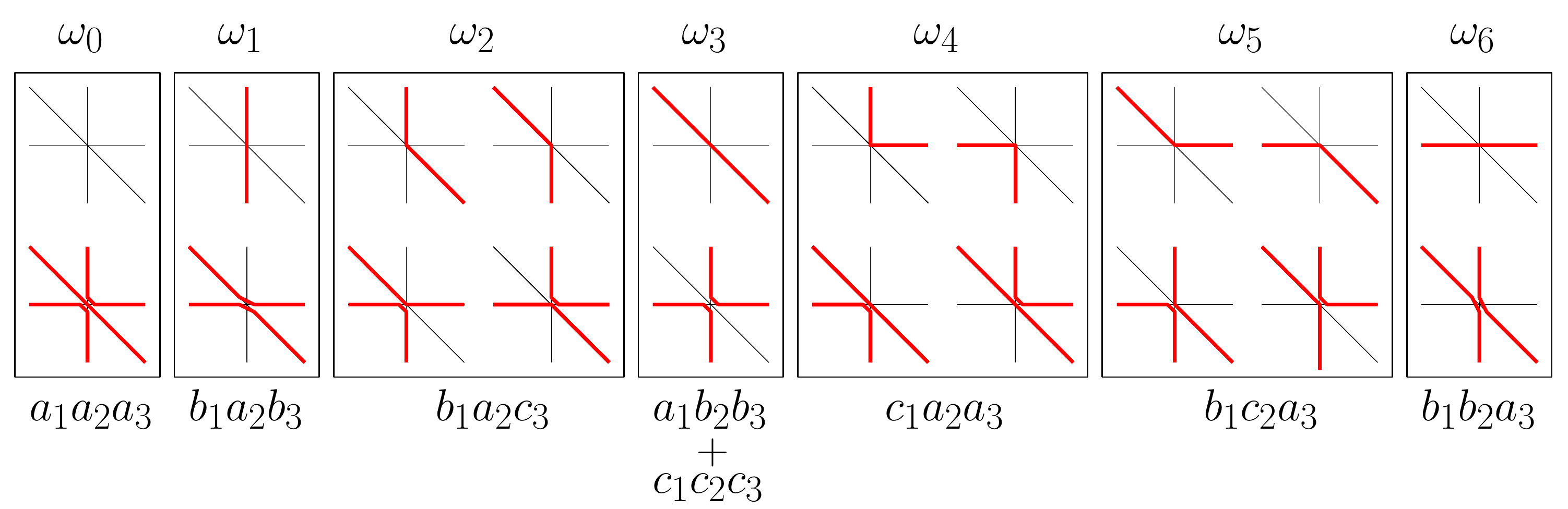}
\end{center}
\caption{\small The seven classes of vertices of the 20V model (in osculating Schr\"oder path formulation), and their corresponding weights $\omega_i$, $i=0,1,2,...,6$.}
\label{fig:weight20v}
\end{figure}

In \cite{Kel,Baxter}, Boltzmann weights for the 20V model are constructed in terms of weights of the Kagome lattice, themselves decomposing into three sets of 6V weights
for the vertices of type 1,2,3 respectively. Indeed as shown in \cite{Kel,Baxter,BDFG}, configurations of the 20V model may be obtained by considering oriented edges of the Kagome lattice satisfying the ice rule at all the vertices (of type 1,2,3). The correspondence 
is 1 to possibly 2. The weights of the 20V model are defined as the sum over the three inner arrow configurations that satisfy the three ice rules at the vertices of type 1,2,3, of the product of the three 6V Boltzmann weights of vertices 1,2,3.
For illustration, we have represented in Fig.~\ref{fig:KagomeWeights} the three types of 6V configurations together with their Boltzmann weights $(a_i,b_i,c_i)$ for i=1,2,3. Integrable weights for the 20-V model were obtained \cite{Kel,Baxter} by further
requiring that the expression for the weights is independent of the resolution of the triple intersections. For instance we could have moved down slightly all diagonal lines, giving rise to different relations for the 20V weights in terms of 6V weights. Solving these algebraic equations leads to a 4-parameter (projective) family of integrable Boltzmann weights.
The first condition is that the three 6V models must share the same quantum parameter $q$. Moreover, attaching spectral parameters $z,t,w$ to horizontal, diagonal, vertical lines respectively, the integrable 6V weights can be written as
\begin{equation} \label{3w6v}
\small{
\begin{matrix}
a_1=\al_1(z-w), \hfill &b_1=\al_1(q^{-2}z -q^2 w),\hfill &c_1=\al_1(q^2-q^{-2}) \sqrt{z w} \hfill\\
a_2=\al_2(q z-q^{-1}t), \hfill &b_2=\al_2(q^{-1}z -q t),\hfill &c_2=\al_2(q^2-q^{-2}) \sqrt{z t} \hfill\\
a_3=\al_3(q t-q^{-1}w), \hfill &b_3=\al_3(q^{-1}t -q w),\hfill &c_3=\al_3(q^2-q^{-2}) \sqrt{t w}\hfill
\end{matrix}}
\end{equation}
where $\alpha_i$ are constant normalization factors.
These finally lead to the following expressions for the 20V integrable weights $\omega_i\equiv \omega_i[z,t,w]$
displayed in Fig.~\ref{fig:weight20v}:
\begin{eqnarray}\label{weights20V}
\omega_0&=&\nu_0\, (z-w)(qz-q^{-1}t)(qt-q^{-1}w)\nonumber \\
\omega_1&=&\nu_0\, (q^{-2}z-q^2w)(qz-q^{-1} t)(q^{-1}t-qw) \nonumber \\
\omega_2&=&\nu_0\, (q^{-2}z-q^2w)(qz-q^{-1}t)(q^2-q^{-2})\sqrt{tw}\nonumber \\
\omega_3&=&\nu_0\,  z t w (q^2-q^{-2})^3+\nu_0\, (z-w)(q^{-1}z -q t)(q^{-1}t-q w) \nonumber \\
\omega_4&=&\nu_0\, (q^2-q^{-2})\sqrt{zw}(qz-q^{-1}t)(qt-q^{-1}w)\nonumber \\
\omega_5&=&\nu_0\, (q^{-2}z-q^2w)(q^2-q^{-2})\sqrt{zt}(qt-q^{-1}w)\nonumber \\
\omega_6&=&\nu_0\, (q^{-2}z-q^2w)(q^{-1}z-q t)(qt-q^{-1}w)
\end{eqnarray}
where $\nu_0=\al_1\al_2\al_3$.

By construction the above weights allow for freely moving around lines across intersections, a key property which we shall use extensively in the following to simplify the model.

With the following parametrization of the quantum and spectral parameters
\begin{equation}\label{param}
q={\rm e}^{{\rm i}\, \eta}\ , \quad z={\rm e}^{{\rm i}\,(\eta+\lambda)}\ , \quad w={\rm e}^{-{\rm i}\, (\eta+\lambda)}\ , \quad t={\rm e}^{{\rm i}\, \mu}\ , \qquad (\eta, \lambda, \mu \in \C),
\end{equation}
the Boltzmann weights of the three 6V models on the sublattices $1,2,3$ with horizontal, diagonal, vertical  parameters $z,t,w$ respectively read:
\begin{eqnarray}
&&a_1=\beta_1\sin(\lambda+\eta)\ , \quad b_1= \beta_1\sin(\lambda-\eta)\ , \quad c_1=\beta_1\sin(2 \eta)\nonumber \\
&&a_2=\beta_2\sin({\scriptstyle \frac{\lambda+3\eta-\mu}{2}})\ , \quad \! b_2= \beta_2\sin({\scriptstyle \frac{\lambda-\eta-\mu}{2}})\ , \quad c_2=\beta_2\sin(2 \eta)\nonumber \\
&&a_3=\beta_3\sin({\scriptstyle\frac{\lambda+3\eta+\mu}{2}})\ , \quad \! b_3= \beta_3\sin({\scriptstyle \frac{\lambda-\eta+\mu}{2}})\ , \quad c_3=\beta_3\sin(2 \eta) .
\label{eq:abc}
\end{eqnarray}
where $\beta_1=\sqrt{z w}\,\al_1$, $\beta_2=\sqrt{z t}\,\al_2$ and $\beta_3=\sqrt{t w}\,\al_3$ are assumed to be positive numbers.

The corresponding 20V model weights read \cite{DFGUI,Kel}:
\begin{eqnarray}
\omega_0&=&\nu\,\sin(\lambda+\eta)\sin\left({\scriptstyle \frac{\lambda+3\eta-\mu}{2}}\right)\sin\left({\scriptstyle \frac{\lambda+3\eta+\mu}{2}}\right)\nonumber \\
\omega_1&=&\nu\,\sin(\lambda-\eta)\sin\left({\scriptstyle \frac{\lambda+3\eta-\mu}{2}}\right)\sin\left({\scriptstyle \frac{\lambda-\eta+\mu}{2}}\right)\nonumber \\
\omega_2&=&\nu\,\sin(2\eta)\sin(\lambda-\eta)\sin\left({\scriptstyle \frac{\lambda+3\eta-\mu}{2}}\right)\nonumber \\
\omega_3&=&\nu\,\sin(2\eta)^3+\nu\sin(\lambda+\eta)\sin\left({\scriptstyle \frac{\lambda-\eta+\mu}{2}}\right)\sin\left({\scriptstyle \frac{\lambda-\eta-\mu}{2}}\right)\nonumber \\
\omega_4&=&\nu\,\sin(2\eta)\sin\left({\scriptstyle \frac{\lambda+3\eta+\mu}{2}}\right)\sin\left({\scriptstyle \frac{\lambda+3\eta-\mu}{2}}\right)\nonumber \\
\omega_5&=&\nu\,\sin(2\eta)\sin(\lambda-\eta)\sin\left({\scriptstyle \frac{\lambda+3\eta+\mu}{2}}\right)\nonumber \\
\omega_6&=&\nu\,\sin(\lambda-\eta)\sin\left({\scriptstyle \frac{\lambda+3\eta+\mu}{2}}\right)\sin\left({\scriptstyle \frac{\lambda-\eta-\mu}{2}}\right)\ ,
\label{eq:explweights}
\end{eqnarray}
with $\nu=\beta_1\beta_2\beta_3$.
We will finally restrict to the so-called Disordered Phase of the model, corresponding to real values of the parameters
$\eta$, $\lambda$ and $\mu$ and to a range of these parameters ensuring that all $\omega$'s are positive:
$$ 0<\eta<\lambda,\quad \eta-\lambda<\mu<\lambda-\eta,\quad \lambda+\eta<\pi, \quad \eta<\frac{\pi}{2} $$
Note as a consequence that $\lambda+3\eta>\mu$. 

Note also the existence of a ``combinatorial point" where the weights $\omega_i$ are uniform and all equal to $1$: 
\begin{equation}\label{combipoint20v}\eta=\frac{\pi}{8},\qquad \lambda=5\eta=\frac{5\pi}{8},\qquad \mu=0, \qquad  \nu=\beta_1\beta_2\beta_3=\sqrt{2} . \end{equation}

\subsection{The 20V-DWBC3 model on the triangle $\cT_m$}

\begin{figure}
\begin{center}
\includegraphics[width=12cm]{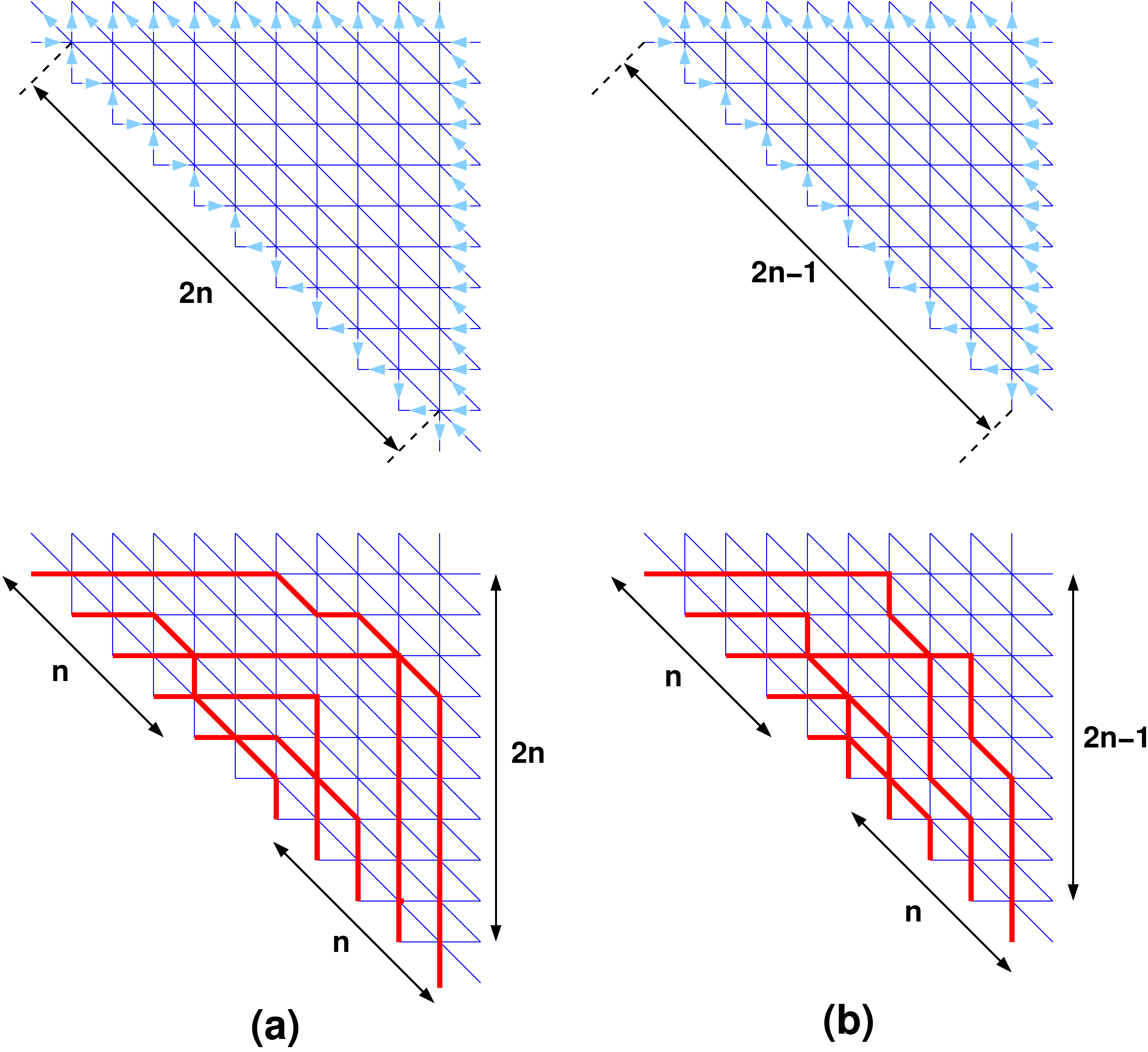}
\end{center}
\caption{\small The 20V model on the triangle $\cT_m$  with DWBC3 for even size $m=2n$ (a) and odd size $m=2n-1$ (b). 
The second row displays a typical configuration for each case, in the osculating Schr\"oder path formulation. }
\label{fig:triangledwbc3}
\end{figure}

We consider the partition function $Z_m^{20V_3}[\bz,\bt,\bw]$ of the fully inhomogeneous 20V model on a triangular domain $\cT_m$ of size $m$ with Domain-Wall type boundary conditions as in Fig. \ref{fig:triangledwbc3}. 
In terms of edge orientations, all edges on the North boundary point out of the domain, all edges on the east boundary point in, and the diagonal boundary edges point in for the top half and out for the bottom half.
The cases of even and odd size are different and displayed in Fig. \ref{fig:triangledwbc3} (a) and (b) respectively: in the odd case, the center vertex of the diagonal boundary has one ingoing and one outgoing external edge. 
The horizontal, diagonal, vertical lines respectively carry spectral parameters
$\bz=z_1,z_2,...,z_m$ (from top to bottom), $\bt=t_1,t_2,...,t_m$ (from top to bottom) and $\bw=w_1,w_2,...,w_m$ (from right to left): each vertex $v$ is weighted by the weight $\omega(v)$ of 
\eqref{eq:explweights} corresponding 
to its local configuration, with parameters $\lambda,\mu$ corresponding to the three (horizontal, diagonal, vertical) lines meeting at $v$, while the partition function is the sum over the 20V configurations
of the product of their local vertex weights.

From a purely enumerative point of view, if we set all parameters to the combinatorial point values \eqref{combipoint20v}, the partition function $Z_m^{20V_3}[\bz,\bt,\bw]$
reduces to the number $Z_m^{20V_3}$ of configurations of the $20V_3$ model on $\cT_m$. As we shall see in Sect. \ref{refinedsec} below,  $Z_m^{20V_3}[\bz,\bt,\bw]$ also gives 
access to {\it refined numbers} $Z_{m,k}^{20V_3}$ of configurations of the $20V_3$ model on $\cT_m$, conditioned so that the topmost vertex visited by a path in the rightmost vertical line (with spectral parameter $w_1$) is at position $k\in [1,m]$ counted from bottom to top. 

\subsection{Transformations of the 20V model}\label{transfosec}

In this section, we use transformations of the 20V model previously exploited in Refs. \cite{BDFG,DF21V} to map the $20V_3$ configurations on $\cT_m$ 
onto configurations of the 20V model on the same domain but with {\it different} boundary conditions. For convenience, we place the origin at the bottom right vertex of $\cT_m$, so that the extremal vertices of $\cT_m$ have coordinates 
$(0,0),(1-m,m-1),(0,m-1)$.

\begin{figure}
\begin{center}
\includegraphics[width=16cm]{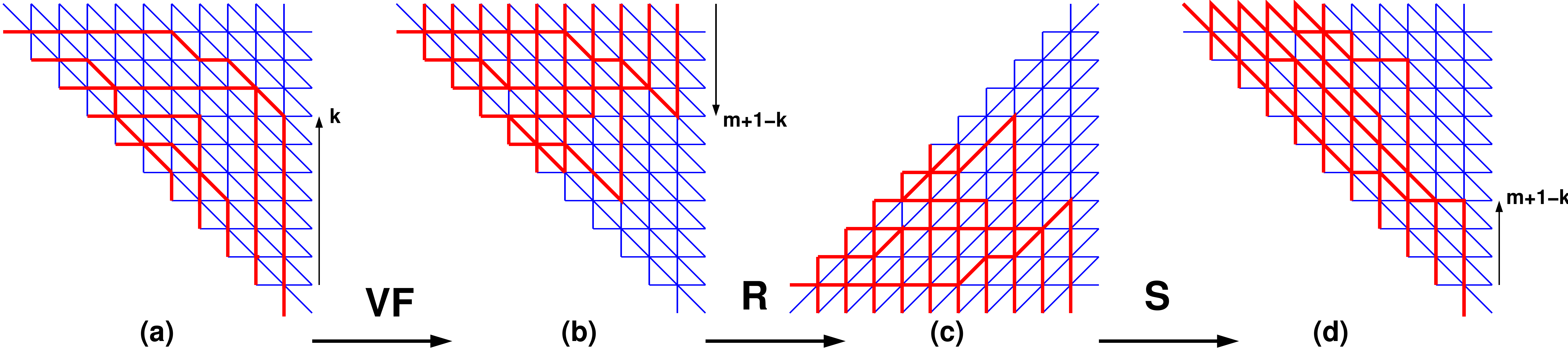}
\end{center}
\caption{\small The bijection between 20V${}_3$ (a) and 20V${}_{2}$ (d) configurations decomposes into three steps: (a)$\to$ (b) Vertical Flip; (b)$\to$(c) Reflection;(c)$\to$(d): Shear. We have indicated the correspondence between the positions $k$ and $m+1-k$ of the vertex first visited by the topmost path in the rightmost column. }
\label{fig:shear}
\end{figure}

%The configurations of the 6V${}_1$-DWBC model with weights $(a,b,c)=(1,\sqrt{2},1)$ are mapped onto those of the 
%6V${}_3$-DWBC model with weights $(a,b,c)=(\sqrt{2},1,1)$ by performing a reflection w.r.t. a horizontal line. In terms of osculating paths, we must also 
%erase every vertical step of path and promote every empty vertical edge to a path step. Comparing the situations in the last column, we see that the
%topmost position $k$ path edges for the 6V${}_1$ configuration corresponds to $n+1-k$ in the reflected the 6V${}_3$ configuration, which implies:
%$$ Z_{n,k}^{6V_1}=Z_{n,n+1-k}^{6V_3} \ \ \Rightarrow \ \ Z_n^{6V_1}(x)=x^{n-1} \, Z_n^{6V_3}(x^{-1}) .$$
%Combining this with the formulas of Theorem \ref{ref20v6vthm}, we deduce that 
%\begin{equation}\label{recip}
%Z_{m}^{20V_{1-3}}(\tau)=\tau^{m-1}Z_{m}^{20V_3}(\tau^{-1}) .
%\end{equation}
%Let us give a direct combinatorial proof of this reciprocality property.

%\begin{thm}
%There is a simple bijection between the configurations on $\cT_m$ of the 20V${}_3$ model and those of the 20V${}_{1-3}$ model, which makes the reciprocality property \eqref{recip} manifest.
%\end{thm}
%\begin{proof}
We introduce the following sequence of transformations of the path configurations of the 20V${}_3$ model, as illustrated in Fig. \ref{fig:shear}. Assume that $k$ is the position of the topmost vertex visited by a path in the rightmost vertical line.
\begin{itemize}
\item{Vertical Flip {\bf VF}:} we ``flip" all vertical edges, i.e. each vertical path edge is erased, and each empty vertical edge receives a path edge.
In particular $k$ becomes the bottom-most position of a vertex visited by a path in the rightmost vertical line.
\item{Reflection {\bf R}:} we reflect the picture w.r.t. a horizontal line, thus sending diagonal lines to anti-diagonal ones. Now $m+1-k$ is the new 
position of the topmost vertex belonging to a vertical step of path in the rightmost vertical line.
\item{Shear {\bf S}:} we apply a shear transformation in order to recover the original domain $\cT_m$, namely the map $(x,y)\mapsto (x,y-x)$. Under this transformation, anti-diagonal lines become horizontal, horizontal lines become diagonal, while $m+1-k$ remains the new position of the topmost vertex visited by a  path in the rightmost vertical line.
\end{itemize}
By inspection we find that the seven classes of 20V local vertex environments depicted in Fig. \ref{fig:weight20v} are mapped bijectively under $\bf S\circ R\circ VF$ as follows:
\begin{equation}\label{netose}
{\bf S\circ R\circ VF}:\ \  (\omega_0,\omega_1,\omega_2,\omega_3,\omega_4,\omega_6)\mapsto (\omega_1,\omega_0,\omega_4,\omega_3,\omega_2,\omega_5,\omega_6) . 
\end{equation}
For further use, we denote by $\pi$ the permutation $\pi=(01)(24)$ corresponding to this mapping of configurations.

\begin{figure}
\begin{center}
\includegraphics[width=12cm]{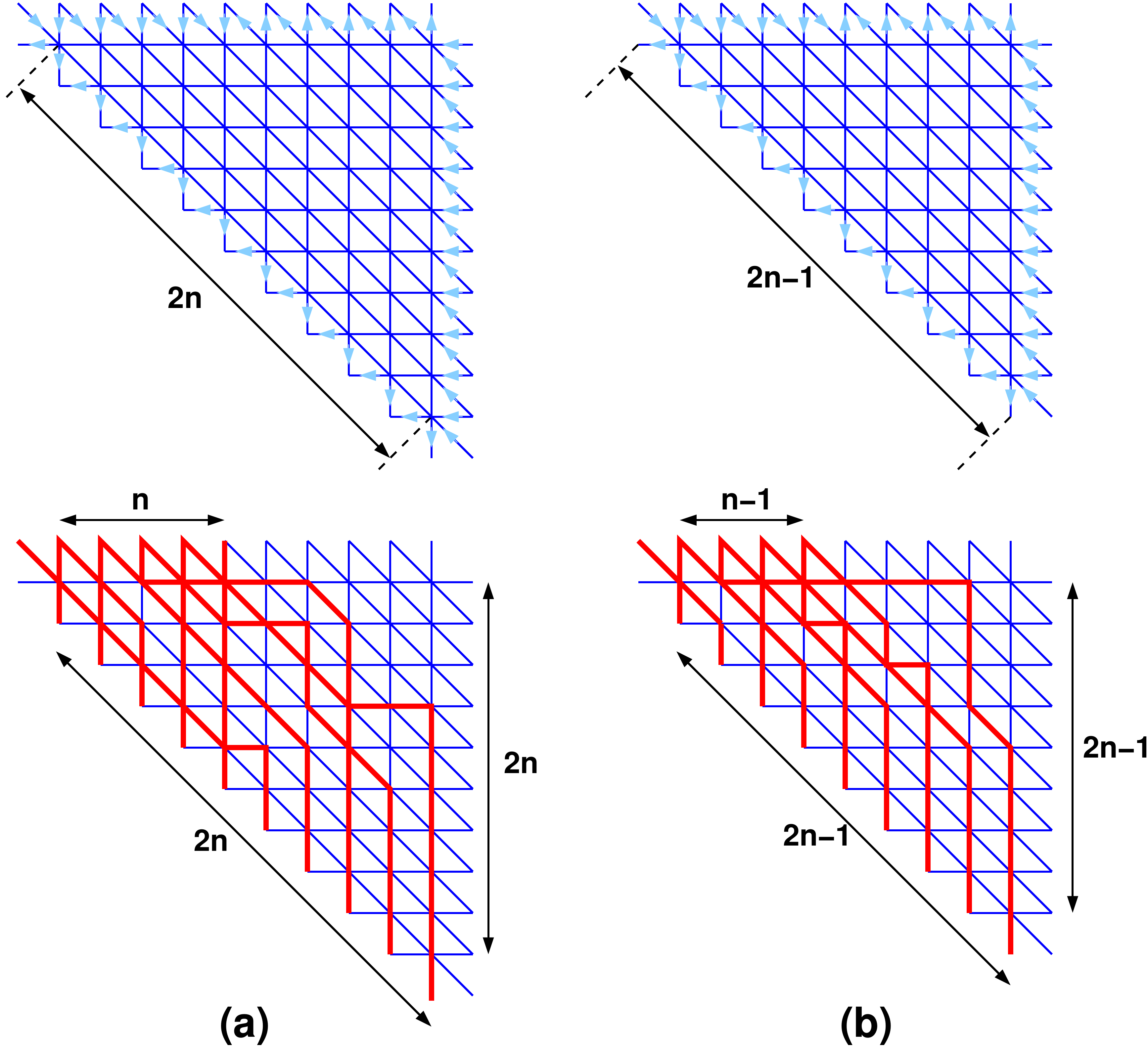}
\end{center}
\caption{\small The 20V${}_2$ model on  the triangle $\cT_m$ for even size $m=2n$ (a) and odd size $m=2n-1$ (b). 
The first row indicates the new boundary conditions in the original formulation.
The second row displays a typical configuration for each case, in the osculating Schr\"oder path formulation. }
\label{fig:triangledwbcm}
\end{figure}

Examining the effect on the boundary conditions, 
it is easy to see that 20V${}_3$ configurations are mapped to 20V configurations on the same triangle $\cT_m$  but with new boundary conditions depicted in Fig. \ref{fig:triangledwbcm} (a-b) for odd and even size $m$: we call this model 20V${}_2$.
Each step being invertible, $\bf S\circ R\circ VF$ is clearly a bijection. 

Denoting by $Z_m^{20V_2}[\bz,\bt,\bw]$ the partition function for the inhomogeneous 20V${}_2$ model, with the same labeling of spectral parameters as for the 20V${}_3$
model, and following the horizontal, diagonal and vertical lines throughout the transformations, we deduce the following relation.

\begin{thm}
The partition functions of the fully inhomogeneous 20V${}_3$ and 20V${}_2$ models are related via:
$$ Z_m^{20V_3}[{\scriptstyle z_1,z_2,...,z_m,t_1,t_2,...,t_m,w_1,w_2,...,w_m}]= Z_m^{20V_2}[{\scriptstyle t_m,t_{m-1},...,t_1,z_m,z_{m-1},...,z_1,w_1,w_2,...,w_m}]\Big\vert_{\omega_i\to\omega_{\pi(i)}} .$$
\end{thm}

Indeed, the mapping of 20V weights under the transformation $\bf S\circ R\circ VF$ amounts to the simple transformation $\omega_i[z,t,w]\mapsto \omega_{\pi(i)}[t,z,w]$,
as diagonal and horizontal lines are interchanged in the process.

From a purely enumerative point of view, it is clear from Fig. \ref{fig:shear} that $\bf S\circ R\circ VF$ maps (refined) configurations of the 20V${}_3$ model to those
of the 20V${}_2$ model (up to $k\to m+1-k$).
\begin{cor}\label{23cor}
The refined numbers of configurations in the 20V${}_3$ and 20V${}_2$ models are related via:
$$Z_{m,k}^{20V_3}=Z_{m,m+1-k}^{20V_{2}} \qquad (k=1,2,...,m). $$
\end{cor}

We may repeat the above with instead of $\bf VF$ a horizontal flip ($\bf HF$), namely erasing all horizontal path edges, 
and promoting all empty horizontal edges to new path edges. The reflection $\bf R$ is now replaced with a reflection $\bf \bar R$ w.r.t. a vertical line and the shear $\bf S$
with a horizontal shear $\bf \bar S$: $(x,y)\mapsto (m+1+x-y,y)$. The net result is the same as applying an extra diagonal reflection $\bf R^*$ w.r.t. a direction perpendicular to that of the diagonal lines
with parameters $t_i$, after the previous transformation $\bf S\circ R\circ VF$, resulting in the relation $\bf \bar S\circ \bar R\circ HF=\bf R^*\circ \bf S\circ R\circ VF$.
We denote by 20V${}_1$ the corresponding model. Its configurations and boundary conditions are simply obtained from those of the 20V${}_2$ model by applying the
diagonal reflection $\bf R^*$, under which the seven weight classes are  mapped as follows:
\begin{equation}\label{setonw}
{\bf  R^*}:\ \  (\omega_0,\omega_1,\omega_2,\omega_3,\omega_4,\omega_6) \mapsto  (\omega_0,\omega_6,\omega_5,\omega_3,\omega_2,\omega_1) .
\end{equation}
Denoting by $\bar \pi$ the corresponding permutation  of labels $\bar \pi=(16)(25)$, the mapping of weights is $\omega_i[z,t,w]\mapsto \omega_{\bar \pi(i)}[w,t,z]$
as horizontal and vertical lines are interchanged. This give the following relation between fully inhomogeneous partition functions:
$$Z_m^{20V_2}[{\scriptstyle z_1,z_2,...,z_m,t_1,t_2,...,t_m,w_1,w_2,...,w_m}]= Z_m^{20V_1}[{\scriptstyle w_1,w_2,...,w_m,t_1,t_2,...,t_m,z_1,z_2,...,z_m}]\Big\vert_{\omega_i\to\omega_{\bar\pi(i)}} .$$
Examples of 20V${}_1$ configurations are easily obtained from those of Fig. \ref{fig:triangledwbcm} by applying the diagonal reflection $\bf R^*$. Finally, denoting by $Z_{m,k}^{20V_{1}}$ the refined number of 20V${}_1$ configurations such that the topmost path leaves the topmost horizontal line at position $k$ counted from the left, then we have
the following.
\begin{cor}\label{12cor}
The refined numbers of configurations in the 20V${}_2$ and 20V${}_1$ models are related via:
$$Z_{m,k}^{20V_1}=Z_{m,k}^{20V_{2}} \qquad (k=1,2,...,m). $$
\end{cor}

\section{Enumerative results}\label{enumsec}
This section is devoted to exact results on the enumeration of the configurations of various 20V models on the triangle $\cT_m$. This includes the refined enumeration of 
configurations according to certain statistics.

\subsection{Inhomogeneous case: Relation between 20V and 6V partition functions}

\begin{figure}
\begin{center}
\includegraphics[width=16cm]{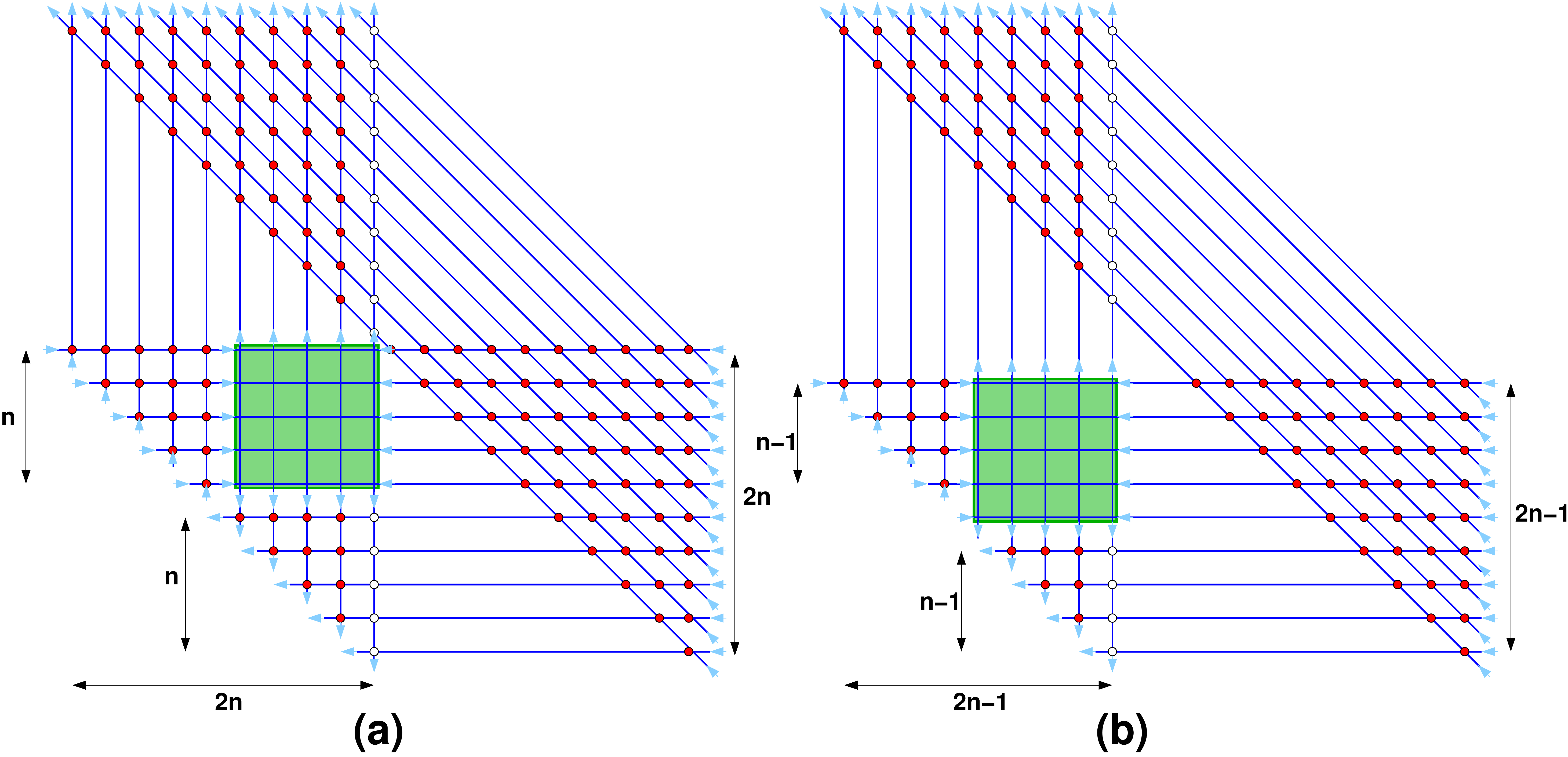}
\end{center}
\caption{\small Transformation of the partition functions $Z_{2n}^{20V_3}[\bz,\bt,\bw]$ (left) and $Z_{2n-1}^{20V_3}[\bz,\bt,\bw]$ (right) by moving up the diagonal lines.
Using the ice rule, one sees that all edge configurations are determined at all the vertices marked by dots, outside of the shaded squares, 
both of which correspond in turn to the partition function $Z_n^{6V_1}[\bz,\bw]$ of the 6V-DWBC model
with the suitable spectral parameters and the Boltzmann weights of the sublattice $1$.}
\label{fig:unravel}
\end{figure}

The partition functions $Z_{m}^{20V_3}[\bz,\bt,\bw]$ and $Z_{m}^{20V_{1,2}}[\bz,\bt,\bw]$ may be expressed simply in terms of the partition function $Z_n^{6V}[\bz,\bw]$ of the 6V model
on a square of size $n=\lfloor \frac{m+1}{2} \rfloor$ with DWBC. As already noted in \cite{DFG20V,BDFG,DF20V}, the integrability of the 20V weights allows to freely move around the spectral lines across intersections 
and to transform the model. 

In the $20V_3$ case, moving the diagonal lines up results in the transformations depicted in Fig. \ref{fig:unravel} for the even and odd size cases respectively.
Once the diagonal lines are moved up, the ice rule allows to propagate the  orientations of the boundary edges to all the edges of the diagonal lines (all oriented upward) and those of the 
horizontal and vertical lines in their domain of intersection. The vertex at the intersection of the $j$-th diagonal line and $i$-th horizontal one is in the configuration $a_2$ and receives the weight 
$a_2(z_i,t_j)$, while each vertex at the intersection of the $j$-th diagonal line and $k$-th vertical one is in the configuration $a_3$ and receives the weight $a_3(t_j,w_k)$.
Finally the vertices at the intersections between the $i$-th horizontal and $j$-th vertical lines not inside the marked squares are all in the configuration $b_1$ and receive weights $b_1(z_i,w_j)$.
The remaining vertices (inside the marked squares) form the $n\times n$ square grid of a 6V model with the weights of the sublattice $1$, 
and with boundary edges oriented horizontally inward and vertically outward, that is exactly
the Domain Wall boundary conditions. We denote by $Z_n^{6V_1}[\bz,\bw]$ the corresponding partition function.
Collecting all weights, the partition functions of the 20V model on the triangle of even and odd sizes read:
\begin{eqnarray}
Z_{2n}^{20V_3}[\bz,\bt,\bw]\!\!&=&\!\!\!\!\! \prod_{1\leq i\leq j\leq 2n} a_2(z_i,t_j) \!\!\prod_{1\leq j\leq 2n+1-k\leq 2n} a_3(t_j,w_k) 
\!\!\prod_{1\leq i\leq 2n+1-j\leq n\atop {{\rm or}\atop n+1\leq i\leq 2n+1-j\leq 2n}} b_1(z_i,w_j) \ Z_n^{6V_1}[\bz,\bw]\nonumber\\
Z_{2n-1}^{20V_3}[\bz,\bt,\bw]&=&\!\!\!\!\!\!\!\! \prod_{1\leq i\leq j\leq 2n-1} a_2(z_i,t_j) \!\!\!\!\prod_{1\leq j\leq 2n-k\leq 2n-1} a_3(t_j,w_k) 
\!\!\prod_{1\leq i\leq 2n-j\leq n-1\atop {{\rm or}\atop n+1\leq i\leq 2n-j\leq 2n-1}} b_1(z_i,w_j) \ Z_n^{6V_1}[\bz,\bw] .\nonumber\\
\label{inhomrela}
\end{eqnarray}

%\begin{figure}
%\begin{center}
%\includegraphics[width=16cm]{unravelm}
%\end{center}
%\caption{\small Transformation of the partition functions $Z_{2n}^{20V_{2}}[\bz,\bt,\bw]$ (left) and $Z_{2n-1}^{20V_{2}}[\bz,\bt,\bw]$ (right) by moving down the horizontal lines.
%Using the ice rule, one sees that all edge configurations are determined at all the vertices marked by dots, outside of the shaded lozenges, 
%both of which correspond in turn to the partition function $Z_n^{6V_3}[\bz,\bw]$ of the 6V-DWBC model
%with the suitable spectral parameters and the Boltzmann weights of the sublattice $3$.}
%\label{fig:unravelm}
%\end{figure}

%In the $20V_{2}$ case, moving the horizontal lines down leads to a similar situation depicted in Figs. \ref{fig:unravelm} (a) or (b) according to the parity of $m$.
%Note that the remaining $n\times n$ lozenge grid 6V-DWBC partition function has the weights of the sublattice $3$.
%This leads to the following relations:
%\begin{eqnarray}
%Z_{2n}^{20V_{1-3}}[\bz,\bt,\bw]\!\!&=&\!\!\!\!\!{\prod_{1\leq i\leq j\leq 2n} a_2(z_i,t_j) \!\!\!\!\prod_{1\leq j\leq 2n+1-k\leq 2n} b_1(t_j,w_k) 
%\!\!\!\!\!\!\prod_{1\leq i\leq 2n+1-j\leq n\atop {{\rm or}\atop n+1\leq i\leq 2n+1-j\leq 2n}} a_3(z_i,w_j)} \ Z_n^{6V_3}[\bt,\bw]\nonumber\\
%Z_{2n-1}^{20V_{1-3}}[\bz,\bt,\bw]&=&\!\!\!\!\!\!\!\!{\prod_{1\leq i\leq j\leq 2n-1} a_2(z_i,t_j) \!\!\!\!\prod_{1\leq j\leq 2n-k\leq 2n-1} b_1(t_j,w_k) 
%\!\!\!\!\!\!\prod_{1\leq i\leq 2n-j\leq n-1\atop {{\rm or}\atop n+1\leq i\leq 2n-j\leq 2n-1}} a_3(z_i,w_j)} \ Z_n^{6V_3}[\bt,\bw]\nonumber\\
%\label{inhomrela13}
%\end{eqnarray}

\subsection{Combinatorial point}

Setting all spectral parameters in \eqref{inhomrela} to the uniform values of the combinatorial point \eqref{combipoint20v}, for which 
$(a_1,b_1,c_1)=\frac{\beta_1}{\sqrt{2}}(1,\sqrt{2},1)$ while
$(a_2,b_2,c_2)=\frac{\beta_2}{\sqrt{2}}(\sqrt{2},1,1)$ and $(a_3,b_3,c_3)=\frac{\beta_3}{\sqrt{2}}(\sqrt{2},1,1)$, we deduce the following.

\begin{thm}
The number of configurations $Z_m^{20V_3}$ of the 20V${}_3$ model on the triangle $\cT_m$ is expressed in terms of the partition function 
$Z_n^{6V}$ of the 6V-DWBC model on a square grid of size 
$n=\lfloor \frac{m+1}{2}\rfloor$ with uniform weights $(a,b,c)=(1,\sqrt{2},1)$ as:
$$ Z_{2n}^{20V_3} = 2^{n(n+1)/2} \, Z_n^{6V} ,\qquad Z_{2n-1}^{20V_3}  = 2^{n(n-1)/2} \, Z_n^{6V}  .$$
\end{thm}
\begin{proof}
In the even case of the $20V_3$ model, we have from the inhomogeneous relation \eqref{inhomrela}:
$$ Z_{2n}^{20V_3}=(a_2a_3)^{n(2n+1)} b_1^{n(n+1)}\!\! \left(\frac{\beta_1}{\sqrt{2}}\right)^{n^2}  Z_n^{6V}=\left(\frac{\beta_1\beta_2\beta_3}{\sqrt{2}}\right)^{n(2n+1)}\!\! 2^{n(n+1)/2} Z_n^{6V}=2^{n(n+1)/2}  \, Z_n^{6V} ,$$
as $(a_1,b_1,c_1)=\frac{\beta_1}{\sqrt{2}}(1,\sqrt{2},1)$, $a_2=\beta_2$, $a_3=\beta_3$, and $\nu=\beta_1\beta_2\beta_3=\sqrt{2}$. The relation for odd $m$ follows similarly:
$$ Z_{2n-1}^{20V_3}=(a_2a_3)^{n(2n-1)} b_1^{n(n-1)}\!\! \left(\frac{\beta_1}{\sqrt{2}}\right)^{n^2} Z_n^{6V}=\left(\frac{\beta_1\beta_2\beta_3}{\sqrt{2}}\right)^{n(2n-1)}\!\! 2^{n(n-1)/2}  Z_n^{6V}=2^{n(n-1)/2}  \, Z_n^{6V} .$$
%Repeating this argument for the $20V_{1-3}$ model, we must now use $(a_3,b_3,c_3)=\frac{\beta_3}{\sqrt{2}}(\sqrt{2},1,1)$ and the fact that the partition function of the 6V-DWBC model for homogeneous weights $(a,b,c)$ is invariant under the interchange $a\leftrightarrow b$. The latter fact is best understood by performing a reflection of the 6V model configurations w.r.t. a horizontal line, under which configurations of type $c$ are invariant, and configurations of types a and b are interchanged.
\end{proof}

A direct consequence of this theorem is the following:
\begin{cor} The even and odd partition functions of the 20V${}_3$ model are related via
$$ Z_{2n}^{20V_3}=2^{n}\, Z_{2n-1}^{20V_3} .$$
\end{cor}

In \cite{DFG20V} it was shown that the integer $Z_n^{6V}$ counts the number of Quarter-Turn symmetric 
domino tilings of the Holey Aztec Square of size $2n$.
A compact formula for this number is \cite{DFG20V}:
$$Z_n^{6V}=\det\limits_{0\leq i,j\leq n-1}\left.\left(\frac{1}{1-x y} +\frac{2x}{(1-x)(1-x-y-xy)}\right)\right\vert_{x^i y^j} ,
$$
where $f(x,y)\vert_{x^i y^j}$ stands for the coefficient of $x^iy^j$ in the series expansion of $f$ around $(0,0)$. 
This leads to the following sequence for $m=1,2,...,10, ...$:
$$Z_{m}^{20V_3}=1,2,6, 24, 184, 1472, 27712, 443392, 20177920, 645693440, ... $$

\subsection{Partially inhomogeneous case: refined enumeration}\label{refinedsec}

We now use Eq.~\eqref{inhomrela} with all spectral parameters equal to their combinatorial point values
\eqref{combipoint20v}, except for 
$w_1=e^{-i(\lambda+\eta+2\xi)}$ with $\lambda,\eta$ as in \eqref{combipoint20v}
and $\xi$ arbitrary (and for which the normalization factors in \eqref{3w6v} and \eqref{weights20V} become $\beta_1[\xi]=e^{-i\xi}\beta_1$, $\beta_2[\xi]=\beta_2$ and $\beta_3[\xi]=e^{-i\xi}\beta_3$, and $\nu_0[\xi]=e^{-2i\xi}\nu_0$ respectively). 
The corresponding rightmost vertical line is marked with empty dots in Figs. \ref{fig:unravel}. 
%and \ref{fig:unravelm}. 
Let $Z_{m}^{20V_3}[\xi]$, denote the corresponding partition function. All local vertex weights $\omega(v)$ are $1$ except
in the last column (corresponding to the spectral parameter $w_1$), where they read:
\begin{eqnarray}
\omega_0[\xi]&=&\sqrt{2}\,\cos(\xi)\cos(\xi+\frac{\pi}{4}), \quad \omega_1[\xi]=\sqrt{2}\,\cos(\xi)\sin(\xi+\frac{\pi}{4})\nonumber \\
\omega_2[\xi]&=&\omega_4[\xi]=\cos(\xi),\qquad \qquad 
\omega_3[\xi]=\omega_5[\xi]=\omega_6[\xi]=\cos^2(\xi)\ .
\label{inhoweights}
\end{eqnarray}
%The following applies to both $20V_3$ and $20V_{1-3}$ models, and we drop the subscript for convenience.
Recall the refined number $Z_{m,k}^{20V_3}$ of 20V${}_3$ configurations on $\cT_m$ with prescribed position $k\in \{1,2,...,m\}$ of the vertex where the top path hits the last column for the first time.
The configurations contributing to $Z_{m,k}^{20V_3}$ in $Z_{m}^{20V_3}[\xi]$ have trivial vertex weights except in the last column, whose total weight is $\omega_1[\xi]^{k-1} \omega_2[\xi]\omega_0[\xi]^{m-k}$. 
Indeed, in the last column, the vertex at position $k$ may either
be in the configuration $\omega_2$ or $\omega_4$ but the two weights are the same, while the $m-k$ top vertices are in the configuration $\omega_0$ (empty) and the $k-1$ bottom ones are in configuration $\omega_1$
(two vertical steps of path). We conclude that
\begin{equation}\label{ref20}
Z_{m}^{20V_3}[\xi]=\frac{\omega_2[\xi]}{\omega_0[\xi]}\, \omega_0[\xi]^m\, Z_{m}^{20V_3}(\tau) \ , 
\end{equation}
where $Z_{m}^{20V_3}(\tau)= \sum_{k=1}^m Z_{m,k}^{20V_3} \tau^{k-1}$ and
$$\tau:= \frac{\omega_1[\xi]}{\omega_0[\xi]}=\tan(\xi+\frac{\pi}{4}) \ .$$

Analogously let $Z_{n}^{6V_1}[\xi]$ denote the partition function of the 6V-DWBC model on the square grid of size $n$ with all spectral parameters $z,w$ 
%(resp. $\bt,\bw$) 
fixed to the combinatorial values \eqref{combipoint20v},
except for the last column (where we set $w_1=e^{-i(\lambda+\eta+2\xi)}$ as before). The weights are $(a_1,b_1,c_1)=\frac{\beta_1}{\sqrt{2}}(1,\sqrt{2},1)$ 
%(resp. $(a_3,b_3,c_3)=\frac{\beta_3}{\sqrt{2}}(\sqrt{2},1,1)$) 
for all vertices except in the last column, where they read respectively:
$$6V_1:\quad a_1[\xi]=\beta_1[\xi] \cos(\xi+\frac{\pi}{4}),\quad b_1[\xi]=\beta_1[\xi] \cos(\xi),\quad c_1[\xi]=\frac{\beta_1[\xi]}{\sqrt{2}} .$$
%&&6V_3:\quad a_3[\xi]=\beta_3 \cos(\xi),\quad b_3[\xi]=\beta_3\sin(\xi+\frac{\pi}{4}),\quad c_3[\xi]=\frac{\beta_3}{\sqrt{2}}.
%\end{eqnarray*}
A decomposition of the 6V configurations similar to the above leads to the relation:
\begin{equation}\label{ref6}
Z_n^{6V_1}[\xi]=e^{-in\xi}\,\frac{c_1[\xi]\,a_1[0]}{a_1[\xi]\,c_1[0]}\,\left(\frac{a_1[\xi]}{a_1[0]}\right)^n\left(\frac{\beta_1}{\sqrt{2}}\right)^{n^2}\,Z_{n}^{6V_1}(\sigma_1),
\end{equation}
where $Z_{n}^{6V_1}(\sigma_1)=\sum_{k=1}^n Z_{n,k}^{6V_1}\, \sigma_1^{k-1}$ and
$$
\sigma_1=\frac{b_1[\xi]\, a_1[0]}{a_1[\xi]\, b_1[0]}=\frac{\cos(\xi)}{\sqrt{2}\cos(\xi+\frac{\pi}{4})} =\frac{\tau+1}{2} .\\
%\sigma_3&=&\frac{b_3[\xi]\, a_3[0]}{a_3[\xi]\, b_3[0]}=\frac{\sqrt{2}\sin(\xi+\frac{\pi}{4})}{\cos(\xi)} =\frac{2\tau}{\tau+1} . 
$$
Here we decompose the 6V${}_1$ configurations according to the position $k$ where the topmost path hits the last column for the first time, and denote by $Z_{n,k}^{6V_1}$ their partition function with uniform weights $(a,b,c)=(1,\sqrt{2},1)$: indeed the prefactor $(\beta_1/\sqrt{2})^{n^2}$ allows to replace the initial weights 
$(a_1,b_1,c_1)$ with those.
To get the correct inhomogeneous weights in the last column, the weight of these configurations must be multiplied by $\left(\frac{b_1[\xi]}{b_1[0]}\right)^{k-1}$ (for the bottom $k-1$ vertices in configuration $b_1$),
$\frac{c_1[\xi]}{c_1[0]}$ (for the vertex at position $k$, in configuration $c_1$) and $\left(\frac{a_1[\xi]}{a_1[0]}\right)^{n-k}$ (for the top $n-k$ vertices in the empty configuration $a_1$).

With this new choice of weights, the relations \eqref{inhomrela} imply:
\begin{eqnarray*}
Z_{2n}^{20V_3}[\xi]&=& (a_2 a_3)^{n(2n+1)}\, b_1^{n(n+1)}\left( \frac{b_1[\xi]}{b_1[0]}\right)^{n}\left( \frac{a_3[\xi]}{a_3[0]}\right)^{2n}\left(\frac{\beta_1}{\sqrt{2}}\right)^{n^2} \, Z_n^{6V_1}[\xi] \\
&=& 2^{n(n+1)/2}\, \left( \frac{b_1[\xi]}{b_1[0]}\right)^{n}\left( \frac{a_3[\xi]}{a_3[0]}\right)^{2n}\, Z_n^{6V_1}[\xi]\\
Z_{2n-1}^{20V_3}[\xi]&=& (a_2 a_3)^{n(2n-1)}\, b_1^{n(n-1)} \left( \frac{b_1[\xi]}{b_1[0]}\right)^{n-1}\left( \frac{a_3[\xi]}{a_3[0]}\right)^{2n-1}\left(\frac{\beta_1}{\sqrt{2}}\right)^{n^2}\, Z_n^{6V_1}[\xi]\\
&=&2^{n(n-1)/2}\, \left( \frac{b_1[\xi]}{b_1[0]}\right)^{n-1}\left( \frac{a_3[\xi]}{a_3[0]}\right)^{2n-1}\, Z_n^{6V_1}[\xi] .\\
%Z_{2n}^{20V_{1-3}}[\xi]&=& (a_2 b_1)^{n(2n+1)}\, a_3^{n(n+1)}\left( \frac{a_3[\xi]}{a_3[0]}\right)^{n}\left( \frac{b_1[\xi]}{b_1[0]}\right)^{2n}\left(\frac{\beta_3}{\sqrt{2}}\right)^{n^2} \, Z_n^{6V_3}[\xi] \\
%&=& 2^{n(n+1)/2}\, \left( \frac{b_1[\xi]}{b_1[0]}\right)^{2n}\left( \frac{a_3[\xi]}{a_3[0]}\right)^{n}\, Z_n^{6V_3}[\xi]\\
%Z_{2n-1}^{20V_{1-3}}[\xi]&=& (a_2 b_1)^{n(2n-1)}\, a_3^{n(n-1)} \left( \frac{a_3[\xi]}{a_3[0]}\right)^{n-1}\left( \frac{b_1[\xi]}{b_1[0]}\right)^{2n-1}\left(\frac{\beta_3}{\sqrt{2}}\right)^{n^2}\, Z_n^{6V_3}[\xi]\\
%&=&2^{n(n-1)/2}\, \left( \frac{b_1[\xi]}{b_1[0]}\right)^{2n-1}\left( \frac{a_3[\xi]}{a_3[0]}\right)^{n-1}\, Z_n^{6V_3}[\xi] .
\end{eqnarray*}
There are analogous relations for the 20V${}_{1,2}$ models defined in Section \ref{transfosec}, respectively related to 6V-DWBC partition functions on the sublattices 2 and 3 respectively, easily obtained by following the effect of the transformations $\bf S\circ \bf R\circ \bf VF$ and $\bf R^*$ respectively. More precisely, we must consider the 
partition function 
$Z_n^{6V_3}[\xi]$ in which all spectral parameters $t_i,w_i$ are set to their combinatorial point value except for $w_1\to e^{2i\xi} w_1$, and the partition function
$Z_n^{6V_2}[\xi]$ in which all spectral parameters $z_i,t_i$ are set to their combinatorial point value except for $z_1\to e^{2i\xi} z_1$.
These lead to the following. 

\begin{thm}\label{ref20v6vthm}
The refined partition function $Z_m^{20V_3}(\tau)$, and that of the 6V-DWBC model 
$Z_n^{6V_1}(\sigma)$ are related via:
\begin{eqnarray*}
Z_{2n}^{20V_3}(\tau)\!\!\!&=&\!\!\!2^{\frac{n(n-1}{2}}  (\tau+1)^n Z_n^{6V_1}\left({\scriptstyle\frac{\tau+1}{2}}\right), \  
Z_{2n-1}^{20V_3}(\tau)=2^{\frac{(n-1)(n-2)}{2}}(\tau+1)^{n-1} Z_n^{6V_1}\left({\scriptstyle\frac{\tau+1}{2}}\right)  .
%\\
%Z_{2n}^{20V_{2}}(\tau)\!\!\!&=&\!\!\!2^{\frac{(n-1)(n-2)}{2}} (\tau+1)^{2n-1} Z_n^{6V_3}\left({\scriptstyle\frac{2\tau}{\tau+1}}\right), \  
%Z_{2n-1}^{20V_{2}}(\tau)=2^{\frac{(n-1)(n-4)}{2}} (\tau+1)^{2n-2} Z_n^{6V_3}\left({\scriptstyle\frac{2\tau}{\tau+1}}\right) .
\end{eqnarray*}
Moreover, we have for all $m\geq 1$:
$$Z_{m}^{20V_{1}}(\tau)=Z_{m}^{20V_{2}}(\tau)=\tau^{m-1}\, Z_{m}^{20V_{3}}(\tau^{-1}) .$$
\end{thm}
\begin{proof}
In the 20V${}_3$ case, we use Eqns. \eqref{ref20}, \eqref{ref6} for $i=1$, and the relation $\sigma_1=\frac{\tau+1}{2}$ to rewrite:
\begin{eqnarray*}
Z_{2n}^{20V_3}(\tau)&=&\frac{2^{n(n+1)/2}}{\omega_2[\xi]\omega_0[\xi]^{2n-1}}\, \left( \frac{b_1[\xi]}{b_1[0]}\right)^{n}\left( \frac{a_3[\xi]}{a_3[0]}\right)^{2n}\,
\frac{c_1[\xi]\,a_1[0]}{a_1[\xi]\,c_1[0]}\,\left(\frac{a_1[\xi]}{a_1[0]}\right)^n\,Z_n^{6V_1}(\sigma_1)\\
&=&2^{n(n+1)/2} \,\frac{\omega_0[\xi]}{\omega_2[\xi]}\,\frac{c_1[\xi]\,a_1[0]}{a_1[\xi]\,c_1[0]}\, \sigma_1^n \,Z_n^{6V_1}(\sigma_1)=2^{n(n+1)/2} \,\sigma_1^n \,Z_n^{6V_1}(\sigma_1) ,
\end{eqnarray*}
where we used the relation $\frac{\omega_0[\xi]}{\omega_2[\xi]}\,\frac{c_1[\xi]\,a_1[0]}{a_1[\xi]\,c_1[0]}=1$. Similarly, we have:
\begin{eqnarray*}
Z_{2n-1}^{20V_3}(\tau)&=&\frac{2^{n(n-1)/2}}{\omega_2[\xi]\omega_0[\xi]^{2n-2}}\, \left( \frac{b_1[\xi]}{b_1[0]}\right)^{n-1}\left( \frac{a_3[\xi]}{a_3[0]}\right)^{2n-1}\,
\frac{c_1[\xi]\,a_1[0]}{a_1[\xi]\,c_1[0]}\,\left(\frac{a_1[\xi]}{a_1[0]}\right)^{n}\,Z_n^{6V_1}(\sigma_1)\\
&=&2^{n(n-1)/2} \,\frac{\omega_0[\xi]}{\omega_2[\xi]}\,\frac{c_1[\xi]\,a_1[0]}{a_1[\xi]\,c_1[0]}\, \sigma_1^{n-1} \,Z_n^{6V}(\sigma_1)=2^{n(n-1)/2} \,\sigma_1^{n-1} \,Z_n^{6V_1}(\sigma_1) .
\end{eqnarray*}
%Similarly, in the 20V${}_{2}$ case:
%\begin{eqnarray*}
%Z_{2n}^{20V_{1-3}}(\tau)&=&\frac{2^{n(n+1)/2}}{\omega_2[\xi]\omega_0[\xi]^{2n-1}}\, \left( \frac{b_1[\xi]}{b_1[0]}\right)^{2n}\left( \frac{a_3[\xi]}{a_3[0]}\right)^{n}\,
%\frac{c_3[\xi]\,a_3[0]}{a_3[\xi]\,c_3[0]}\,\left(\frac{a_3[\xi]}{a_3[0]}\right)^n\,Z_n^{6V_3}(\sigma_3)\\
%&=&2^{n(n+1)/2} \,\frac{\omega_0[\xi]}{\omega_2[\xi]}\,\frac{c_3[\xi]\,a_3[0]}{a_3[\xi]\,c_3[0]}\, \sigma_1^{2n} \,Z_n^{6V_3}(\sigma_3)=2^{n(n+1)/2} \,\sigma_1^{2n-1} \,Z_n^{6V_3}(\sigma_3) ,\\
%Z_{2n-1}^{20V_{1-3}}(\tau)&=&\frac{2^{n(n-1)/2}}{\omega_2[\xi]\omega_0[\xi]^{2n-2}}\, \left( \frac{b_1[\xi]}{b_1[0]}\right)^{2n-1}\left( \frac{a_3[\xi]}{a_3[0]}\right)^{n-1}\,
%\frac{c_3[\xi]\,a_3[0]}{a_3[\xi]\,c_3[0]}\,\left(\frac{a_3[\xi]}{a_3[0]}\right)^{n}\,Z_n^{6V_3}(\sigma_3)\\
%&=&2^{n(n-1)/2} \,\frac{\omega_0[\xi]}{\omega_2[\xi]}\,\frac{c_3[\xi]\,a_3[0]}{a_3[\xi]\,c_3[0]}\, \sigma_1^{2n-1} \,Z_n^{6V_3}(\sigma_3)=2^{n(n-1)/2} \,\sigma_1^{2n-2} \,Z_n^{6V_3}(\sigma_3) ,
%\end{eqnarray*}
%where we used the relation $\frac{\omega_0[\xi]}{\omega_2[\xi]}\,\frac{c_3[\xi]\,a_3[0]}{a_3[\xi]\,c_3[0]} \sigma_1=1$.
The refined partition functions for the cases of 20V${}_2$ and 20V${}_1$ follow from Corollaries \ref{23cor} and \ref{12cor} respectively. The Theorem follows.
\end{proof}

A direct consequence of this theorem is the following:
\begin{cor} The refined even and odd partition functions of all 20V${}_i$ models $i=1,2,3$ are related via
$$ Z_{2n}^{20V_i}(\tau)=2^{n-1}\,(1+\tau)\, Z_{2n-1}^{20V_i}(\tau) , \quad i=1,2,3.$$
\end{cor}

A compact formula for the refined 6V-DWBC partition function $Z_n^{6V_1}\left({\scriptstyle\frac{\tau+1}{2}}\right)$ is \cite{DFG20V}:
$$Z_n^{6V_1}\left({\scriptstyle\frac{\tau+1}{2}}\right)=\det\limits_{0\leq i,j\leq n-1}\left.\left(\frac{1}{1-x y} +\frac{2x}{(1-x)(1-x-y-xy)}+x y^{n-1}\frac{\tau-1}{1-\tau x} \frac{(1+x)^n}{(1-x)^{n+1}}\right)\right\vert_{x^i y^j} .
$$
This was obtained by identifying it with the refined partition function of the 20V-DWBC2 model on a square grid of size $n$, itself identical to the ``type 2"
refined domino tiling partition function of the Holey Aztec square of size $2n$ (see Ref. \cite{DFG20V} for details). 
This gives an easy access to the refined
20V partition functions on $\cT_m$. We have for instance:
\begin{eqnarray*}
Z_{1}^{20V_3}(\tau)\!\!&=&\!\!1\\
Z_{2}^{20V_3}(\tau)\!\!&=&\!\!1+\tau\\
Z_{3}^{20V_3}(\tau)\!\!&=&\!\!2+3\tau+\tau^2\\
Z_{4}^{20V_3}(\tau)\!\!&=&\!\!4 + 10\tau + 8 \tau^2 + 2 \tau^3\\
Z_{5}^{20V_3}(\tau)\!\!&=&\!\!20 + 60  \tau+ 66  \tau+ 32 \tau^3 + 6 \tau^4\\
Z_{6}^{20V_3}(\tau)\!\!&=&\!\!80 + 320  \tau + 504  \tau^2 + 392  \tau^3 + 152  \tau^4 + 24  \tau^5\\
Z_{7}^{20V_3}(\tau)\!\!&=&\!\!976 + 4384\tau + 8144\tau^2 + 8072\tau^3 + 4552\tau^4 + 1400\tau^5 + 184 \tau^6\\
Z_{8}^{20V_3}(\tau)\!\!&=&\!\!7808 \!+\! 42880 \tau \!+\! 100224 \tau^2 \!+\! 129728 \tau^3 \!+ \!100992 \tau^4 \!+\! 47616\tau^5 \!+\! 12672 \tau^6 \!+\! 1472 \tau^7
\end{eqnarray*}
and
\begin{eqnarray*}
Z_{1}^{20V_{2}}(\tau)\!\!&=&\!\!1\\
Z_{2}^{20V_{2}}(\tau)\!\!&=&\!\!1+\tau\\
Z_{3}^{20V_{2}}(\tau)\!\!&=&\!\!1+3\tau+2\tau^2\\
Z_{4}^{20V_{2}}(\tau)\!\!&=&\!\!2 + 8\tau + 10 \tau^2 + 4 \tau^3\\
Z_{5}^{20V_{2}}(\tau)\!\!&=&\!\!6 + 32  \tau+ 66  \tau+ 60 \tau^3 + 20 \tau^4\\
Z_{6}^{20V_{2}}(\tau)\!\!&=&\!\!24 + 152  \tau + 392  \tau^2 + 504  \tau^3 + 320  \tau^4 + 80  \tau^5\\
Z_{7}^{20V_{2}}(\tau)\!\!&=&\!\!184 + 1400\tau + 4552\tau^2 + 8072\tau^3 +8144 \tau^4 + 4384\tau^5 + 976 \tau^6\\
Z_{8}^{20V_{2}}(\tau)\!\!&=&\!\! \!1472+\! 12672 \tau \!+\! 47616 \tau^2 \!+\! 100992 \tau^3 \!+ \!129728 \tau^4 \!+\! 100224\tau^5 \!+\! 42880 \tau^6 \!+\! 7808 \tau^7
\end{eqnarray*}

Note the obvious relation $Z_{m,m-1}^{20V_3}=Z_{m-2}^{20V_3}=Z_{m-2}^{20V_3}(1)$, as the contributing configurations have a top path made of $m$ horizontal steps followed by $m$ vertical ones,
while the remaining paths may take any configuration on $\cT_{m-2}$.

\section{Asymptotic results: Arctic curves}\label{arcticsec}

\subsection{Tangent method}

\begin{figure}
\begin{center}
\includegraphics[width=12cm]{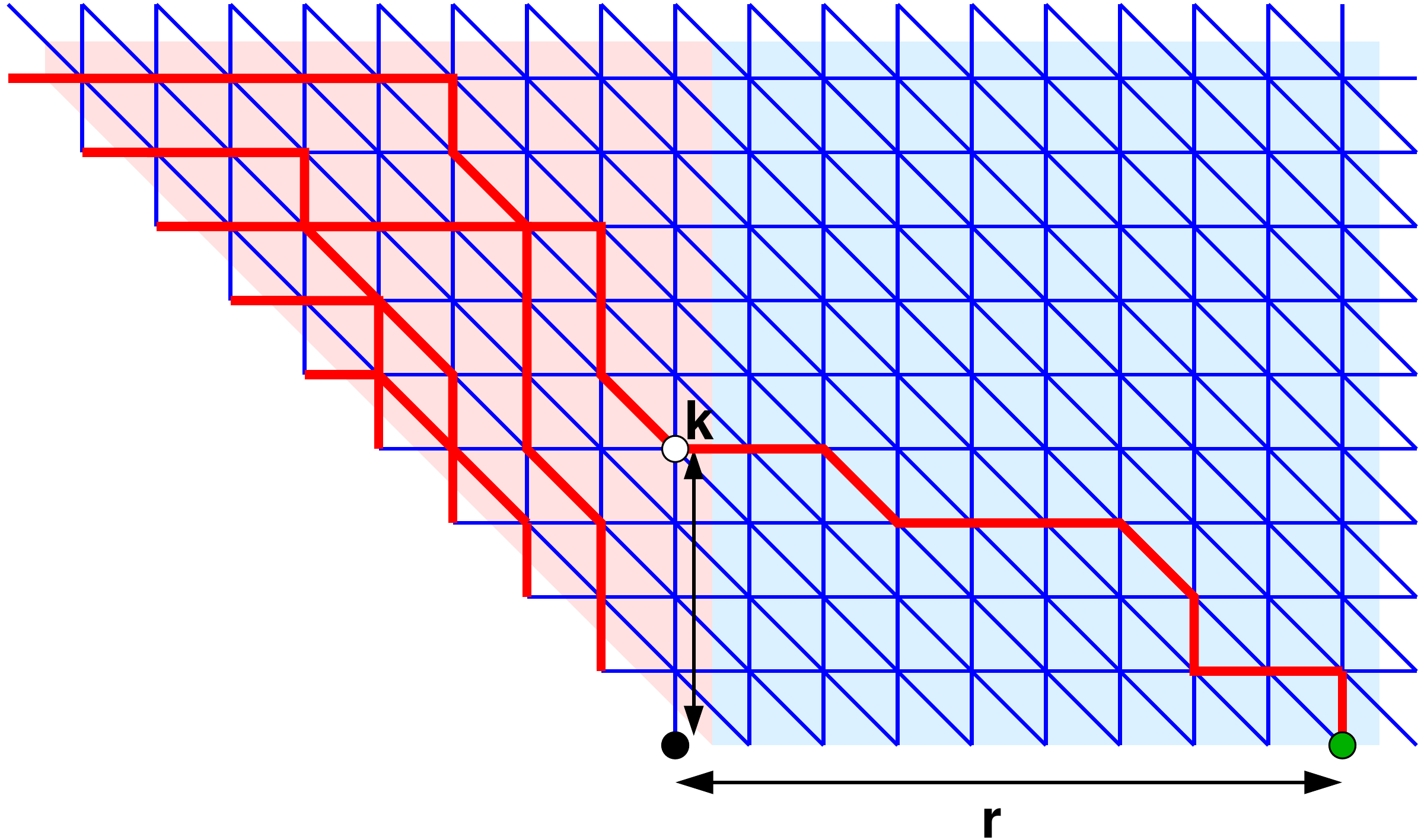}
\end{center}
\caption{\small The tangent method applied to the 20V${}_3$ model.  The original endpoint (black dot) of the outer path is moved to a position at distance $r$ (green dot).
The outer path now exits the original domain (shaded in pink) at point with position $k$ on the vertical line (white dot). The total partition function is a sum over the position $k$ of the product of the partition function for the pink domain (normalized into the one-point function $H^{20V_3}_{m,k}$) and the partition function for a single path in the blue domain
($Y_{k,r}$). All vertices are weighted with 20V weights.}
\label{fig:tangent}
\end{figure}

The tangent method, devised by F. Colomo and A. Sportiello \cite{COSPO} is a non-rigorous method to determine the arctic curve for various statistical models of paths exhibiting crystalline phases (with ordered paths) and liquid phases (with disordered paths) separated by a sharp transition curve usually called arctic curve. Starting from 
a path formulation of the model (such as the osculating Schr\"oder paths in the present case), one identifies one of the phase separations with the outermost path which is the natural boundary between the empty cristalline phase and a disordered path phase. To study the asymptotic shape of this separation, the method consists in moving 
the endpoint of this outermost path away from the original domain. The path is expected to still determine the arctic separation, until it detaches itself from the other paths
and essentially follows a geodesic (here a straight line) until its new endpoint. For large sizes, the latter becomes tangent to the arctic curve. This line is determined by its endpoint and the point at which it exits the original domain. We represent in Fig. \ref{fig:tangent} the setting for the tangent method applied to the 20V${}_3$ 
model on the triangle $\cT_m$. The total partition function of the new model (with moved endpoint) decomposes into a sum over the exit position $k$ of the outer path
of a product of two partition functions (corresponding respectively to the pink and blue shaded areas).
The first piece is a new partition function\footnote{The notation $\tilde Z_{m,k}^{20V_3}$ is to distinguish this quantity from the refined partition function $Z_{m,k}^{20V_3}$, which has $k$ additional vertical steps joining the white dot to the black one.} $\tilde Z_{m,k}^{20V_3}[z,t,w]$ identical to $Z_m^{20V_3}[z,t,w]$ except for the last endpoint, moved from its original position  (black dot) to position $k$ (white dot). (In this section, the weights of the 20V model are uniform but arbitrary: they correspond to choosing all $z_i=z$, $t_i=t$, $w_i=w$; 
to make the dependence on $z,t,w$ explicit, we write them as arguments from now on.).
The latter partition function is usually normalized by dividing it with the original partition function $Z_m^{20V_3}[z,t,w]$, and is called the one-point function
$$H_{m,k}^{20V_3}[z,t,w]=\frac{\tilde Z_{m,k}^{20V_3}[z,t,w]}{Z_m^{20V_3}[z,t,w]} .$$
The second piece is the partition function $Y_{k,r}[z,t,w]$ of a single path between the white dot and the green dot, namely of single Schr\"oder paths in $\Z_+^2$
from point $(0,k)$ to $(r,0)$, weighted by the product on their vertices of the local 20V weights \eqref{weights20V} functions of $z,t,w$, 
whereas all empty vertices in the blue domain receive the weight $\omega_0[z,t,w]$.

The total partition function reads:
$$\hat Z^{20V_3}_{m,r}[z,t,w]=\sum_{k=1}^m H_{m,k}^{20V_3}[z,t,w]\, Y_{k,r}[z,t,w] .$$
The tangent method is based on the remark that for large size $m$, the {\it most probable} exit position corresponds to the dominant contribution to the sum, and
is therefore determined as the solution of an extremization problem. The latter involves first estimating the large $m$ ($k/m, r/m$ finite) behavior of the two quantities $H_{m,k}^{20V_3}[z,t,w]$ and $Y_{k,r}[z,t,w]$, which will be done in the two following sections.

Finally, having determined a family of tangent lines parameterized by the displaced position $r$, we will deduce the arctic curve as the envelope of this family.

\subsection{20V one-point function and asymptotics}

As explained above, the first ingredient of the Tangent Method is the one-point function of the 20V model defined as
the ratio
$$H_{m,k}^{20V_3}[z,t,w]= \left(\frac{\omega_0}{\omega_1}\right)^{k-1}\frac{Z_{m,k}^{20V}[z,t,w]}{Z_{m}^{20V}[z,t,w]}$$
of the refined 20V partition function obtained by conditioning the topmost path to end at its first visit to the rightmost vertical, at position $k\in \{1,2,...,m\}$,
by the total partition function $Z_{m}^{20V}[z,t,w]$. The prefactor replaces the weight $\omega_1^{k-1}$ of the original refined partition function $Z_{m,k}^{20V}[z,t,w]$
by $\omega_0^{k-1}$ as we omit the last $k-1$ (vertical) steps of the path. 
%Here we pick all spectral horizontal, diagonal, vertical spectral parameters 
%to be respectively $z$, $t$, $w$ as in \eqref{param}, leading to the general homogeneous weights \eqref{eq:explweights}, and use a notation that makes the dependence on $z,t,w$ explicit. 

As before, we may compute the quantities $Z_{m,k}^{20V}[z,t,w]$ for arbitrary homogeneous weights by modifying the rightmost vertical parameter to $w_1=w e^{-2i\xi}$.
By the relation \eqref{inhomrela} applied to parameters $z_1=z_2=...=z_m=z$, $t_1=t_2=...=t_m=t$, $w_1=w e^{-2i\xi}$ and $w_2=w_3=...=w_m=w$, 
the corresponding partition function $Z_m^{20V}[z,t,w,\xi]$ of the 20V model is related to that $Z_n^{6V_1}[z,w,\xi]$ of the 6V model with horizontal spectral parameters $z$ 
and vertical spectral parameters $w$ except in the last column where $w\to w_1$, via
\begin{eqnarray}\label{20to6gen}
Z_{2n}^{20V_3}[z,t,w,\xi]&=&(a_2a_3)^{n(2n+1)} \,b_1^{n(n+1)}\, \left( \frac{b_1[\xi]}{b_1[0]}\right)^{n}\left( \frac{a_3[\xi]}{a_3[0]}\right)^{2n}\, Z_n^{6V_1}[z,w,\xi]\nonumber \\
Z_{2n-1}^{20V_3}[z,t,w,\xi]&=&(a_2a_3)^{n(2n-1)} \,b_1^{n(n-1)}\, \left( \frac{b_1[\xi]}{b_1[0]}\right)^{n-1}\left( \frac{a_3[\xi]}{a_3[0]}\right)^{2n-1}\, Z_n^{6V_1}[z,w,\xi] ,
\end{eqnarray}
where $a_i,b_i,c_i,\omega_i$ stand for $a_i[0],b_i[0],c_i[0],\omega_i[0]$ respectively.
Repeating the analysis of the refined partition functions for the 20V${}_3$ model, noting that $\omega_2[\xi] \neq \omega_4[\xi]$ in general,
we must now distinguish the contribution $Z_{m,k}^{20V_3\, \mbox{--}}[z,t,w]$ of configurations that reach the last vertical at position $k$ with a horizontal step from that $Z_{m,k}^{20V_3\, {\scriptscriptstyle{\diagdown}}}[z,t,w]$ of configurations that reach the last vertical with a diagonal step,
with $Z_{m,k}^{20V_3}[z,t,w]=Z_{m,k}^{20V_3\, \mbox{--}}[z,t,w]+Z_{m,k}^{20V_3\, {\scriptscriptstyle{\diagdown}}}[z,t,w]$. We find that
\begin{equation}\label{ref20ztw}
Z_{m}^{20V_3}[z,t,w,\xi]=\frac{\omega_2[\xi]\omega_0[0]}{\omega_0[\xi]\omega_2[0]}\, \left(\frac{\omega_0[\xi]}{\omega_0[0]}\right)^{m}\,  
\sum_{k=1}^m \left(Z_{m,k}^{20V_3\, {\scriptscriptstyle{\diagdown}}}[z,t,w] +\gamma\, Z_{m,k}^{20V_3\, \mbox{--}}[z,t,w]\right) \tau^{k-1} \ , 
\end{equation}
where
\begin{equation}\label{taugen}
\tau=\tau[\xi]:=\frac{\omega_1[\xi]\, \omega_0[0]}{\omega_0[\xi]\, \omega_1[0]}= \frac{\sin(\xi+\lambda-\eta)\sin(\lambda+\eta)\,\sin\left(\xi+{\scriptstyle \frac{\lambda-\eta+\mu}{2}}\right)\sin\left({\scriptstyle \frac{\lambda+3\eta+\mu}{2}}\right)}{\sin(\xi+\lambda+\eta)\sin(\lambda-\eta)\,\sin\left(\xi+{\scriptstyle \frac{\lambda+3\eta+\mu}{2}}\right)\sin\left({\scriptstyle \frac{\lambda-\eta+\mu}{2}}\right)}
\end{equation}
and
$$\gamma=\frac{\omega_4[\xi]\omega_2[0]}{\omega_2[\xi]\omega_4[0]} =\frac{\sin(\lambda-\eta)\sin\left(\xi+{\scriptstyle \frac{\lambda+3\eta+\mu}{2}}\right)}{\sin(\xi+\lambda-\eta)\sin\left({\scriptstyle \frac{\lambda+3\eta+\mu}{2}}\right)}$$

Similarly for the 6V${}_1$ model, we get
\begin{equation}\label{ref6zw}
Z_n^{6V_1}[z,w,\xi]=\frac{c_1[\xi]\,a_1[0]}{a_1[\xi]\,c_1[0]}\,\left(\frac{a_1[\xi]}{a_1[0]}\right)^n\left(\frac{\beta_1}{\sqrt{2}}\right)^{n^2}\,\sum_{k=1}^n Z_{n,k}^{6V_1}[z,w]\, \sigma^{k-1} ,
\end{equation}
where
\begin{equation}\label{siggen}\sigma=\sigma[\xi]:=\frac{b_1[\xi]\, a_1[0]}{a_1[\xi]\, b_1[0]}=\frac{\sin(\xi+\lambda-\eta)\,\sin(\lambda+\eta)}{\sin(\xi+\lambda+\eta)\,\sin(\lambda-\eta)} . \end{equation}

By a slight abuse of notation, we still denote by $Z_m^{20V_3}(\tau)$ and $Z_n^{6V_1}(\sigma)$ the refined partition functions with arbitrary homogeneous weights
\begin{eqnarray*}
Z_m^{20V_3}(\tau)&=&\sum_{k=1}^m \left(Z_{m,k}^{20V_3\, {\scriptscriptstyle{\diagdown}}}[z,t,w] +\gamma\, Z_{m,k}^{20V_3\, \mbox{--}}[z,t,w]\right) \, \tau^{k-1} \\
Z_n^{6V_1}(\sigma)&=&\sum_{k=1}^n Z_{n,k}^{6V_1}[z,w]\, \sigma^{k-1} .
\end{eqnarray*}
These are related via the following 
\begin{thm}\label{20to6genthm}
We have the relations:
\begin{eqnarray*}
Z_{2n}^{20V_3}(\tau)&=&(a_2a_3)^{n(2n+1)} \,b_1^{n(n+1)}\,\gamma\,  \sigma^n\,Z_n^{6V_1}(\sigma)\\
Z_{2n-1}^{20V_3}(\tau)&=&(a_2a_3)^{n(2n-1)} \,b_1^{n(n-1)}\,\gamma\,  \sigma^{n-1}\,Z_n^{6V_1}(\sigma)
\end{eqnarray*}
\end{thm}
\begin{proof}
Combining the even case of \eqref{20to6gen} with \eqref{ref20ztw} and \eqref{ref6zw}, we find:
\begin{eqnarray*}
Z_{2n}^{20V_3}(\tau)\!\!\!\!\!&=&\!\!\!\!\!(a_2a_3)^{n(2n+1)} \,b_1^{n(n+1)}\,\frac{\omega_2[0]\omega_0[\xi]}{\omega_0[0]\omega_2[\xi]}\frac{c_1[\xi]\,a_1[0]}{a_1[\xi]\,c_1[0]}\, \!\!\!\left(\frac{\omega_0[0]}{\omega_0[\xi]}\frac{a_1[\xi]}{a_1[0]}\frac{a_3[\xi]}{a_3[0]}\right)^{2n} \!\!\!\left( \frac{b_1[\xi]a_1[0]}{b_1[0]a_1[\xi]}\right)^{n}\, Z_n^{6V_1}(\sigma)\\
\!\!\!\!\!&=&\!\!\!\!\!(a_2a_3)^{n(2n+1)} \,b_1^{n(n+1)}\,\gamma\,  \sigma^n\,Z_n^{6V_1}(\sigma)
\end{eqnarray*}
by using $\frac{\omega_0[\xi]}{\omega_0[0]}=\frac{a_1[\xi]a_3[\xi]}{a_1[0]a_3[0]}$ and 
$\frac{\omega_2[0]\omega_0[\xi]}{\omega_0[0]\omega_2[\xi]}\frac{c_1[\xi]\,a_1[0]}{a_1[\xi]\,c_1[0]}=\gamma$. The odd case follows similarly.
\end{proof}

The above analysis suggests to introduce two refined one-point functions:
$$H_{m,k}^{20V_3\, {\scriptscriptstyle{\diagdown}}}[z,t,w]
=\left(\frac{\omega_0}{\omega_1}\right)^{k-1}\frac{Z_{m,k}^{20V_3\, {\scriptscriptstyle{\diagdown}}}[z,t,w]}{Z_{m}^{20V_3}[z,t,w]},  
\quad H_{m,k}^{20V_3\, \mbox{--}}[z,t,w]
=\left(\frac{\omega_0}{\omega_1}\right)^{k-1}\frac{Z_{m,k}^{20V_3\, \mbox{--}}[z,t,w]}{Z_{m}^{20V_3}[z,t,w]}$$
with $H_{m,k}^{20V_3}[z,t,w]=H_{m,k}^{20V_3\, {\scriptscriptstyle{\diagdown}}}[z,t,w]+H_{m,k}^{20V_3\, \mbox{--}}[z,t,w]$. 
Theorem \ref{20to6genthm} may be rephrased as:
\begin{equation}\label{htoh}
H_{m,k}^{20V_3\, {\scriptscriptstyle{\diagdown}}}[z,t,w]+\gamma H_{m,k}^{20V_3\, \mbox{--}}[z,t,w]=\left(\frac{\omega_0}{\omega_1}\right)^{k-1}\,\oint \frac{d\tau}{2i\pi\tau^k}\gamma\,  \sigma^{\lfloor \frac{m}{2}\rfloor} \,\frac{Z_{\lfloor \frac{m}{2}\rfloor}^{6V_1}(\sigma)}{Z_{\lfloor \frac{m}{2}\rfloor}^{6V_1}}
\end{equation}
where the contour integral around the origin picks up the coefficient of $\tau^{k-1}$ (
by interpreting $\gamma=\gamma[\xi]$ and $\sigma=\sigma[\xi]$ of Eqs. \eqref{taugen} and \eqref{siggen} as implicit series expansions of $\tau=\tau[\xi]$).

We are now ready to estimate $H_{m,k}^{20V}$ in the scaling limit where $k,m$ are large, while $\kappa=2k/m$ remains finite. 
We resort to the result\footnote{See also \cite{DF21V} for a compact proof.} of \cite{CP2010}
for the large $n=\lfloor \frac{m}{2}\rfloor$  asymptotics of the 6V one-point function 
$H_n^{6V_1}(\sigma)=\frac{Z_{\lfloor \frac{m}{2}\rfloor}^{6V_1}(\sigma)}{Z_{\lfloor \frac{m}{2}\rfloor}^{6V_1}}$:
$$H_n^{6V_1}(\sigma)\sim e^{n f(\sigma)} $$
where the function $f$ is the following implicit function of $\sigma=\sigma[\xi]$ parametrized by $\xi$:
\begin{equation}
f(\sigma[\xi])= {\rm Log}\left(\frac{\sin(\al(\lambda-\eta))\, \sin(\xi+\lambda-\eta)\,\sin(\al\xi)}{\al\, \sin(\lambda-\eta)\, \sin(\al(\xi+\lambda-\eta))\,\sin(\xi)} \right)
\end{equation}
and where $\al=\frac{\pi}{\pi-2\eta}$.

This exponential growth of $H_n^{6V_1}(\sigma)$ induces an exponential growth for the combination
$H_{m,k}^{20V_3\, {\scriptscriptstyle{\diagdown}}}[z,t,w]+\gamma\, H_{m,k}^{20V_3\, \mbox{--}}[z,t,w]$ in \eqref{htoh}.
As the parameter $\gamma$ is independent of $m$, this combination is dominated by either term or both at large $m$. 
We deduce that $H_{m,k}^{20V_3}[z,t,w]\sim H_{m,k}^{20V_3\, {\scriptscriptstyle{\diagdown}}}[z,t,w]+\gamma\, H_{m,k}^{20V_3\, \mbox{--}}[z,t,w]$
have the same leading behavior as $m,k,n\to \infty$ with $m\sim 2n$ and $k\sim \kappa n$:
$$H_{2n,n\kappa}^{20V_3}[z,t,w]\sim \oint \frac{d\tau}{2i\pi\tau} e^{n S_V(\kappa,\xi,\tau)}$$
where we have introduced the Vertex model action
\begin{equation}\label{SV}
S_V(\kappa,\xi,\tau)=f(\sigma[\xi]) +{\rm Log}(\sigma[\xi]) -\kappa \,{\rm Log}\left(\tau[\xi]\frac{\omega_1}{\omega_0}\right)
\end{equation}
The leading contribution to the integral occurs at the saddle-point 
$\partial_\xi S_V(\kappa,\xi,\tau[\xi])=0$, which is solved as a parametric expression for $\kappa$:
\begin{eqnarray}
\kappa&=&\kappa[\xi]=\frac{\tau[\xi]}{\partial_\xi \tau[\xi]} \left( \partial_\xi f(\sigma[\xi]) +\frac{\partial_\xi \sigma[\xi]}{\sigma[\xi]} \right)\nonumber \\
&=&\Big(2\cot(\xi+\lambda-\eta)-\cot(\xi)-\cot(\xi+\lambda+\eta)+\alpha\cot(\alpha\, \xi)
-\alpha\cot(\alpha(\xi+\lambda-\eta))\Big)\nonumber \\
&&\ \ \times
\frac{\sin(\xi+\lambda+\eta)\sin(\xi+\lambda-\eta)\sin\left(
\xi+\frac{\lambda-\eta+\mu}{2}\right)\sin\left(
\xi+\frac{\lambda+3\eta+\mu}{2}\right)}{\sin(2\eta)\left(\sin(\xi+\lambda+\eta)\sin(\xi+\lambda-\eta)+\sin\left(
\xi+\frac{\lambda-\eta+\mu}{2}\right)\sin\left(
\xi+\frac{\lambda+3\eta+\mu}{2}\right)\right)}
\label{kappagen}
\end{eqnarray}

\subsection{Schr\"oder path partition function and asymptotics}

Consider the partition function $Y_{k,r}[z,t,w]$ for a single Schr\"oder path with 20V weights in the positive 
quadrant of $\Z^2$, starting form the point $(0,k)$ 
and ending at point $(r,0)$. In the uniform combinatorial case (where all $\omega_i=1$), we have a trinomial coefficient sum:
$$Y_{k,r}[z,t,w]=\sum_{\ell=0}^{\min(k,r)} {k+r-\ell \choose k-\ell,r-\ell,\ell} .$$
In the scaling limit where $n\to \infty$ with $\kappa=k/n$ and $\rho=r/n$ fixed ($\kappa\in [0,2]$ and $\rho\in [0,\infty)$), 
and replacing the sum by an integral over $\theta=\ell/n$ we have the estimate:
$$Y_{k,r}[z,t,w] \sim  \int_0^{\min(\kappa,\rho)} d\theta e^{-n S_{path}(\kappa,\rho,\theta)}, $$
where
$$S_{P}(\kappa,\rho,\theta)= (\kappa+\rho-\theta){\rm Log}(\kappa+\rho-\theta)
-\theta{\rm Log}(\theta)-(\kappa-\theta){\rm Log}(\kappa-\theta)-(\rho-\theta){\rm Log}(\rho-\theta) .$$
For arbitrary weights, a simple transfer matrix calculation \cite{BDFG} allows to compute $Y_{k,r}[z,t,w]$ as a function of the quantities
\begin{equation}
\begin{split}
\al_1&= \frac{\omega_1}{\omega_0} ,\quad \al_2=\frac{\omega_6}{\omega_0},\quad \al_3=\frac{\omega_0\omega_3+\omega_4^2-\omega_1\omega_6}{\omega_0^2},\quad
\al_4=\frac{\omega_2^2-\omega_1\omega_3}{\omega_0^2},\\
\al_5&= \frac{\omega_5^2-\omega_6\omega_3}{\omega_0^2},\quad 
\al_6= \frac{2\omega_2\omega_4\omega_5+\omega_1\omega_6\omega_3-\omega_3\omega_4^2-\omega_1\omega_5^2-\omega_6\omega_2^2}{\omega_0^3}\ .\\
\end{split}
\label{eq:alphas}
\end{equation}
and to derive the large $n$ scaling estimate for $(k,r)/n=(\kappa,\rho)$ fixed:
$$Y_{n\kappa,n\rho} \sim \int_0^1 dp_3dp_4dp_5dp_6  {\rm e}^{n\, S_{P}(\kappa,\rho,p_3,p_4,p_5,p_6)} ,$$
where the path action $S_P$ reads
\begin{equation}
\begin{split}
S_{P}(\kappa,\rho,p_3,p_4,p_5,p_6)&=(\kappa+\rho-p_3-2p_4-2p_5-3p_6){\rm Log}(\kappa+\rho-p_3-2p_4-2p_5-3p_6)\\
&\!\!\!\!\!\!\!\!\!\!\!\!\!\!\!\!\!\!\!\!\!\!\!\!\!\!\!\!-(\kappa-p_3-2p_4-p_5-2p_6){\rm Log}\left(\frac{\kappa-p_3-2p_4-p_5-2p_6}{\al_1}\right)\\
&\!\!\!\!\!\!\!\!\!\!\!\!\!\!\!\!\!\!\!\!\!\!\!\!\!\!\!\!-(\rho-p_3-p_4-2p_5-2p_6)
{\rm Log}\left(\frac{\rho-p_3-p_4-2p_5-2p_6}{\al_2}\right)-\sum_{i=3}^6 p_i{\rm Log}\left(\frac{p_i}{\al_i}\right) . 
\label{eq:explS}
\end{split}
\end{equation}
As discussed in \cite{BDFG}, the estimate remains valid if any subset of the $\omega_i$, $i=3,4,5,6$ vanishes, in which case we must set the corresponding $p_i$ to zero
and integrate over the leftover ones. In particular we recover the uniform case where $\al_4=\al_5=\al_6=0$ while $\al_1=\al_2=\al_3=1$, and by identifying the leftover
integration variable $p_3=\theta$.

\subsection{Arctic curve I: main (NE) branch}

Putting together the asymptotic results for the vertex model one-point function and the path partition function, we find the most likely
exit point $k=\kappa^* n$ for which the contribution to the total partition function
$\sum_{k=1}^m H_{m,k}^{20V_3} Y_{k,r} $ is maximal. We use the large $n$ estimate with $m\sim 2n$ and $(k,r)\sim n(\kappa,\rho)$:
$$\sum_{k=1}^m H_{m,k}^{20V_3} Y_{k,r} \sim n \int_0^2 d\kappa \oint \frac{d\tau}{2i\pi\tau} \int_0^1 dp_3dp_4dp_5dp_6 e^{n S} ,$$
where the total action is $S=S_V(\kappa,\xi,\tau)+S_P(\kappa,\rho,p_3,p_4,p_5,p_6)$. We already found the leading behavior of the Vertex model one-point function
in the form of a parametric equation relating $\kappa$ and $\tau$ via \eqref{taugen} and \eqref{kappagen}. We must now 
write the saddle-point equations
\begin{equation}\label{sapoS}
\partial_{p_i}S=\partial_{p_i}S_P=0 \quad (i=3,4,5,6), \qquad \partial_{\kappa} S=\partial_\kappa S_P-{\rm Log}(\tau\al_1) =0 .
\end{equation}
These equations are identical to those of Ref.\cite{BDFG}.
In terms of rescaled variables $q_i=p_i/\kappa$, $s=\rho/\kappa$ these read:
\begin{eqnarray*}
\frac{\al_1 \al_2}{\al_3}&=&\frac{( q_3 + 2 q_4 + q_5 + 2 q_6-1) (q_3 + q_4 + 2 q_5 + 2 q_6 - s) }{q_3 (s+1- q_3 - 2 q_4 -2 q_5 - 3 q_6 )}\\ 
\frac{\al_1^2 \al_2}{\al_4}&=&\frac{( q_3 + 2 q_4 + q_5 + 2 q_6-1)^2 (s-q_3 - q_4 - 2 q_5 - 2 q_6)}{q_4 ( q_3 + 2 q_4 + 2 q_5 + 3 q_6 - s-1)^2}\\
\frac{\al_1 \al_2^2}{\al_5}&=&\frac{( q_3 + 2 q_4 + q_5 + 2 q_6-1) (s-q_3 - q_4 - 2 q_5 - 2 q_6)^2}{q_5 ( q_3 + 2 q_4 + 2 q_5 + 3 q_6 - s-1)^2}\\
\frac{\al_1^2 \al_2^2}{\al_6}&=&\frac{( q_3 + 2 q_4 + q_5 + 2 q_6-1)^2 (s-q_3 - q_4 - 2 q_5 - 2 q_6)^2}{q_6 ( 1+s-q_3 - 2 q_4 - 2 q_5 - 3 q_6 )^3}\\
\tau&=&\frac{q_3 + 2 q_4 + 2 q_5 + 3 q_6 - s-1}{q_3 + 2 q_4 + q_5 + 2 q_6-1}
\end{eqnarray*} 
Eliminating the $q_i$'s and using the parametrization \eqref{eq:explweights} of the weights we obtain a parametric relation between $s$ and $\tau$ in the form $\tau=\tau[\xi]$
as in \eqref{taugen} and \cite{BDFG}:
\begin{eqnarray}s^{-1}=s[\xi]^{-1}&:=&\frac{\sin(\xi+\lambda+\eta)\,\sin(\xi+\lambda-\eta)}{\sin(\xi)\,\sin(\xi+2\eta)}\nonumber \\
&&\qquad \times  \frac{\sin(\xi)\,\sin(\xi+2\eta)+\sin(\xi+{\scriptstyle \frac{\lambda-\eta+\mu}{2}})\,\sin(\xi+{\scriptstyle \frac{\lambda+3\eta+\mu}{2}})}{\sin(\xi+\lambda+\eta)\,\sin(\xi+\lambda-\eta)+\sin(\xi+{\scriptstyle \frac{\lambda-\eta+\mu}{2}})\,\sin(\xi+{\scriptstyle \frac{\lambda+3\eta+\mu}{2}})}
\label{slopegen}
\end{eqnarray}

This gives the following family\footnote{Note that the rescaled origin is at the SE corner of the rescaled domain, which is defined by $x+y\geq0$, $x\leq 0$, $y\leq 2$.} 
of tangent lines, through the rescaled points $(\rho[\xi],0)$ and $(0,\kappa[\xi])$ with $A[\xi]:=s^{-1}[\xi]=\kappa[\xi]/\rho[\xi]$:
$$ y+ A[\xi] x -\kappa[\xi]=0 .$$
The envelope of this family has the following parametric equation:
$$ X[\xi]=\frac{\kappa'[\xi]}{A'[\xi]} , \quad Y[\xi]= \kappa[\xi]-A[\xi]\,\frac{\kappa'[\xi]}{A'[\xi]} . $$ 
The last task is to determine the range of the parameter $\xi$. From the definition of the method, the relevant family of lines have slopes $-A[\xi]$ ranging from $0$ to $-\infty$.
From the explicit expression \eqref{slopegen}, this determines the range of the parameter $\xi$ to be $\xi\in [0,\pi-\eta-\lambda]$.
The results of this section are summarized in the following.
\begin{thm}\label{NEbranchthm}
The North-East (NE) portion of the arctic curve for the 20V${}_3$ model on large triangles $\cT_m$ is given by the following parametric equations:
$$X_{NE}[\xi]=\frac{\kappa'[\xi]}{A'[\xi]} , \quad Y_{NE}[\xi]= \kappa[\xi]-A[\xi]\,\frac{\kappa'[\xi]}{A'[\xi]} \qquad \xi \in [0,\pi-\eta-\lambda]. $$
with $\kappa[\xi]$ as in \eqref{kappagen} and $A[\xi]=s^{-1}[\xi]$ as in \eqref{slopegen}.
\end{thm}

\subsection{Arctic curve II: other branches}

In this section, we show how to use the transformations of Section \ref{transfosec} to derive other branches of the arctic curve for the 20V${}_3$ model.

\begin{figure}
\begin{center}
\includegraphics[width=13cm]{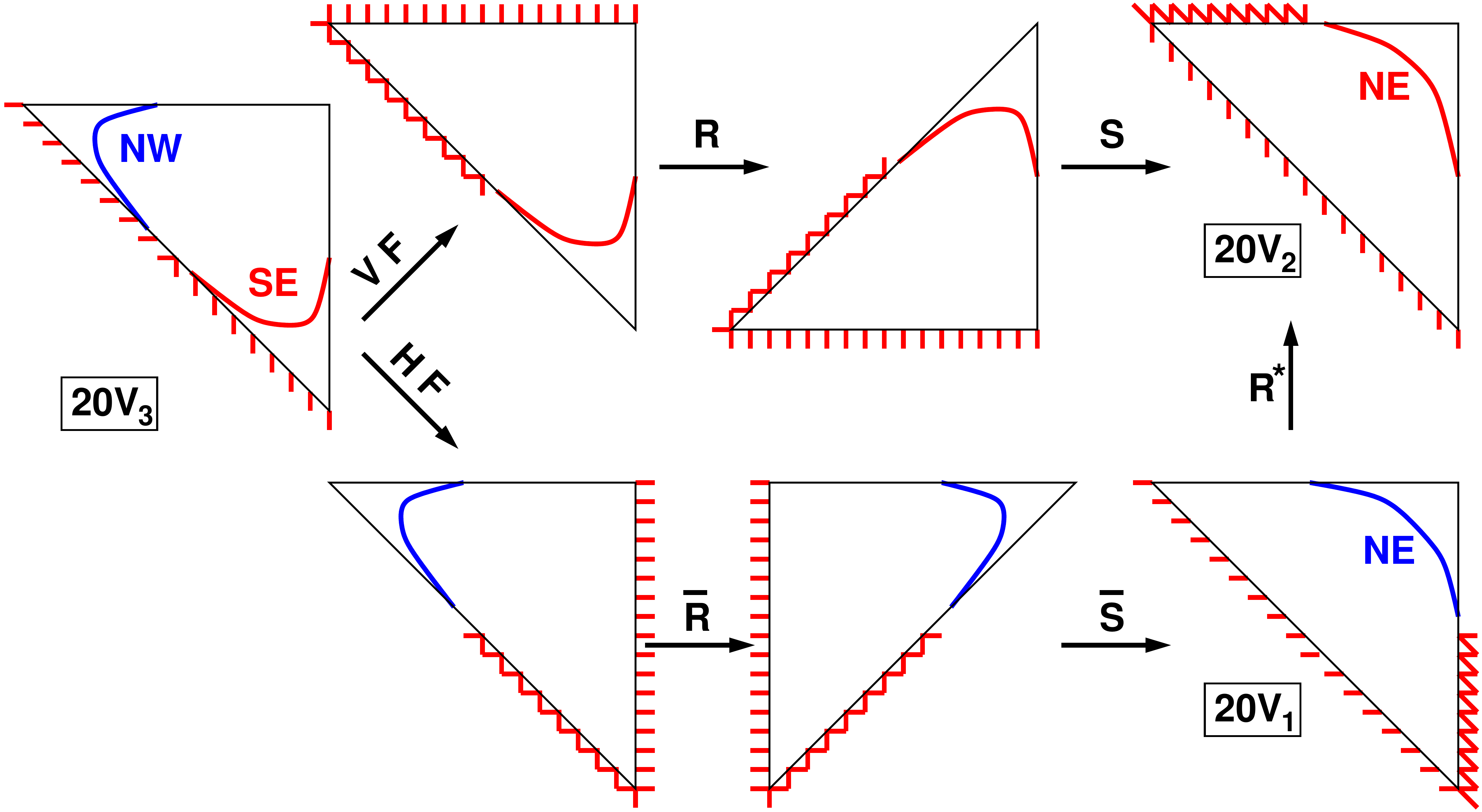}
\end{center}
\caption{\small Transformations of the 20V${}_3$ model into the 20V${}_2$ and 20V${}_1$ models. We indicate how the SE branch of the arctic curve (in red on the left)
for the 20V${}_3$ is mapped onto the NE branch of the 20V${}_2$ model (top right), as well as how the NW branch (in blue on the left) is mapped 
onto the NE branch of the 20V${}_1$ model (bottom right). The latter two are related via the diagonal reflection $\bf R^*$.}
\label{fig:branches}
\end{figure}

\subsubsection{The SE branch}

Let us consider the transformations $\bf VF$, $\bf R$ and $\bf S$ applied successively to the configurations of the 20V${}_3$ model (see Fig. \ref{fig:shear}).
First note that if we apply the transformation $\bf VF$ alone (see Fig. \ref{fig:shear} (b)), we create a domain empty of path steps in the SE corner of the triangle $\cT_m$.
The large $m$ asymptotic boundary of this new crystalline phase tends to another branch of the arctic curve of the 20V${}_3$ model, which we call South-East branch (SE),
and which, in the original 20V${}_3$ model, is the boundary of a crystalline phase with only vertical paths. 
As indicated in Fig. \ref{fig:branches} (top), under the further transformation $\bf S\circ \bf R$, this crystalline region is mapped to another region empty of paths situated in the NE corner of the 
20V${}_2$ configuration domain. So the SE branch of the arctic curve of the 20V${}_3$ model is the inverse image 
(i.e. under $\bf R \circ \bf S^{-1}$, now acting on the rescaled domain) of the 
NE branch of the arctic curve of the 20V${}_2$ model, with weights $w_{\pi(i)}[t,z,w]$ as noted in Sect. \ref{transfosec}.

The latter is obtained by applying the tangent method in the {\it same} way as we did for the NE branch of the 20V${}_3$ model. To this end, we need estimates of the one-point function $H^{(2)}_{m,k}$ and of the path partition function 
$Y^{(2)}_{k,r}$ of the 20V${}_2$ model. The one point function $H^{(2)}_{m,k}$ is just a reinterpretation of that of the 20V${}_3$ model, $H_{m,k}^{20V_{3}}$.

Under the transformation $\bf S\circ \bf R\circ \bf VF$ the original exit point $k$
is mapped onto $m+1-k$ (as it is measured upside-down). We deduce that $H_{m,k}=H_{m,m+1-k}^{20V_{3}}$,
and the large $n$ estimate $H^{(2)}_{2n,n\kappa_2}\sim H_{2n,(2-\kappa_2)n}^{20V_{3}}$ is dominated by the saddle-point of the action $S_V^{(2)}(2-\kappa_2,\xi,\tau)$, with $S_V$ as in \eqref{SV}, leading to $\kappa_2[\xi]=2-\kappa[\xi]$ 
with $\kappa[\xi]$ as in \eqref{kappagen}. As before this is an implicit relation between $\kappa_2$ and $\tau$, upon
parameterizing $\tau=\tau[\xi]$ as in \eqref{taugen}.

On the other hand, the path partition function $Y^{(2)}_{k,r}$ is obtained
by substituting the new path weights $w_{\pi(i)}[t,z,w]$ into the expression for the original path partition function $Y_{m+1-k,r}$.
The following shows how to implement the transformation \eqref{netose} on the 20V weights. 
\begin{lemma}\label{hatlem}
The transformation $\bf S\circ \bf R \circ \bf VF$ on the weights of the 20V model is implemented 
by the following involutive transformation of the parameters 
$(\lambda,\mu,\xi)\mapsto (\hat \lambda,\hat \mu,\hat \xi)$, where:
\begin{equation} \hat \lambda=\pi -\frac{\lambda+\eta+\mu}{2}, \quad \hat \mu=\pi- \frac{3\lambda+\eta-\mu}{2}, \quad \hat \xi=-\xi ,
\end{equation}
namely we have $\omega_i[z,t,w,\xi]\big\vert_{(\lambda,\mu,\xi)\to (\hat \lambda,\hat \mu,\hat \xi)}=\omega_{\pi(i)}[t,z,w,\xi]$ with the permutation $\pi=(01)(24)$.
\end{lemma}
\begin{proof} By inspection, using the weights \eqref{eq:explweights}.
\end{proof}
Let us use the notation $\hat f=f\vert_{(\lambda,\mu,\xi)\to  (\hat \lambda,\hat \mu,\hat \xi)}$. The total action
is $S_V^{(2)}(\kappa_2,\tau,\xi)+{\hat S_P}(\kappa_2,\rho_2,p_3,p_4,p_5,p_6)$ with $S_P$ as in \eqref{eq:explS}, and where applying the involution $S_P\to \hat S_P$ of Lemma \ref{hatlem} simply amounts to permute the weights $\omega_i\to\omega_{\pi(i)}$ in the expressions \eqref{eq:alphas}. The saddle point equations \eqref{sapoS} become:
$$\partial_{p_i}\, \hat S_P=0\quad (i=3,4,5,6), \qquad \partial_{\kappa_2} \hat S_P+{\rm Log}(\tau \al_1)=0 .$$
From the expression \eqref{taugen}, and the fact that the involution interchanges the weights $\omega_0[\xi]\leftrightarrow \omega_1[\xi]$, we see  that $\hat \tau=\tau^{-1}$ as well as $\hat \al_1=\al_1^{-1}$. 
This allows to rewrite the last equation as 
$\partial_{\kappa_2} \hat S_P-{\rm Log}(\hat \tau \hat \al_1)=0$ henceforth the new saddle point equations are simply the
 transform of the original ones \eqref{sapoS}. The new solution is therefore the transform of the previous one, and in particular $s_2=\kappa_2/\rho_2=\hat s[\xi]$ with $s[\xi]$ given by \eqref{slopegen}. Defining as before the inverse slope $A_2=s_2^{-1}$, 
we end up with the family of tangent lines $y+A_2[\xi] x-\kappa_2[\xi]=0$, with  $\kappa_2[\xi]=2-\kappa[\xi]$ and $A_2[\xi]=\hat s[\xi]^{-1}=\hat A[\xi]$, and the NE branch of the arctic curve for the 20V${}_2$ model
is given in parametric form by 
\begin{equation}\label{arctic2} X[\xi]=\frac{\kappa_2'[\xi]}{A_2'[\xi]} =-\frac{\kappa'[\xi]}{\hat A'[\xi]}, \quad Y[\xi]= \kappa_2[\xi]-A_2[\xi]\,\frac{\kappa_2'[\xi]}{A_2'[\xi]}=2-\kappa[\xi]+\hat A[\xi]\,\frac{\kappa'[\xi]}{\hat A'[\xi]} . 
\end{equation}
Note finally that the range of $\xi$ is mapped onto $\xi \in [-\frac{\lambda-\eta+\mu}{2},0]$.
We must finally apply the inverse transformation ${\bf R \circ \bf S^{-1}}:\, (x,y)\mapsto (x,2-x-y)$ (in rescaled variables) to recover the SE branch for the 20V${}_3$ model, resulting in the following.
\begin{thm}\label{SEbranchthm}
The SE branch of the arctic curve for the 20V${}_3$ model is given in parametric form by:
$$X_{SE}[\xi]=-\frac{\kappa'[\xi]}{\hat A'[\xi]} , \quad Y_{SE}[\xi]= \kappa[\xi]-(\hat A[\xi]-1)\,\frac{\kappa'[\xi]}{\hat A'[\xi]} \qquad \xi \in \left[-\frac{\lambda-\eta+\mu}{2},0\right] ,$$
where the hat refers to the transformation of Lemma \ref{hatlem}, and $\kappa[\xi],A[\xi]=s[\xi]^{-1}$ are as in \eqref{kappagen} and \eqref{slopegen}.
\end{thm}

\subsubsection{The NW branch}

Applying the transformation $\bf H F$ to configurations of the 20V${}_3$ model creates a domain empty of path steps in the NW corner of the triangle $\cT_m$ (see Fig. \ref{fig:branches}).
The large $m$ asymptotic boundary of this new crystalline phase tends to yet another branch of the arctic curve of the 20V${}_3$ model, which we call North-West branch (NW),
and which corresponds in the original 20V${}_3$ model to the boundary of a crystalline phase where all paths are horizontal.
Recall that the transformation $\bf \bar S \circ \bf \bar R\circ \bf H F$ on 20V${}_3$ configurations produces the 20V${}_1$ configurations
with transformed weights $\omega_{\bar \pi(i)}$, which are the diagonal reflection $\bf R^*$ of the configurations of the 20V${}_2$ model with weights $\omega_{\pi(i)}$.
The NW branch of the arctic curve for the 20V${}_3$ model is therefore the inverse image (i.e. under $\bf \bar R\circ \bf \bar S^{-1}$, now acting on the rescaled domain) 
of the NE branch of the arctic curve of the 20V${}_1$ model, with weights $\omega_{\bar\pi(i)}[z,w,t]$. The latter can be directly derived from that of the 20V${}_2$ model, by
applying the diagonal reflection $\bf R^*$.
\begin{lemma}\label{barlem}
The transformation $\bf \bar S\circ \bf \bar R \circ \bf HF$ on the weights of the 20V model is implemented by the following transformation of the parameters 
$(\lambda,\mu,\xi)\mapsto (\bar \lambda,\bar \mu,\bar \xi)$, where:
\begin{equation} \bar \lambda=\pi -\frac{\lambda+\eta-\mu}{2}, \quad \bar \mu=\pi- \frac{3\lambda+\eta+\mu}{2}, \quad \bar \xi=-\xi ,
\end{equation}
namely we have $\omega_i[z,t,w,\xi]\big\vert_{(\lambda,\mu,\xi)\to (\bar \lambda,\bar \mu,\bar \xi)}=\omega_{\bar\pi(i)}[w,t,z,\xi]$  with the permutation $\bar\pi=(16)(25)$.
\end{lemma}
\begin{proof} By inspection, using the weights \eqref{eq:explweights}.
\end{proof}
As a consequence, we may deduce the NE branch of the 20V${}_1$ model from that of the 20V${}_2$ model by the diagonal reflection $\bf R^*$.
\begin{cor}\label{barcor}
The transformation $\bf R^*$ on the weights of the 20V${}_2$ model is implemented by the following transformation of the parameters 
$(\lambda,\mu,\xi)\mapsto (\lambda^*, \mu^*, \xi^*)$, where:
\begin{equation} \lambda^*=\lambda, \quad \mu^*=-\mu, \quad  \xi^*=\xi ,
\end{equation}
\end{cor}

Using the notations $\bar f=f\vert_{(\lambda,\mu,\xi)\to  (\bar\lambda,\bar \mu,\bar \xi)}$ and  $f^*=f\vert_{(\lambda,\mu,\xi)\to  (\lambda^*,\mu^*,\xi^*)}=f\vert_{\mu\to-\mu}$,
we conclude that $\kappa_1[\xi]=2-\kappa^*[\xi]$, whereas $A_1[\xi]=s_1[\xi]^{-1}=\tilde A_2[\xi]=\bar A[\xi]$ (where the index $1$ now refers to the 20V${}_1$ model). Finally, using the mapping ${\bf  R^*}:\, (x,y)\mapsto (y-2,2+x)$ in rescaled coordinates, the NE branch of the arctic curve for the 20V${}_1$ model
is given in parametric form by 
\begin{equation}\label{arctic1} X[\xi]=\kappa_1[\xi]-A_1[\xi]\,\frac{\kappa_1'[\xi]}{A_1'[\xi]}-2=-\kappa^*[\xi]+\bar A[\xi]\,\frac{{\kappa^*}'[\xi]}{\bar A'[\xi]}, \quad Y[\xi]=2+\frac{\kappa_1'[\xi]}{A_1'[\xi]} =2-\frac{{\kappa^*}'[\xi]}{\bar A'[\xi]}  .
\end{equation}
Note finally that the range of $\xi$ is mapped onto $\xi \in -[\frac{\lambda-\eta-\mu}{2},0]$.
We must finally apply the inverse transformation ${\bf \bar R \circ \bf {\bar S}^{-1}}:\,  (x,y)\mapsto (-x-y,y)$ (in rescaled variables) to recover the NW branch for the 20V${}_3$ model, resulting in the following.
\begin{thm}\label{NWbranchthm}
The NW branch of the arctic curve for the 20V${}_3$ model is given in parametric form by:
$$X_{NW}[\xi]={\kappa^*}[\xi]-2-(\bar A[\xi]-1)\,\frac{{\kappa^*}'[\xi]}{\bar A'[\xi]} , \quad Y_{NW}[\xi]=2-\frac{{\kappa^*}'[\xi]}{\bar A'[\xi]} \qquad \xi \in -\left[\frac{\lambda-\eta-\mu}{2},0\right], $$
where the bar refers to the transformation of Lemma \ref{barlem}, ${\kappa^*}[\xi]=\kappa[\xi]\vert_{\mu\to -\mu}$,
and $\kappa[\xi],A[\xi]=s[\xi]^{-1}$ are given by \eqref{kappagen} and \eqref{slopegen}.
\end{thm}

\subsection{Frozen domains}\label{frosec}

Frozen domains in 20V configurations correspond to the following seven different phases:
\begin{itemize}
\item $F_0$: empty region (no path)
\item $F_1$: region filled with parallel horizontal paths (all horizontal edges occupied by path steps)
\item $F_2$: region filled with parallel vertical paths (all horizontal edges occupied by path steps)
\item $F_3$: region filled with parallel diagonal/horizontal paths (all horizontal and diagonal edges occupied by path steps)
\item $F_4$: region filled with parallel diagonal/vertical paths (all vertical and diagonal edges occupied by path steps)
\item $F_5$: region fully packed with paths (all horizontal, diagonal and vertical edges occupied by path steps)
\item $F_6$: region filled with parallel diagonal paths (all diagonal edges occupied by path steps)
\end{itemize}

By the very definition of the tangent method we applied, it is clear that the NE branch of the arctic curves in the previous sections separates the liquid (disordered) phase below the curve from the empty frozen phase $F_0$ above the curve.
The frozen region below the NW branch is by continuity from the boundary condition of type $F_1$, while that below the SE branch is of type $F_2$. 

We now argue that there exist phase separation segments originating at the center point $(-1,1)$ of the diagonal boundary.

\begin{figure}
\begin{center}
\begin{minipage}{0.5\textwidth}
        \centering
        \includegraphics[width=5.9cm]{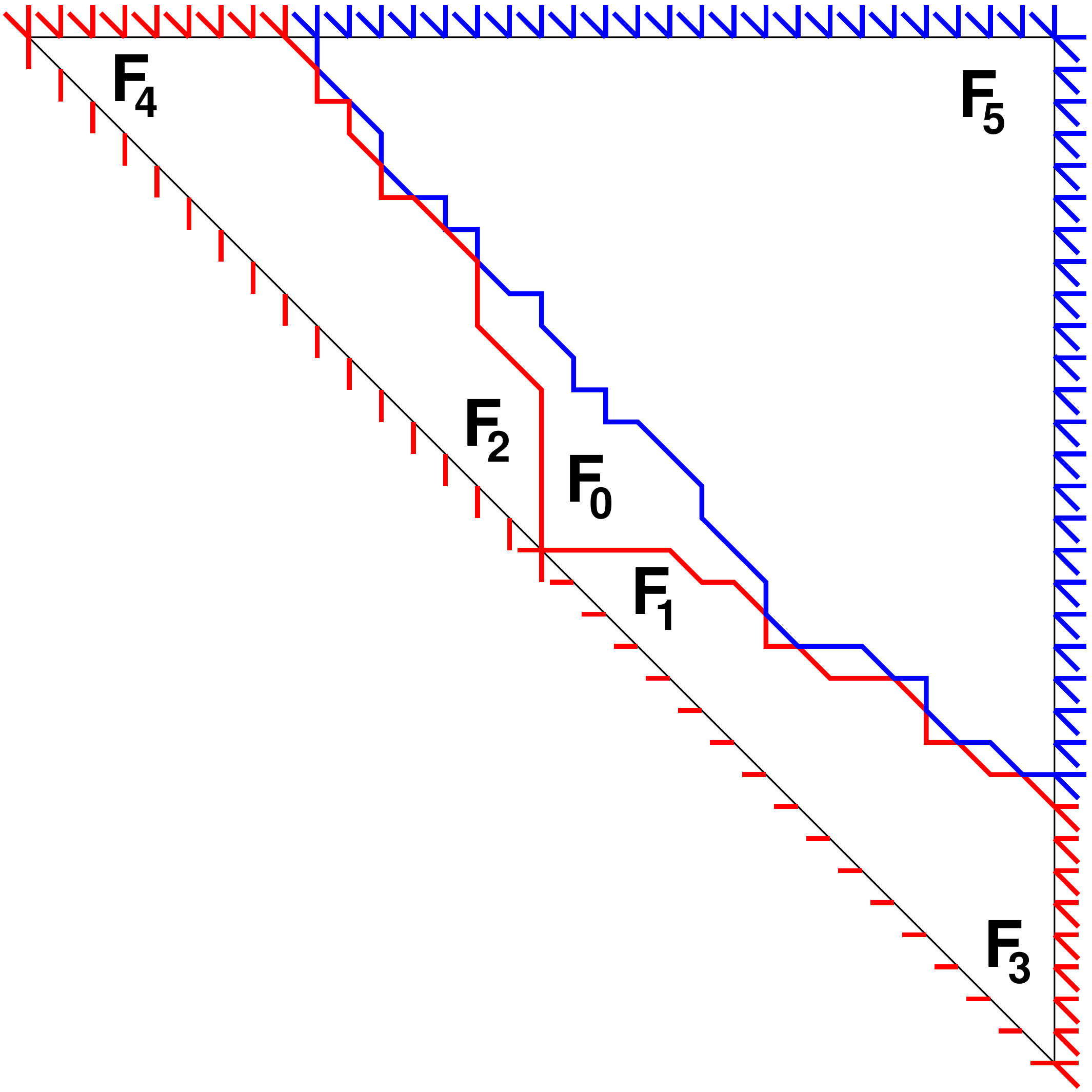} % first figure itself
        %\caption{first figure}
    \end{minipage}\hfill
    \begin{minipage}{0.5\textwidth}
        \centering
        \includegraphics[width=5.9cm]{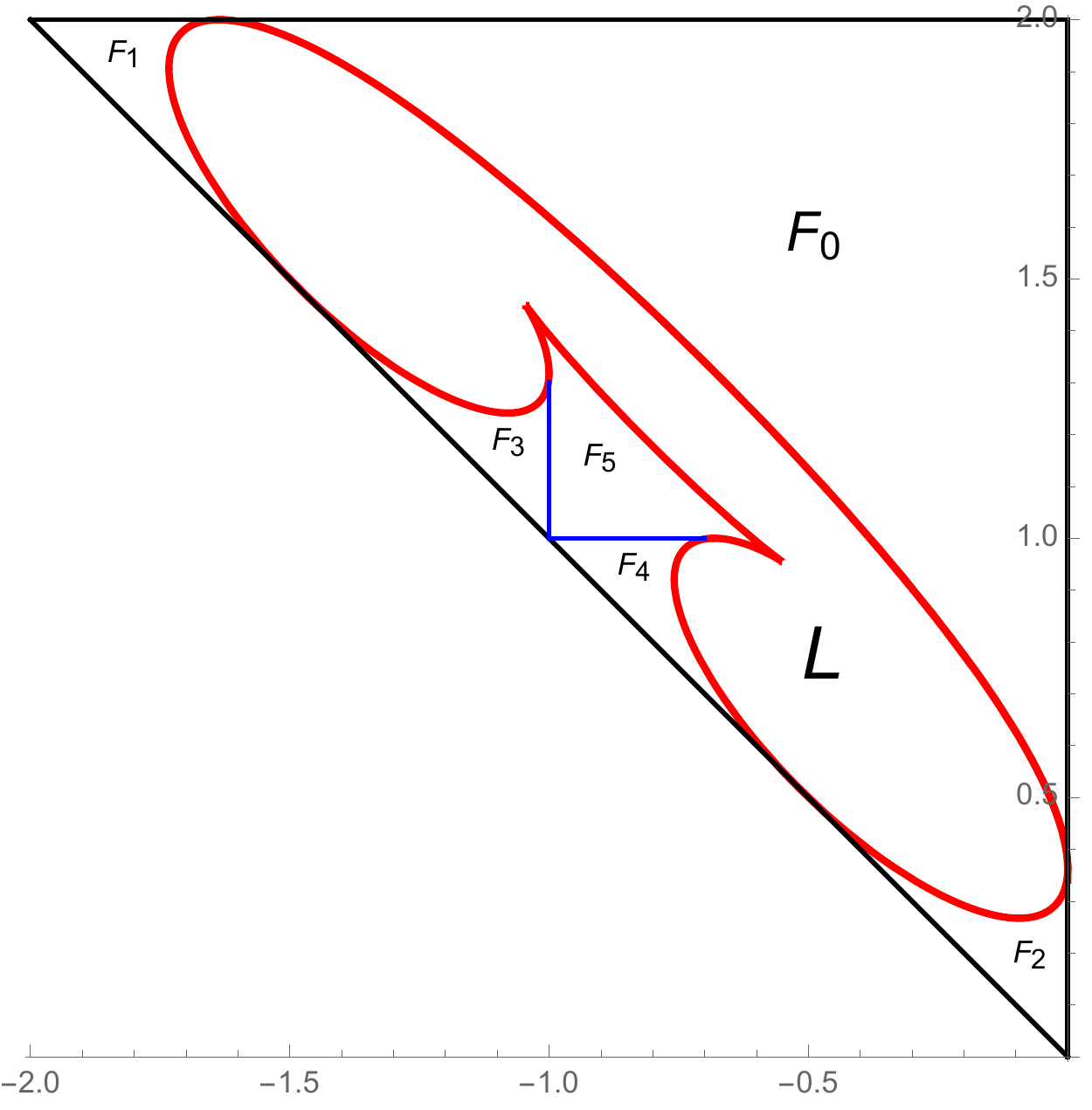}  % second figure itself
        %\caption{second figure}
    \end{minipage}
\end{center}
\caption{\small Emergence of frozen regions near the diagonal S boundary. Left: we apply a total flip to the 20V${}_3$ model (of even size here), and show the corresponding boundary conditions on the path configurations. The rightmost path linking the N boundary to the diagonal and the topmost path from the diagonal to the E boundary (in red), as well as the bottom-most path connecting the N boundary to the E boundary (in blue) delimit a frozen empty region (type $F_0$), and two frozen zones with vertical (type $F_2$) and horizontal (type $F_1$) paths. We also indicate the frozen NE corner which is fully packed with paths (type $F_5$), NW corner (type $F_4$) and SE corner (type $F_3$). Right: after flipping back all edges, predicted complete arctic curve with six frozen regions and the central liquid one (L). We expect three more portions of arctic curve in addition to the NW, NE, SE ones.}
\label{fig:frozen}
\end{figure}

Starting from the 20V${}_3$ model configurations, let us apply the {\it total flip} i.e. replace each path edge by an empty edge and vice-versa. The boundary conditions are then changed into those of Fig. \ref{fig:frozen} (left). There are $2m$ paths entering along the N boundary, $n$ of which exit along the SW diagonal (above the central point) while the remaining ones exit on the (top 3/4 of) vertical E boundary. There are also $n$ paths entering along the SW diagonal boundary (below the central point), exiting on the (bottom 1/4 of) vertical E boundary.

In particular, the $n$ leftmost paths starting on the N boundary must exit at the $n$ topmost positions along the diagonal SW boundary. The rightmost of those ends up at the center on the diagonal. Similarly the $n$ paths that enter the SE diagonal beyond the center exit at the $n$ bottom-most positions along the vertical E boundary. In particular, the topmost one starts at the center. This clearly creates an empty region (of type $F_0$, see Fig.\ref{fig:frozen}(left)), namely in the NE corner (positive quadrant) from the center of the diagonal SW boundary. Moreover the bottom-most path connecting the N to the E boundaries closes up this empty region from above. 
The vertical and horizontal boundaries of this empty region separate it from regions with parallel paths (vertical on the left, type $F_2$ and horizontal on the right, type $F_1$), as dictated by the new SW boundary conditions. 
We have also indicated in Fig. \ref{fig:frozen} (left) the flipped corner phases (respectively $F_4$, $F_5$, $F_3$ delimited respectively by the NW, NE, SE branches of the previous sections). 

Applying the total flip back, and noting that this interchanges regions $F_0\leftrightarrow F_5$, $F_1\leftrightarrow F_4$ 
and $F_2\leftrightarrow F_3$, this produces a corner central region of type $F_5$, NE of the center of the SW diagonal boundary, with a region of type $F_3$ on its left, and of type $F_4$ on its right. This shows that there are missing portions of the arctic curve (that encompasses the liquid phase), which should look like Fig. \ref{fig:frozen} (right).

The next section is devoted to examples. In particular, in the so-called free fermion case corresponding to $\eta=\pi/4$,
we expect all curves to be analytic, and the tangent method predictions match the expected phase structure displayed in Fig. \ref{fig:frozen} (right).

\section{Examples}\label{exsec}

This section is devoted to explicit examples of arctic curves of the 20V${}_3$ model with various choices of
the weights $\omega_i$.

\subsection{Example I: The uniform (combinatorial) case}\label{unisec}

\begin{figure}
\begin{center}
\includegraphics[width=8.cm]{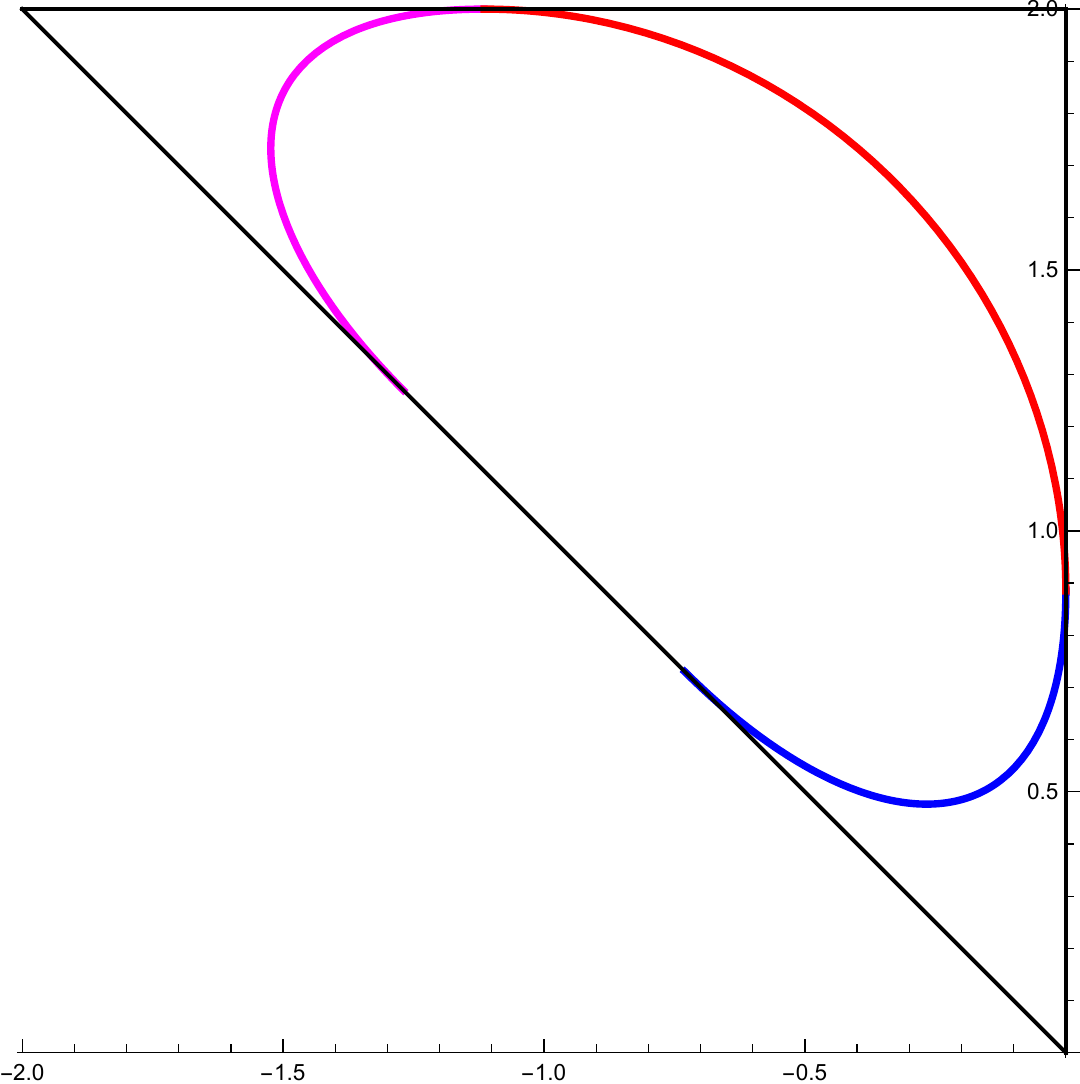}
\end{center}
\caption{\small The three predicted (NE, SE, NW) portions of the arctic curve of the 20V${}_3$ model 
with uniform weights respectively in red, blue, and magenta. }
\label{fig:uniform}
\end{figure}

As mentioned above, the case of the 20V${}_3$ model with uniform weights corresponds to the choice of 
parameters:  $\eta=\frac{\pi}{8}$, $\lambda=\frac{5\pi}{8}$ and $\mu=0$. We have represented in Fig. \ref{fig:uniform} the corresponding three portions of arctic curve as predicted by Theorems \ref{NEbranchthm}, \ref{SEbranchthm} and \ref{NWbranchthm}. We now argue that the three (NE,SE,NW) portions of arctic curve are pieces of the {\it same} algebraic curve. 

\begin{figure}
\begin{center}
\begin{minipage}{0.5\textwidth}
        \centering
        \includegraphics[width=5.9cm]{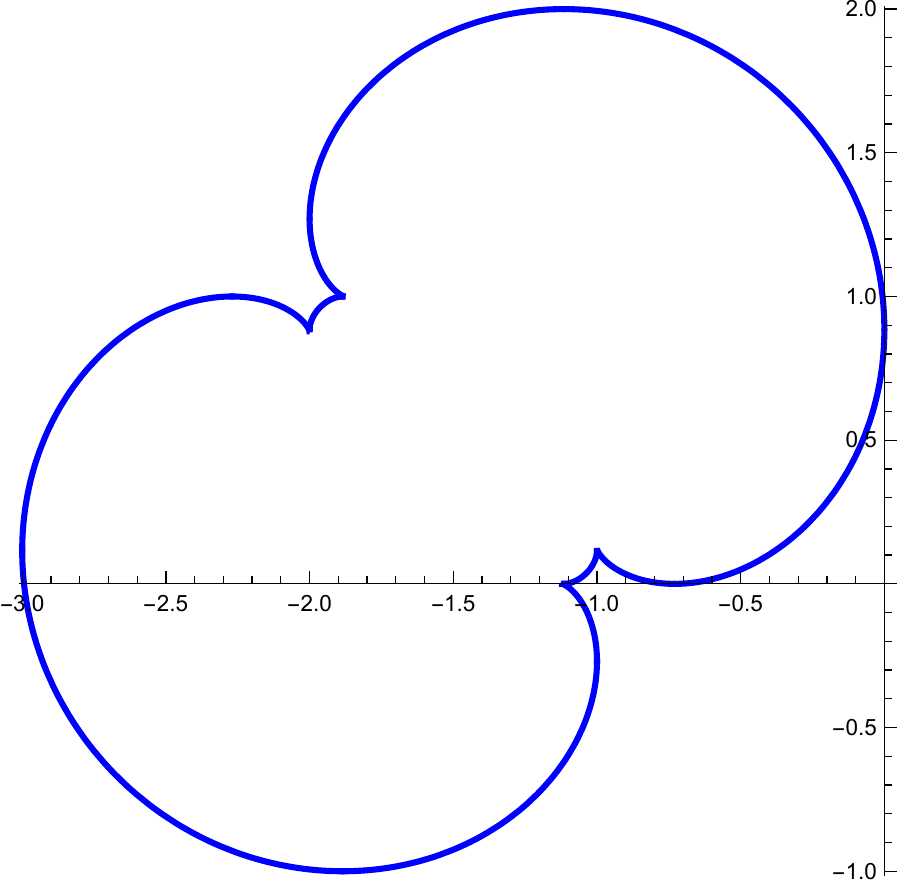} % first figure itself
        %\caption{first figure}
    \end{minipage}\hfill
    \begin{minipage}{0.5\textwidth}
        \centering
        \includegraphics[width=5.9cm]{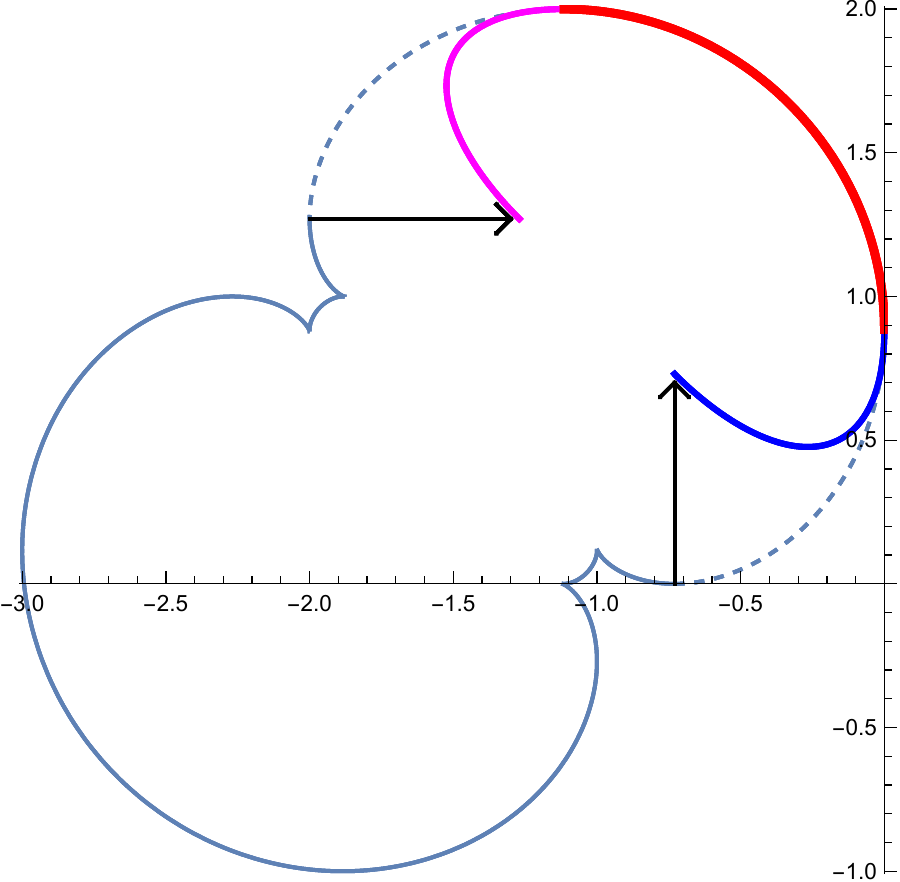}  % second figure itself
        %\caption{second figure}
    \end{minipage}
\end{center}
\caption{\small Left: the real algebraic curve containing the NE portion of the arctic curve of the 
uniform 20V${}_3$ model. Right: the NE (in red), SE (in blue) and SW (in magenta) portions of the 20V${}_3$ arctic curve: the latter two are obtained from the algebraic curve by the indicated shearing transformations.}
\label{fig:unifplus}
\end{figure}

First note that $A[\xi]=\cot(2\xi)$ is an odd function of $\xi$.
To get the $SE$ branch, we must apply Theorem \ref{SEbranchthm} with the new values $\hat\lambda,\hat\mu,\hat\xi$ of $\lambda,\mu,\xi$
of Lemma \ref{hatlem}. Similarly, the $NW$ branch involves applying Theorem \ref{NWbranchthm} with the
new values $\bar\lambda,\bar\mu,\bar\xi$ of $\lambda,\mu,\xi$ of Lemma \ref{barlem}. However, one easily computes
$\hat \lambda=\bar \lambda=\lambda^*= \lambda=\frac{5\pi}{8}$, and $\hat \mu=\bar \mu=\mu^*=\mu=0$, while $\hat\xi=\bar\xi=-\xi$ and $\xi^*=\xi$,
henceforth we have $A_2[\xi]=\hat A[\xi]=A[-\xi]=-A[\xi]$ and $A_1[\xi]=\bar A[\xi]=A[-\xi]=-A[\xi]$, as well as $\kappa^*[\xi]=\kappa[\xi]$.
Using these values, the three branches can be written as
\begin{eqnarray*}
NE:&&X_{NE}[\xi]=\frac{\kappa'[\xi]}{A'[\xi]} , \quad Y_{NE}[\xi]= \kappa[\xi]-A[\xi]\,\frac{\kappa'[\xi]}{A'[\xi]} \qquad \xi \in [0,\frac{\pi}{4}].\\
SE:&&X_{SE}[\xi]=\frac{\kappa'[\xi]}{A'[\xi]} , \quad Y_{SE}[\xi]= \kappa[\xi]-(A[\xi]+1)\,\frac{\kappa'[\xi]}{A'[\xi]} \qquad \xi \in \left[-\frac{\pi}{4},0\right]\\
NW:&&X_{NW}[\xi]= \kappa[\xi]-2-(A[\xi]+1)\,\frac{\kappa'[\xi]}{A'[\xi]} , \quad Y_{NW}[\xi]=2+\frac{\kappa'[\xi]}{A'[\xi]} \qquad \xi \in \left[-\frac{\pi}{4},0\right]
\end{eqnarray*}
We see that the SE branch corresponds to the shearing $(x,y)\mapsto (x,y-x)$ of the analytic continuation of the NE branch to the domain $\xi\in\left[-\frac{\pi}{4},0\right]$, while the NW branch is the image of the same analytic continuation
of the NE branch under $(x,y)\mapsto (y-x-2,x+2)$, which is the composition of the diagonal reflection wrt the line $y=x+2$, i.e. $(x,y)\mapsto(y-2,x+2)$, and the horizontal shear $(x,y)\mapsto (2+x-y,y)$.

We finally note that the NE branch is actually a portion of a real algebraic curve of degree 10.  
Expressed in the variables
$(x,y)=(X_{NE}+\frac{3}{2},Y_{NE}-\frac{1}{2})$, it reads: 
\begin{eqnarray*}&&2^6  \times 3^{11} (x^2 + y^2)^5 - 2^4 \times 3^9 \times 67 (x^2 + y^2)^4 \\
&&\qquad - 
 2^2 \times 3^6 (x^2 + y^2)^2 (13\times 83 (x^2 + y^2) + 2^6 \times 5^3 \sqrt{3} \,x y) \\
 &&\qquad - 
 3^2 \times 17\times 9323 (x^2 + y^2)^2 - 
 2^{10} 3^2 5^5 x^2 y^2 - 2^6 3^4 5^2 11 \sqrt{3}\, x y (x^2 + y^2)  \\
 &&\qquad - 2^{10} \times 3 \times 41 (x^2 + y^2)- 
 2^7 \times 3 \times 23 \times 241 \sqrt{3}\, x y - 2^{10} \times 13^2=0
 \end{eqnarray*}
 This curve is clearly $x\leftrightarrow y$ symmetric, which means that the curve is symmetric wrt the line $y=x+2$ in the original coordinates. We conclude that the NW branch is the horizontal shearing $(x,y)\mapsto (2+x-y,y)$ of the portion of the algebraic curve corresponding to the diagonal reflection of the analytic continuation to $\xi\in [-\frac{\pi}{4},0]$. One can in fact check that this reflected piece is the analytic continuation of the algebraic curve to $\xi\in [\frac{\pi}{4},\frac{\pi}{2}]$.

This algebraic curve is represented in Fig. \ref{fig:unifplus} (Left), and (Right) together with the (NE,SE,NW) portions
of 20V${}_3$ arctic curve, the latter two obtained as the indicated shear transformations.

\subsection{Example II: The free fermion case}
\begin{figure}
\begin{center}
\begin{minipage}{0.33\textwidth}
        \centering
        \includegraphics[width=4.cm]{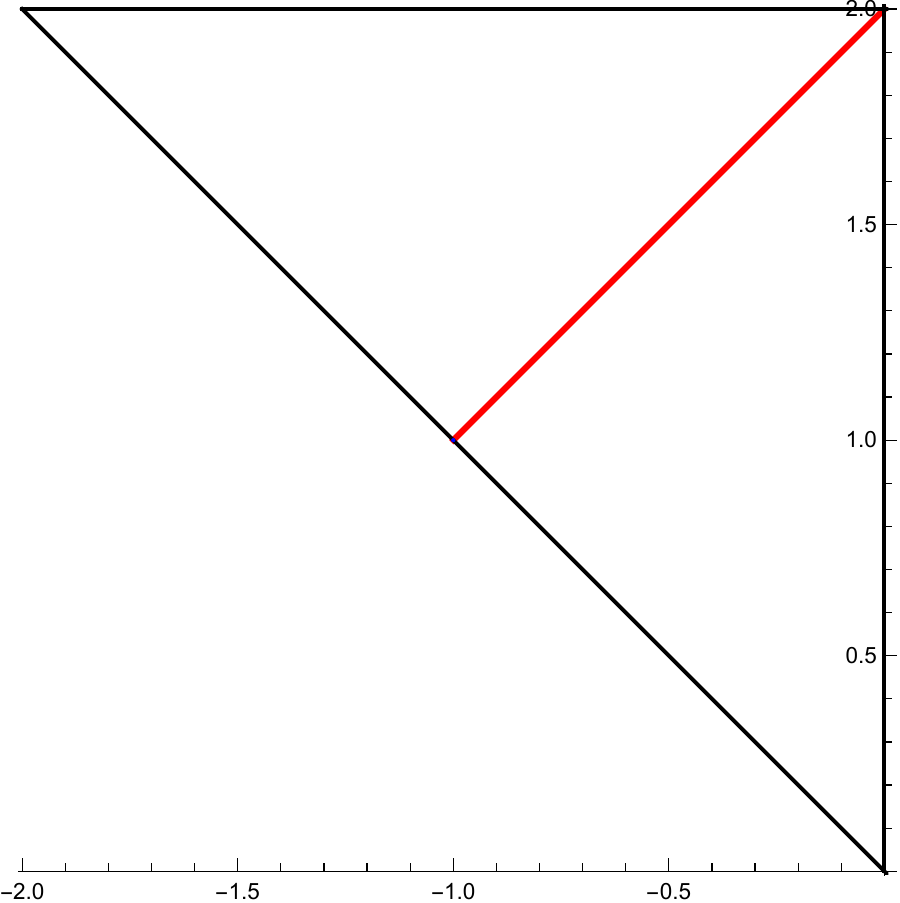} % first figure itself
        %\caption{first figure}
    \end{minipage}\hfill
    \begin{minipage}{0.33\textwidth}
        \centering
        \includegraphics[width=4.cm]{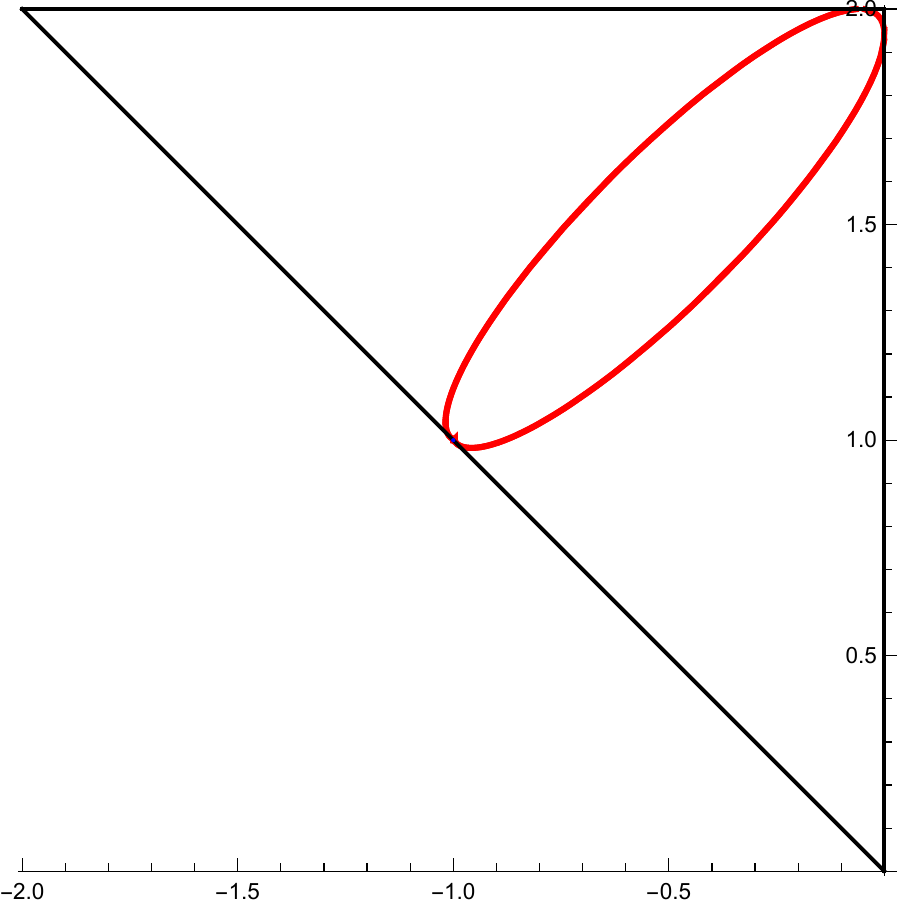}  % second figure itself
        %\caption{second figure}
    \end{minipage}
     \begin{minipage}{0.33\textwidth}
        \centering
        \includegraphics[width=4.cm]{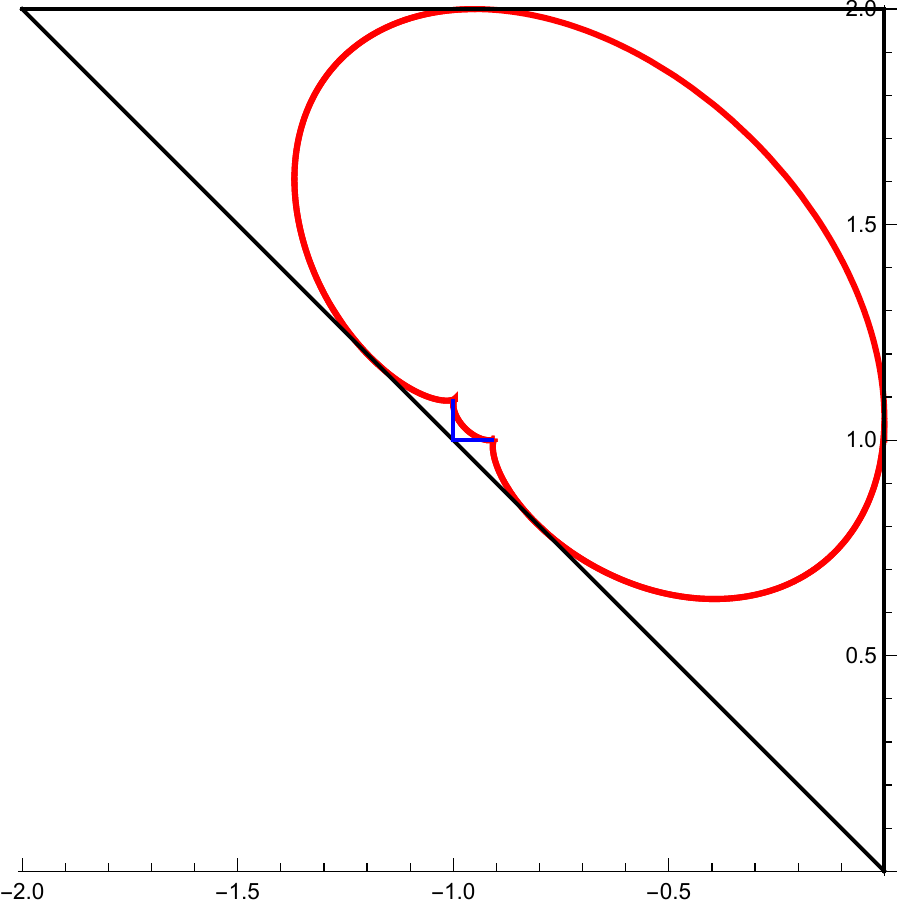}  % second figure itself
        %\caption{second figure}
    \end{minipage}\\
\begin{minipage}{0.33\textwidth}
        \centering
        \includegraphics[width=4.cm]{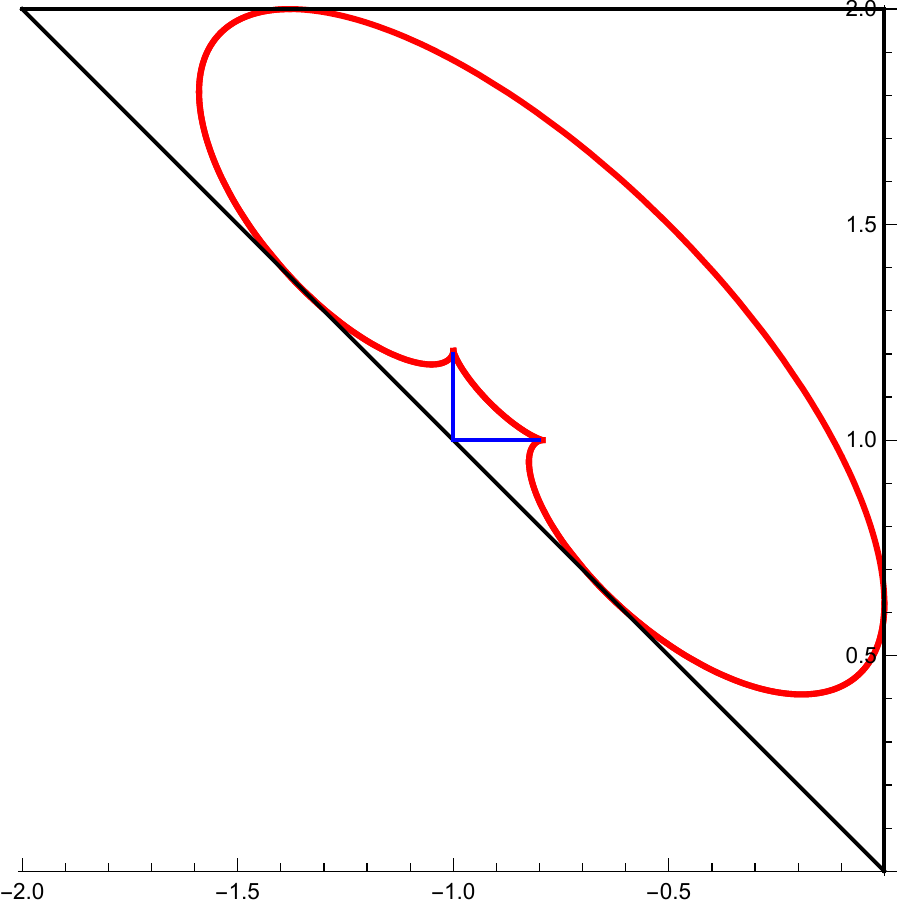} % first figure itself
        %\caption{first figure}
    \end{minipage}\hfill
    \begin{minipage}{0.33\textwidth}
        \centering
        \includegraphics[width=4.cm]{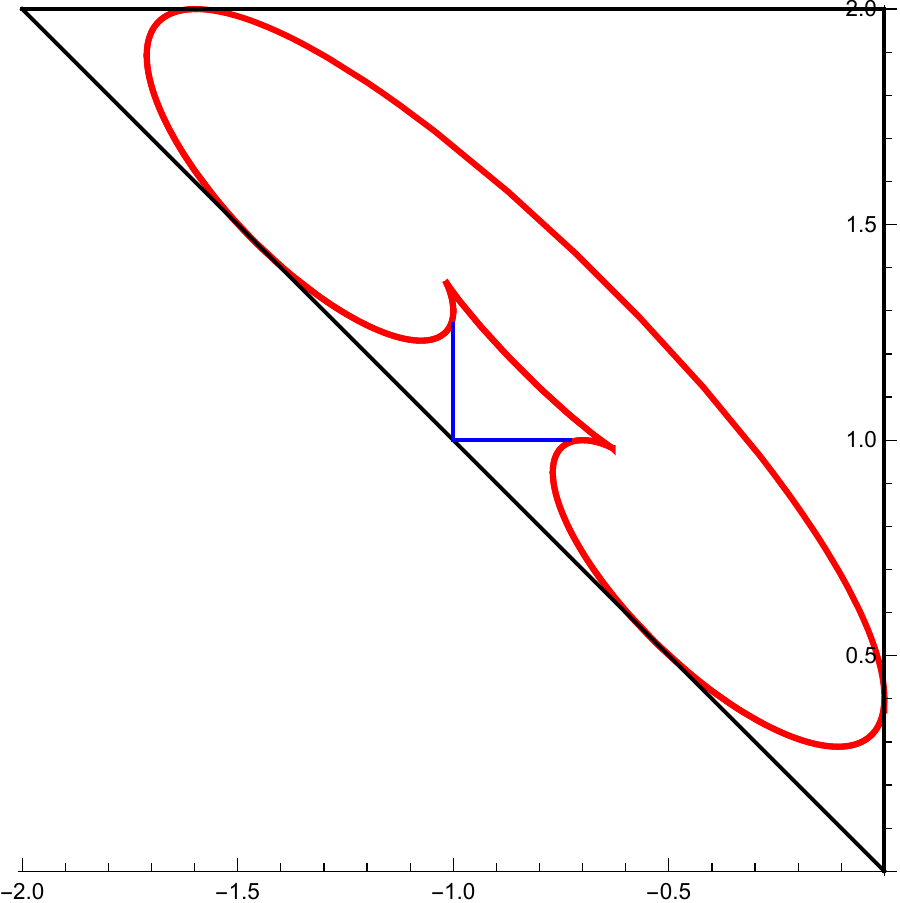}  % second figure itself
        %\caption{second figure}
    \end{minipage}
     \begin{minipage}{0.33\textwidth}
        \centering
        \includegraphics[width=4.cm]{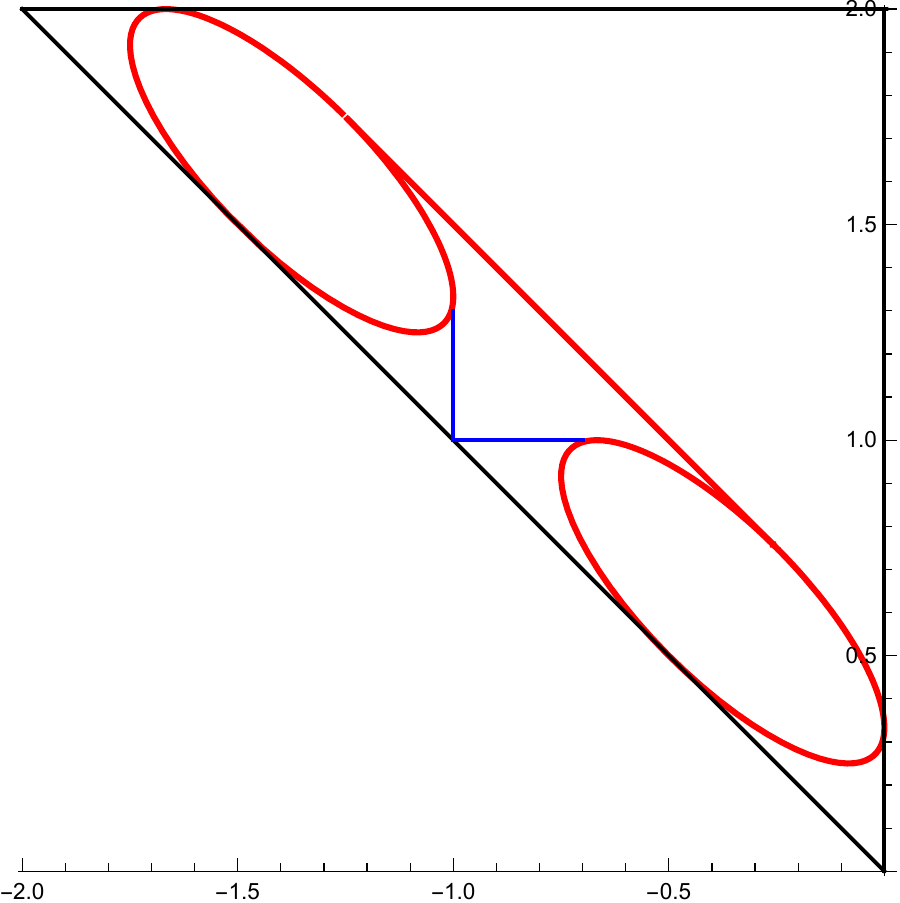}  % second figure itself
        %\caption{second figure}
    \end{minipage}
\end{center}
\caption{\small Arctic curves of the 
free fermion ($\eta=\pi/4,\mu=0$) 20V${}_3$ model. We represent from left to right and top to bottom the curves corresponding to the values $\lambda=-\frac{\pi}{4},-\frac{\pi}{4}+.01,-\frac{\pi}{8},0,\frac{\pi}{8},\frac{\pi}{4}$. }
\label{fig:lam}
\end{figure}

The free fermion case $\eta=\frac{\pi}{4}$ corresponds to the situation when all three 6V models are themselves in their free fermion phase. We expect on general grounds the arctic curve to be analytic, and we indeed check that the three portions NE,SE,NW are analytic continuations of each other. By plotting the entire curve, we can see here the ``missing" portions of arctic curve that we have no prediction for in the case of general weights.

\begin{figure}
\begin{center}
\begin{minipage}{0.33\textwidth}
        \centering
        \includegraphics[width=4.1cm]{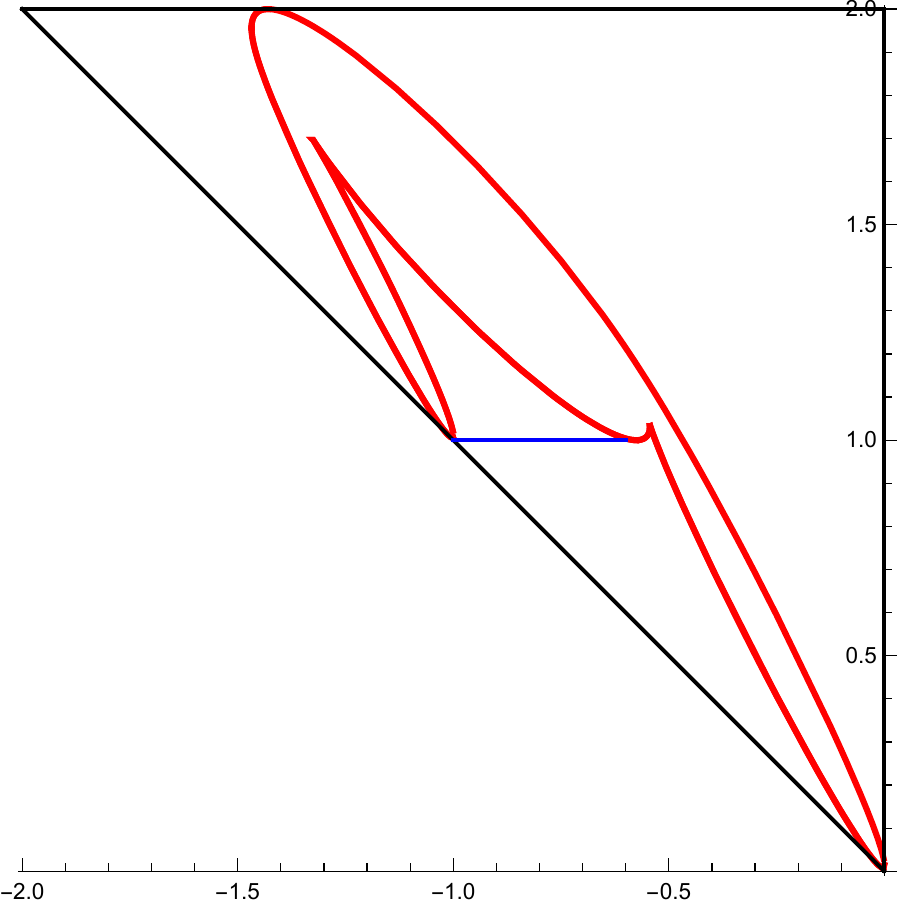} % first figure itself
        %\caption{first figure}
    \end{minipage}\hfill
    \begin{minipage}{0.33\textwidth}
        \centering
        \includegraphics[width=4.1cm]{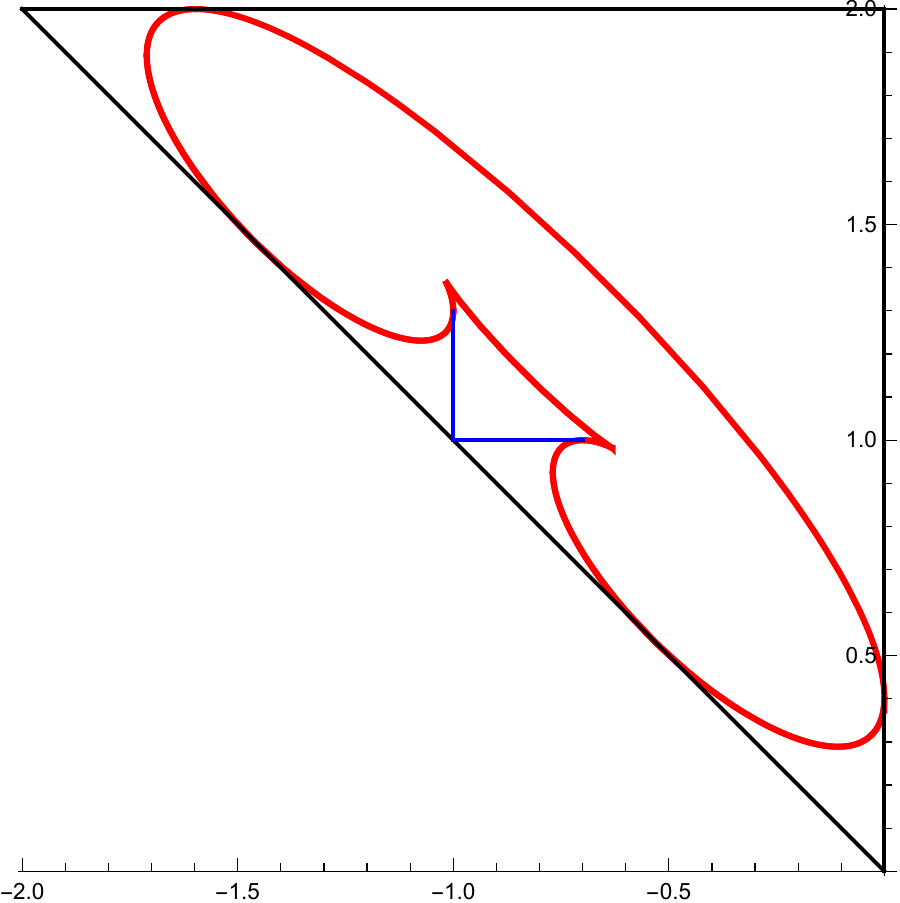}  % second figure itself
        %\caption{second figure}
    \end{minipage}
     \begin{minipage}{0.33\textwidth}
        \centering
        \includegraphics[width=4.1cm]{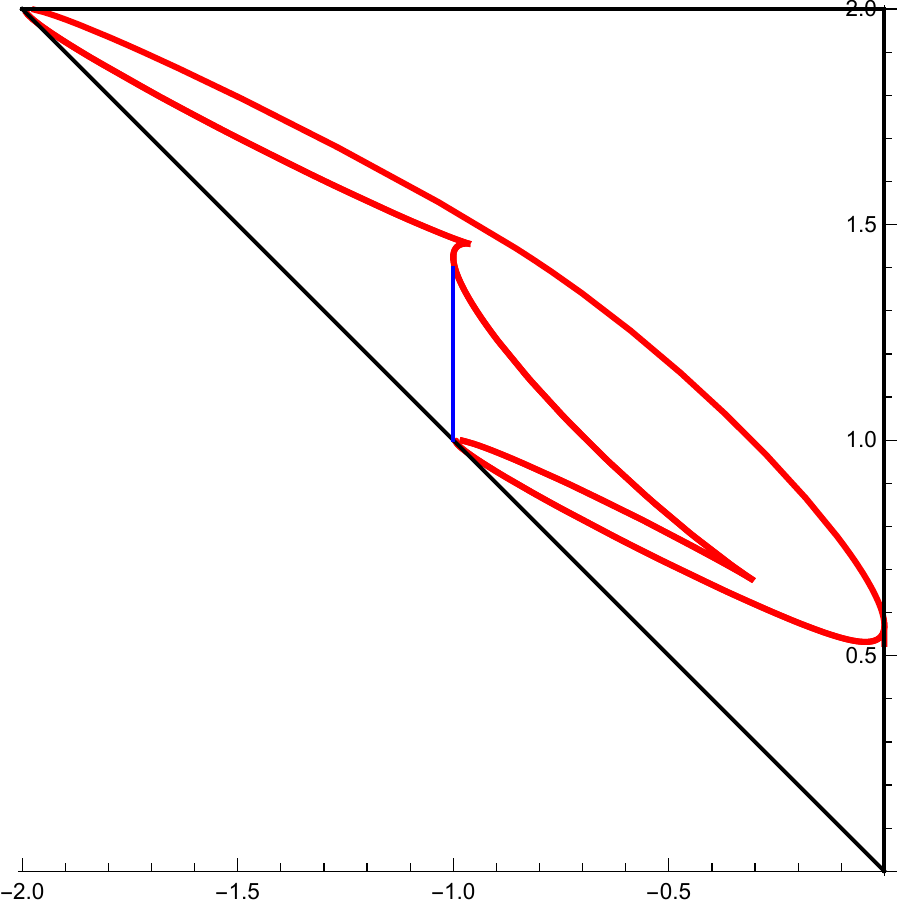}  % second figure itself
        %\caption{second figure}
    \end{minipage}\\
\begin{minipage}{0.5\textwidth}
        \centering
        \includegraphics[width=4.1cm]{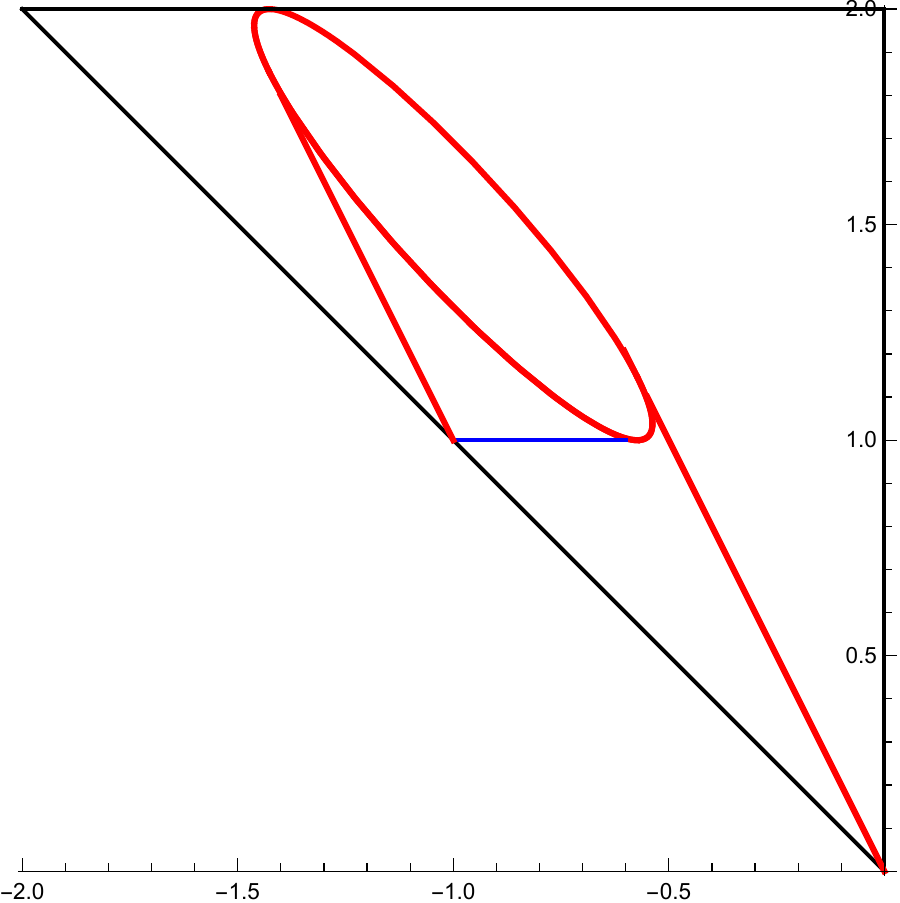} % first figure itself
        %\caption{first figure}
    \end{minipage}\hfill
    \begin{minipage}{0.5\textwidth}
        \centering
        \includegraphics[width=4.1cm]{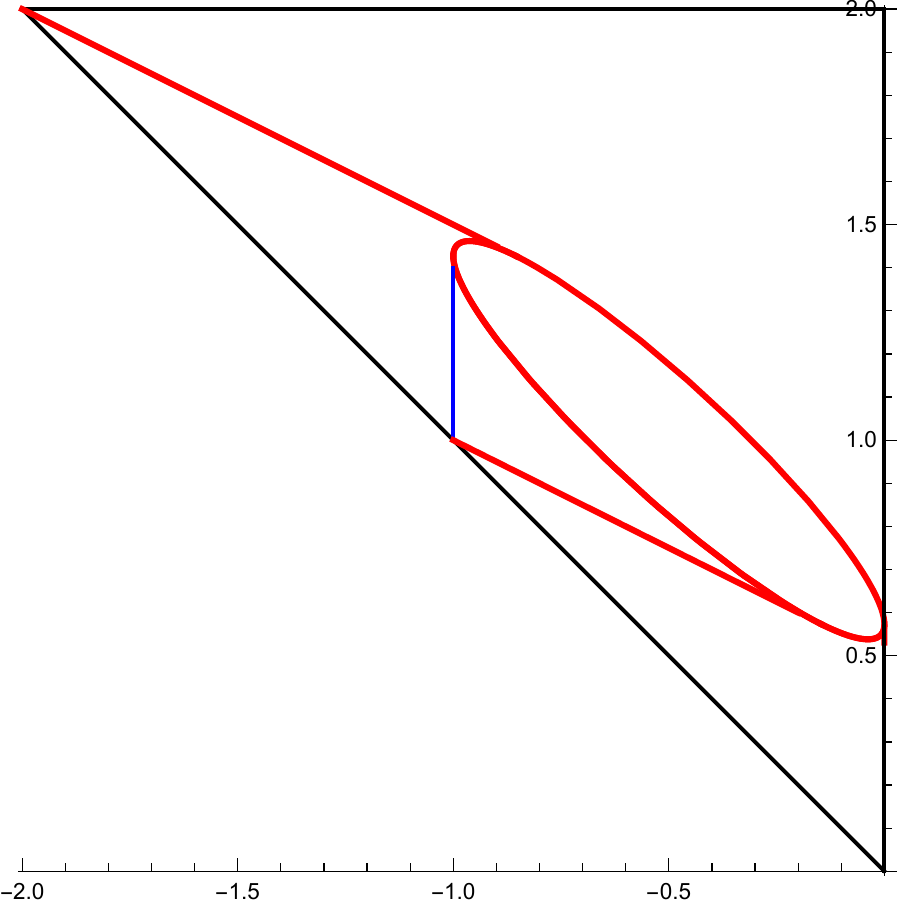}  % second figure itself
        %\caption{second figure}
 \end{minipage}       
\end{center}
\caption{\small Arctic curves of the 
free fermion ($\eta=\pi/4,\lambda=\pi/8$) 20V${}_3$ model. We represent values $\mu=-\frac{\pi}{8}+.01, 0, \frac{\pi}{8}-.01$ on the top row, and the two limiting cases $\mu=-\frac{\pi}{8},+\frac{\pi}{8}$ on the bottom row. }
\label{fig:mu}
\end{figure}

We first illustrate in Fig. \ref{fig:lam}
the effect of a variation of $\lambda$ from $-\frac{\pi}{4}$ to $+\frac{\pi}{4}$ for $\mu=0$. The straight blue lines originating from the point $(-1,1)$ delimit different crystalline phases, as explained in Sect. \ref{frosec}. Note that the
two extreme points $\lambda=\pm \frac{\pi}{4}$ must be computed as limits. 
More precisely, we find that
$$\lim_{\epsilon\to 0^+} (X_{NE}[\xi \sqrt{\epsilon}],Y_{NE}[\xi\sqrt{\epsilon}])\vert_{\lambda=-\frac{\pi}{4}+\epsilon}=
(-\frac{4\xi^2}{1+4\xi^2},2 -\frac{4\xi^2}{1+4\xi^2})$$
which gives the segment $(-1,1)-(0,2)$.
Likewise, we have
$$\lim_{\epsilon\to 0^+} (X_{NE}[\xi],Y_{NE}[\xi])\vert_{\lambda=\frac{\pi}{4}-\epsilon}=\frac{1}{4}(2\cos(2\xi)-3,-2\cos(2\xi)+5)$$
which gives the segment $(-\frac{1}{4},-\frac34)-(-\frac54, \frac74)$.
We also have to include the limiting ellipses:
\begin{eqnarray*}
&&\lim_{\epsilon\to 0^+} (X_{NE}[\xi\epsilon],Y_{NE}[\xi\epsilon])\vert_{\lambda=\frac{\pi}{4}-\epsilon}=\left(\frac{2\xi^2}{8\xi(1-\xi)-3},-\frac{6\xi^2-4\xi+1}{8\xi(1-\xi)-3}\right)\\
&&\lim_{\epsilon\to 0^+} (X_{NE}[\frac{\pi}{2}-\xi\epsilon],Y_{NE}[\frac{\pi}{2}-\xi\epsilon])\vert_{\lambda=\frac{\pi}{4}-\epsilon}=\left(-\frac{10\xi^2+8\xi+3}{8\xi(1+\xi)+3},\frac{14\xi^2+12\xi+4}{8\xi(1+\xi)+3}\right)
\end{eqnarray*}

Next we show in Fig. \ref{fig:mu} the effect of a variation of $\mu$ from $-\frac{\pi}{8}$ to $+\frac{\pi}{8}$
for $\eta=\frac{\pi}{4},\lambda=\frac{\pi}{8}$. The two limiting cases are made of an ellipse and two tangent segments.

%\begin{figure}
%\begin{center}
%\includegraphics[width=8.cm]{phases.pdf}
%\end{center}
%\caption{\small The full arctic curve of the 20V${}_3$ model (thick red line) encompasses the liquid phase ($L$) and is represented together with the five 
%frozen phases $F_i$, $i=0,1,2,3,4,5$. }
%\label{fig:phasesfree}
%\end{figure}

\subsection{Example III: Arbitrary weights}

\begin{figure}
\begin{center}
\begin{minipage}{0.33\textwidth}
        \centering
        \includegraphics[width=4.1cm]{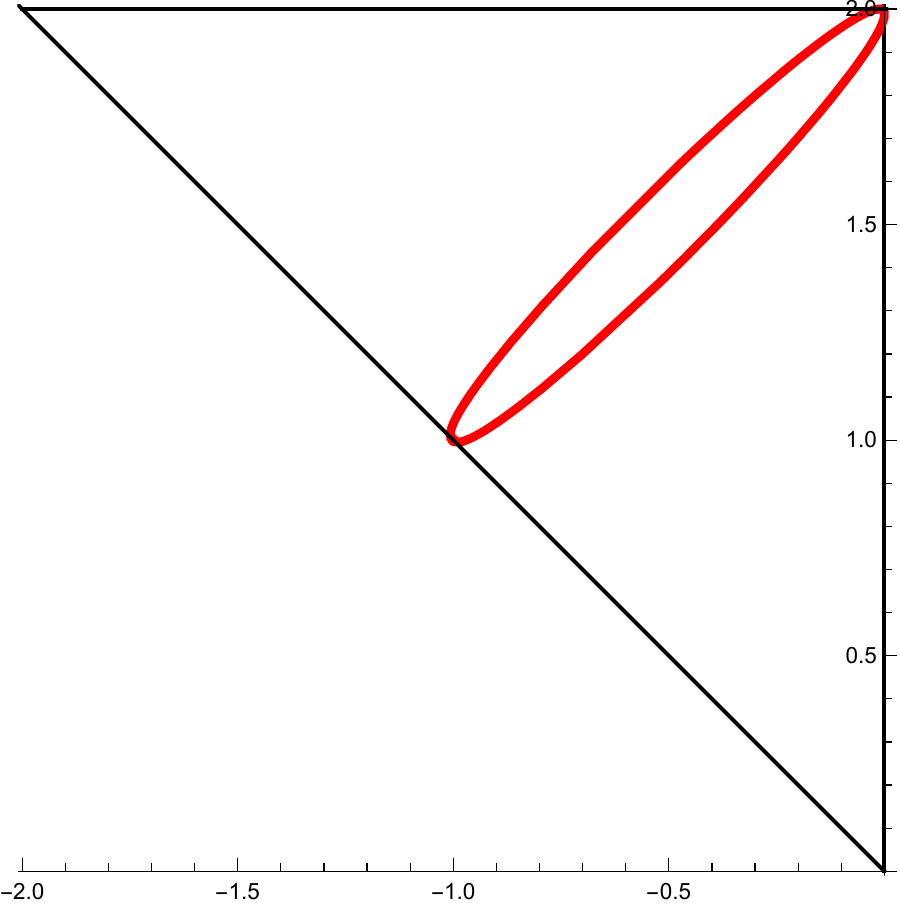} % first figure itself
        %\caption{first figure}
    \end{minipage}\hfill
    \begin{minipage}{0.33\textwidth}
        \centering
        \includegraphics[width=4.1cm]{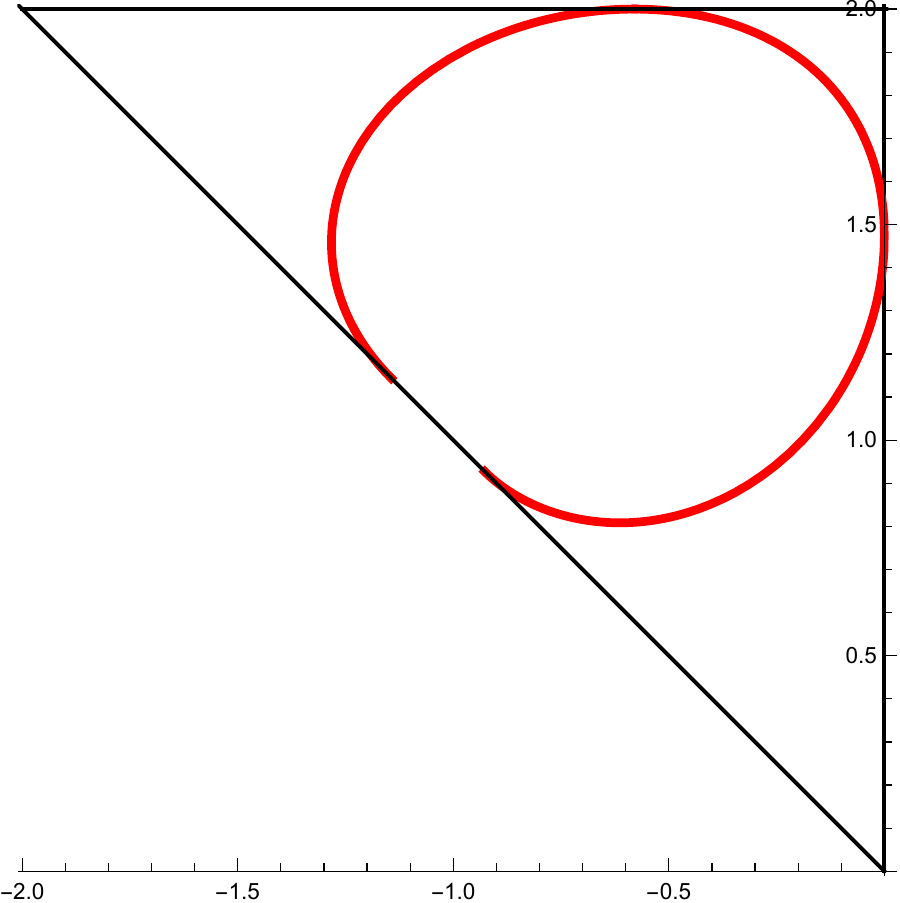}  % second figure itself
        %\caption{second figure}
    \end{minipage}
     \begin{minipage}{0.33\textwidth}
        \centering
        \includegraphics[width=4.1cm]{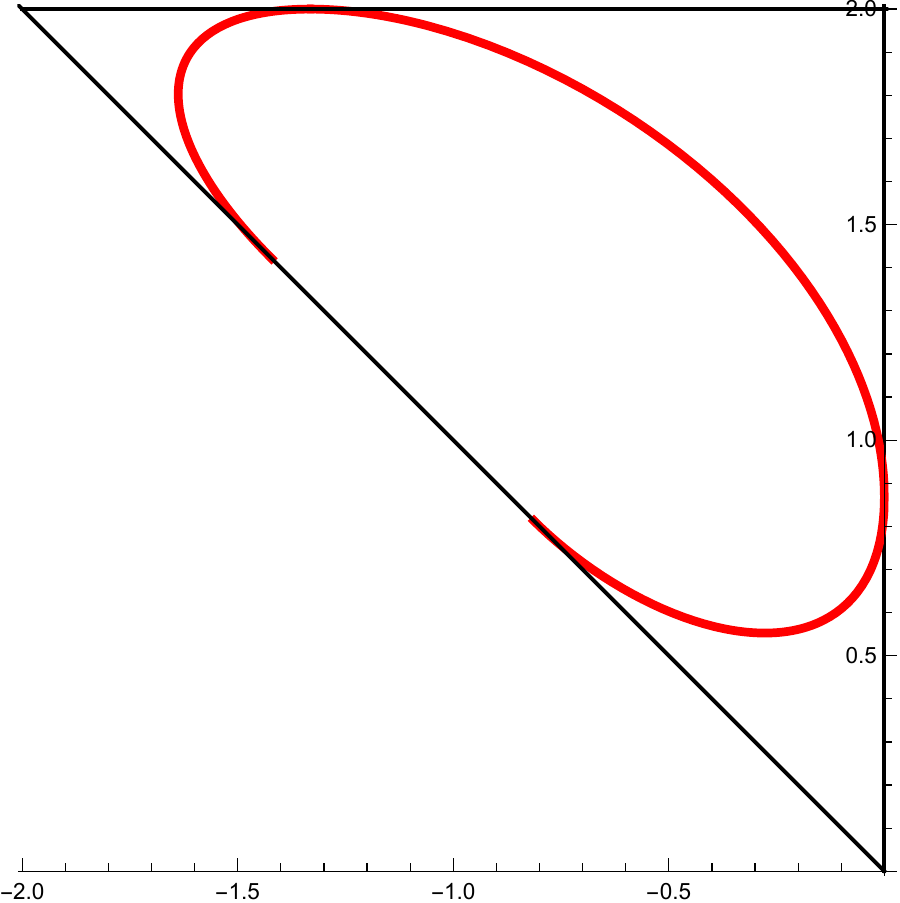}  % second figure itself
        %\caption{second figure}
    \end{minipage}\\
\begin{minipage}{0.33\textwidth}
        \centering
        \includegraphics[width=4.1cm]{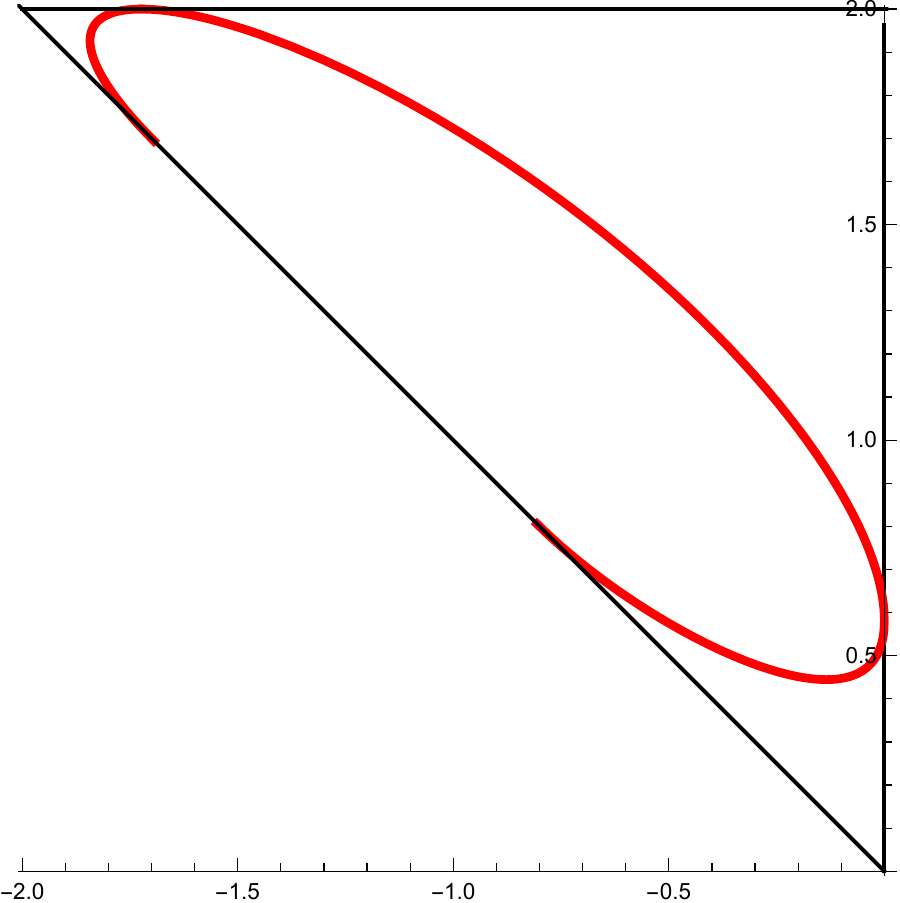} % first figure itself
        %\caption{first figure}
    \end{minipage}\hfill
    \begin{minipage}{0.33\textwidth}
        \centering
        \includegraphics[width=4.1cm]{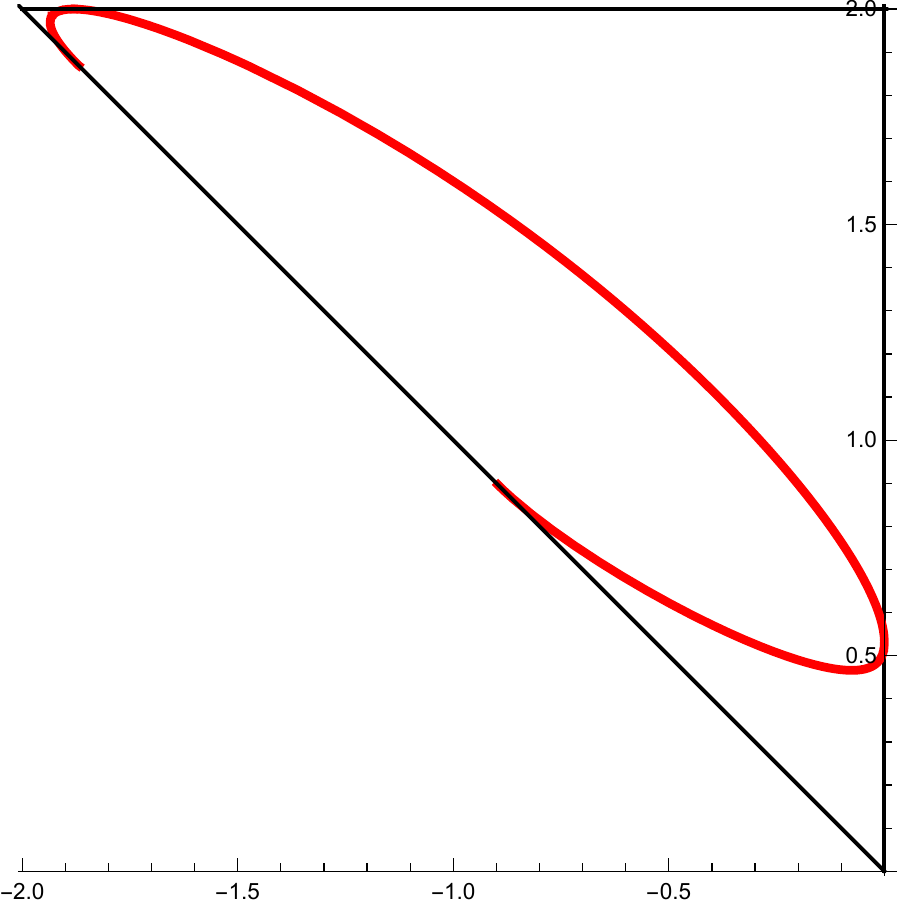}  % second figure itself
        %\caption{second figure}
 \end{minipage}     
  \begin{minipage}{0.33\textwidth}
        \centering
        \includegraphics[width=4.1cm]{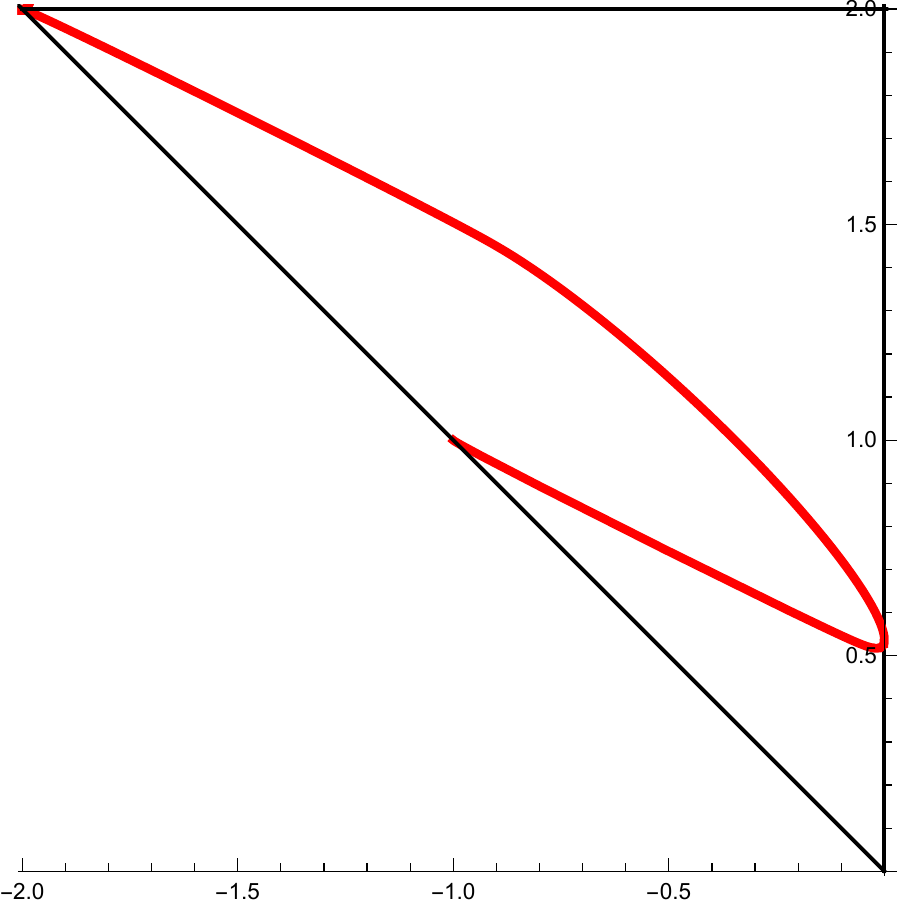}  % second figure itself
        %\caption{second figure}
 \end{minipage}       
\end{center}
\caption{\small Arctic curves of the 20V${}_3$ model for $\eta=\pi/6,\mu=\pi/8$. We represent values $\lambda=\frac{5\pi}{6}-.001, \frac{3\pi}{4} \frac{7\pi}{12}, \frac{5\pi}{12},\frac{\pi}{3},\frac{7\pi}{24}+.01$. }
\label{fig:gen}
\end{figure}

We finally present in Fig. \ref{fig:gen} the NE,SE and NW portions of arctic curve for $\eta=\frac{\pi}{6},\mu=\frac{\pi}{8}$
and $\lambda$ varying from $\mu+\eta=\frac{7\pi}{24}$ to $\pi-\eta=\frac{5\pi}{6}$.

\section{Discussion/conclusion}\label{discsec}

In this paper we have derived exact predictions for the arctic curves of the 20V${}_3$ model on a triangle with suitable DWBC, using the tangent method. This was done in the same spirit as \cite{BDFG} by using the integrability of the lattice model which allows to relate various asymptotic enumerations of the 20V${}_3$ model to those of the 6V model. Our results include the complete ``outer" arctic curves for the full range of parameters describing integrable weights. We also offered a peek into the inner structure of the arctic curves in the (analytic) case of free fermions.

\begin{figure}
\begin{center}
\includegraphics[width=8.cm]{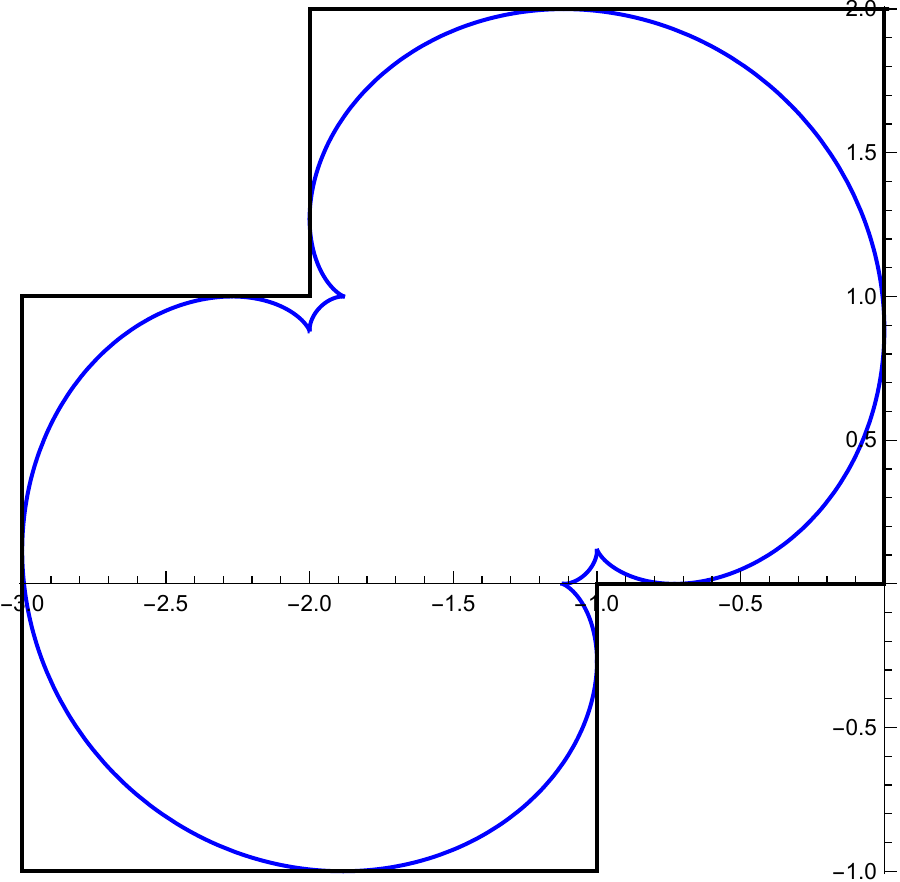}
\end{center}
\caption{\small The algebraic curve describing the NE branch of the 20V${}_3$ model (in blue) is inscribed in a domain made of two overlapping squares of size 2. }
\label{fig:domino}
\end{figure}

We found in particular that the NE portion of the arctic curve for the uniform case is algebraic of degree 10, and that the SE and NW portions are simple shears applied to other portions of this algebraic curve. 
In Ref.~ \cite{BDFG}, the NE part of the arctic curve of the same 20V model for another (square) domain with domain-wall type boundary conditions
(called DWBC1,2) was also found to be a degree 10 algebraic curve, with a similar ``shear" phenomenon occurring for other parts.  Interestingly, the algebraic curve appears to be the arctic curve for different objects: the quarter-turn symmetric domino tilings of a holey Aztec square domain. In view of this we have represented in Fig.~\ref{fig:domino}
the algebraic curve of Sect.~\ref{unisec}, together with a domain made of two overlapping squares in which it is inscribed. We conjecture that there should exist a domino tiling problem on a domain whose shape approximates
the two overlapping squares (possibly with some holes of small size in its interior), and for which the arctic curve is our degree 10 algebraic curve. One could even hope for an identity between the numbers of configurations of the 20V${}_3$ model and of such a tiling problem, possibly with a half-turn rotational symmetry.

As to the numbers of configurations themselves, we showed that for even/odd size triangles: 
$Z_{2n}^{20V_3}=2^{n(n+1)/2}Z_n^{6V}$ and $Z_{2n-1}^{20V_3}=2^{n(n-1)/2}Z_n^{6V}$, where $Z_n^{6V}$
counts the number of quarter turn symmetric domino tilings of the holey Aztec square of size $n$. This mysterious relation calls for some combinatorial explanation. In particular the presence of overall powers of $2$ may point towards tilings of symmetric domains (see \cite{CIUCUsym}), unless it is simply a factor by the number of domino tilings of some Aztec diamond (equal to $2^{n(n+1)/2}$ for a diamond of size $n$).

\bibliographystyle{amsalpha} %Can also try "amsplain", "plain", and "alpha" styles

\bibliography{ArcticTriangle}
\end{document}